\documentclass[prx,epsf,showpacs,twocolumn,superscriptaddress,footinbib]{revtex4-2}

\usepackage[pdftex]{graphicx}
\usepackage{dcolumn}
\usepackage{bm}
\usepackage{epsfig}
\usepackage{latexsym}
\usepackage{amsmath}
\usepackage{amssymb}
\usepackage{amsfonts}
\usepackage{braket}
\usepackage[position=top,caption=false,singlelinecheck=off, justification=raggedright]{subfig}
\usepackage{color}
\usepackage{comment}
\usepackage{array}
\usepackage{braket}
\usepackage{bbold}
\usepackage{dsfont}
\usepackage{mathrsfs}
\usepackage{stmaryrd}
\usepackage{tikz}
\usetikzlibrary{quantikz}
\usepackage{hyperref}
\hypersetup{
    colorlinks,
    citecolor=blue,
    filecolor=red,
    linkcolor=magenta,
    urlcolor=blue
}

\usepackage{algpseudocode}

\newcommand{\nkd}[3]{\left \llbracket #1,\, #2,\, #3 \right \rrbracket}

\renewcommand{\>}{\rangle}

\renewcommand{\[}{\left[}

\usepackage[normalem]{ulem}

\usepackage[dvipsnames]{xcolor}
\usepackage{hyperref}
\hypersetup{
  colorlinks   = true,
  urlcolor     = blue,
  linkcolor    = blue,
  citecolor   = PineGreen
}

\DeclareMathOperator{\tr}{tr}

\newcommand{\logX}{\overline{X}}

\newcommand{\logZ}{\overline{Z}}
\newcommand{\I}[1]{\mathcal{I}_{#1}}
\newcommand{\CX}{\text{CX}}

\def\cambridge{Quantinuum, Terrington House, Cambridge, CB2 1NL, UK}
\def\london{Quantinuum, Partnership House, Carlisle Place, London, SW1P 1BX, UK}
\def\broomfield{Quantinuum, Broomfield, CO 80021, USA}
\newcommand{\brooklynpark}{Quantinuum, Brooklyn Park, MN 55422, USA}
\def\tokyo{Quantinuum K.K., Chiyoda-ku, Tokyo, Japan}
\usepackage{graphicx}
\usepackage{dcolumn}
\usepackage{tabularx}
\usepackage{xcolor}
\usepackage{bm}
\usepackage{amsthm}
\newtheorem{proposition}{Proposition}
\newtheorem{definition}{Definition}
\newcommand{\Tr}{\text{Tr}\,}

\begin{document}

\title{Computing with many encoded logical qubits beyond break-even}

\author{Shival Dasu}
\thanks{These authors contributed equally.}
\affiliation{\broomfield}
\author{Matthew DeCross}
\thanks{These authors contributed equally.}
\affiliation{\broomfield}

\author{Andrew Y. Guo}
\affiliation{\broomfield}
\author{Ali Lavasani}
\affiliation{\broomfield}

\author{Jan Behrends}
\affiliation{\cambridge}
\author{Asmae Benhemou}
\affiliation{\london}
\author{Yi-Hsiang Chen}
\affiliation{\broomfield}
\author{Karl Mayer}
\affiliation{\broomfield}
\author{Chris N. Self}
\affiliation{\cambridge}
\author{Selwyn Simsek}
\affiliation{\london}
\author{Basudha Srivastava}
\affiliation{\cambridge}

\author{M.S. Allman}
\affiliation{\broomfield}
\author{Jake Arkinstall}
\affiliation{\cambridge}
\author{Justin G. Bohnet}
\affiliation{\broomfield}
\author{Nathaniel Q. Burdick}
\affiliation{\brooklynpark}
\author{J.P. Campora III}
\affiliation{\broomfield}
\author{Alex Chernoguzov}
\affiliation{\broomfield}
\author{Samuel F. Cooper}
\affiliation{\broomfield}
\author{Robert D. Delaney}
\affiliation{\broomfield}
\author{Joan M. Dreiling}
\affiliation{\broomfield}
\author{Brian Estey}
\affiliation{\broomfield}
\author{Caroline Figgatt}
\affiliation{\broomfield}
\author{Cameron Foltz}
\affiliation{\broomfield}
\author{John P. Gaebler}
\affiliation{\broomfield}
\author{Alex Hall}
\affiliation{\broomfield}
\author{Craig A. Holliman}
\affiliation{\tokyo}
\author{Ali A. Husain}
\affiliation{\brooklynpark}
\author{Akhil Isanaka}
\affiliation{\broomfield}
\author{Colin J. Kennedy}
\affiliation{\broomfield}
\author{Yuga Kodama}
\affiliation{\tokyo}
\author{Nikhil Kotibhaskar}
\affiliation{\london}
\author{Nathan K. Lysne}
\affiliation{\tokyo}
\author{Ivaylo S. Madjarov}
\affiliation{\broomfield}
\author{Michael Mills}
\affiliation{\broomfield}
\author{Alistair R. Milne}
\affiliation{\london}
\author{Brian Neyenhuis}
\affiliation{\broomfield}
\author{Annie J. Park}
\affiliation{\broomfield}
\author{Anthony Ransford}
\affiliation{\broomfield}
\author{Adam P. Reed}
\affiliation{\broomfield}
\author{Steven J. Sanders}
\affiliation{\broomfield}
\author{Charles H. Baldwin}
\affiliation{\broomfield}
\author{David Hayes}
\affiliation{\broomfield}
\author{Ben Criger}
\affiliation{\cambridge}

\author{Andrew C. Potter}
\affiliation{\broomfield}
\author{David Amaro}
\affiliation{\london}

\date{\today}

\begin{abstract}
High-rate quantum error correcting (QEC) codes encode many logical qubits in a given number of physical qubits, making them promising candidates for quantum computation.
Implementing high-rate codes at a scale that both frustrates classical computing and improves performance by encoding requires both high fidelity gates and long-range qubit connectivity---both of which are offered by trapped-ion quantum computers.
Here, we demonstrate computations that outperform their unencoded counterparts in the high-rate $\nkd{k+2}{k}{2}$ iceberg quantum error detecting (QED) and $\nkd{(k_2 + 2)(k_1 + 2)}{k_2k_1}{4}$ two-level concatenated iceberg QEC codes, using the 98-qubit Quantinuum Helios trapped-ion quantum processor.
Utilizing new gadgets for encoded operations, we realize this ``beyond break-even" performance with reasonable postselection rates across a range of fault-tolerant (FT) and partially-fault-tolerant (pFT) component and application benchmarks with between $48$ and $94$ logical qubits.
These benchmarks include FT state preparation and measurement, QEC cycle benchmarking, logical gate benchmarking, GHZ state preparation, and a pFT quantum simulation of the three-dimensional $XY$ model of quantum magnetism. Additionally, we illustrate that postselection rates can be suppressed by increasing the code distance via concatenation. Our results represent state-of-the-art logical component and state fidelities and provide evidence that high-rate QED/QEC codes are viable on contemporary quantum computers for near-term beyond-classical-scale computation.
\end{abstract}
\maketitle

Achieving high fidelity in lengthy quantum computations necessary for most practical applications requires encoding logical qubits into a larger number of physical qubits in a quantum code to suppress errors.
Computation and error correction with encoded logical qubits adds space (memory) and time (circuit complexity) overheads, that, so far, have limited experimental implementations on existing quantum hardware to modest numbers of logical qubits that do not yet challenge the limits of classical computation.

The arbitrary qubit connectivity available in trapped-ion and neutral atom devices enables a variety of geometrically non-local codes with substantially reduced overheads~\cite{yoder2025tour, Yoshida2025concatenate, berthusen2025simple}. As an extreme example, iceberg quantum error-detection (QED) codes~\cite{gottesmanPHD,chao2018quantum,Self:2024py} require only two additional physical qubits to encode an arbitrary number of logical qubits capable of detecting any single-qubit error.

In this work, we introduce a family of concatenated iceberg codes that generalize many-hypercubes codes~\cite{Goto:2024xd} to variable code block sizes to achieve asymptotically high encoding rates and code distances~\cite{Yoshida2025concatenate}.
Using the 98-qubit Quantinuum Helios trapped-ion quantum processor~\cite{helios} (see Methods~\ref{sec:Helios} for hardware details), we implement iceberg codes at the first and second levels of concatenation, respectively forming QED and quantum error correction (QEC) codes of distances $d=2$ and $4$.

While arbitrary error suppression ultimately requires fully fault-tolerant (FT) operation, there is a significant leap in complexity and scale of computation before FT algorithms can outperform unencoded ones.
To bridge this gap, partially fault-tolerant (pFT) techniques~\cite{Fujii2024star, Fujii2025} that can immediately improve upon unencoded performance and extend the practical reach of quantum computers can be useful, even if these techniques do not scale efficiently to asymptotically large circuits.
In Quantinuum trapped-ion quantum processors, a practical pFT scheme consists in implementing arbitrary-angle logical Pauli rotations via continuously-parameterized native two-qubit Pauli rotations to avoid the sizable cost of FT non-Clifford gate synthesis~\cite{iceberg-beyond-the-tip, Yamamoto2025molecular}. 

This work reports encoded (``logical") vs. unencoded (``physical") performance results across a range of FT and pFT techniques, including experiments with up to 94 logical qubits in iceberg QED codes, and up to 48 logical qubits in concatenated iceberg QEC codes. 

For iceberg QED codes, we introduce improved FT initialization, syndrome extraction, and readout gadgets, as well as new FT and pFT gates, that develop iceberg codes as a practical tool for beyond-classical-scale quantum computation. We show that encoded algorithms can outperform their unencoded counterparts, that is, demonstrate \emph{beyond break-even} performance, on a variety of logical benchmarks including state preparation-and-measurement (SPAM) experiments, cycle benchmarking of logical gates, preparation of GHZ states, and the dynamics of the three-dimensional $XY$ model relevant to simulation of magnetic materials.

For concatenated iceberg QEC codes, we similarly introduce improvements to initialization, QEC, and entangling gadgets, demonstrating beyond break-even performance on SPAM experiments and the preparation of GHZ states. We benchmark the performance of the newly introduced QEC cycle and find below pseudo-threshold performance for the $\nkd{80}{48}{4}$ code. Additionally, we find improved QEC acceptance rates for the $\nkd{16}{4}{4}$ code compared to analogous QED acceptance rates for the $\nkd{4}{2}{2}$ code at comparable fidelity, allowing for longer FT computations to be performed at higher levels of concatenation. Postselecting out detectable, but not correctable, weight-$d/2$ faults in even-distance codes should not be misconstrued as a non-scalable technique, since, below threshold, both the error rate and the postselection overheads are exponentially suppressed with increasing code distance.

\begin{figure*}[t]
\centering
\includegraphics[width=\textwidth]{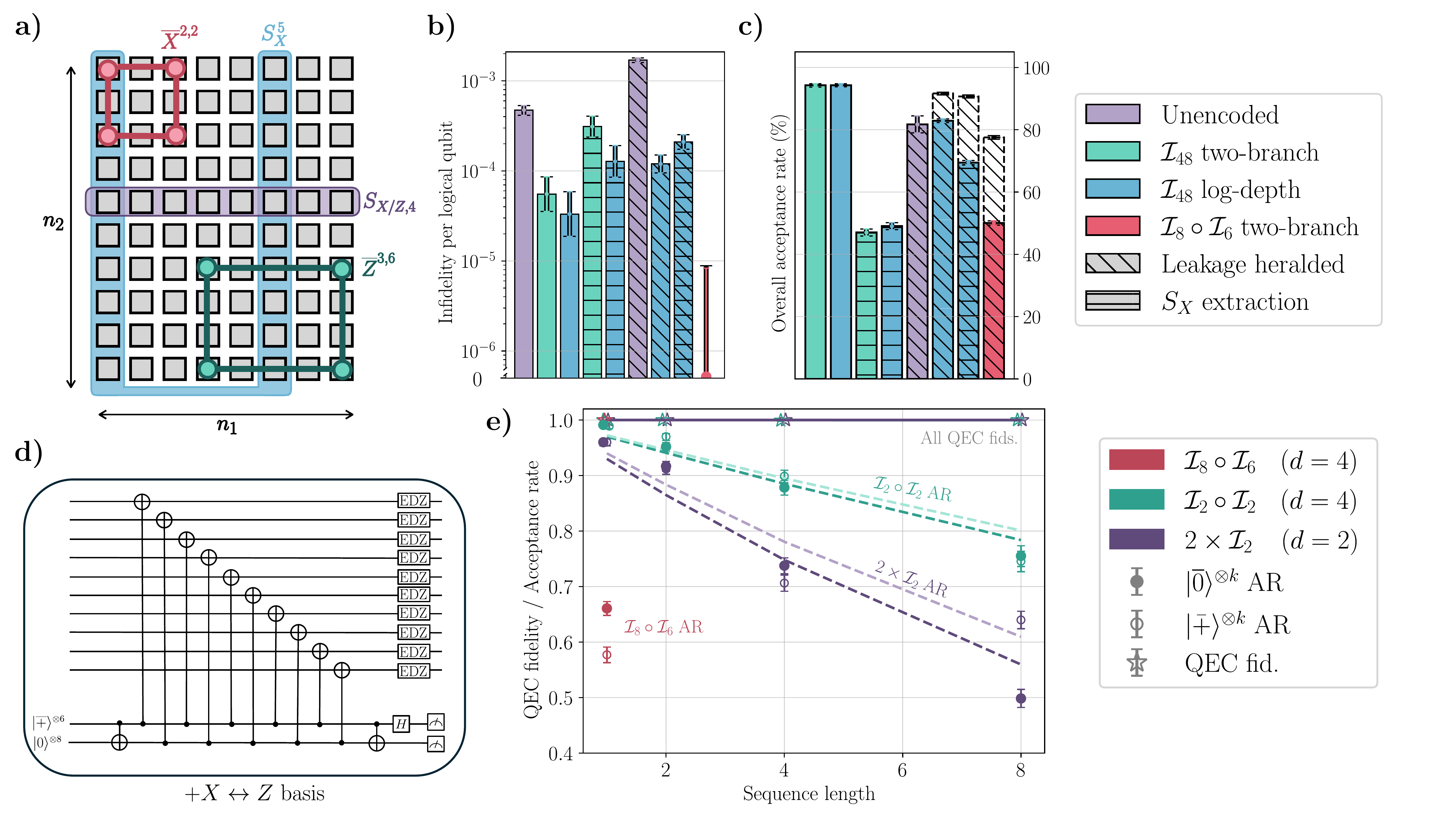}
\caption{
{\bf Iceberg codes and their concatenation}
\textbf{(a)} Grid layout depicting stabilizers and logical operators of the $\I{8}\circ\I{6}$ code. \textbf{(b)} SPAM infidelity per logical qubit across experiments, some using leakage-heralded measurement and/or non-destructive extraction of $S_X$. Note that the $\I{8}\circ\I{6}$ SPAM experiment recorded no logical errors, bounding the infidelity as indicated in the rightmost column. \textbf{(c)} Acceptance rates across each SPAM experiment. When leakage-heralded measurement is used, the leakage acceptance rate is outlined (black, dashed). \textbf{(d)} The QEC cycle for the $\I{8}\circ\I{6}$ code; Z stabilizer syndrome extraction is shown, and is followed by extraction of the $X$ stabilizers. EDZ denotes $d=2$ Bell-pair syndrome extraction. \textbf{(e)} Fidelity and acceptance rate (AR) as a function of number of QEC/SE cycles for (concatenated) iceberg codes for different logical basis input states. All uncertainties reported through this work denote 68\% Wilson confidence intervals unless otherwise noted, with point estimates corresponding to empirical observations. Acceptance and fidelity curves are obtained from fitting the data to exponential decays using maximum likelihood estimation. Fit values are given in Methods~\ref{sec:QEC Benchmark}.
}\label{fig:encoding}
\end{figure*}

\noindent\textbf{Concatenated Iceberg Codes:}
Iceberg codes,  $\I{k}$, are low-overhead quantum error-detecting (QED) codes that encode an even number, $k$, of logical qubits at code distance $d=2$ using only two additional qubits for a total of $n = k + 2$ physical qubits~\cite{gottesmanPHD,chao2018quantum,Self:2024py}. The stabilizer generators $S_X$ and $S_Z$ of an iceberg code are global Pauli products: $S_{X/Z} = \prod_{j=0}^{n-1}(X/Z)_j$. Any Pauli operator containing an even number of $X$'s and $Z$'s is a logical operator. These can be generated by $\logX^j = X_{0}X_{j }$ and $\logZ^j = Z_{j}Z_{k+1}$ for $1 \leq j \leq k$. Here we index physical qubit operators by subscripts with range $[n] \equiv \{0,1,\dots n-1\}$, and logical operators with an overline and superscripts with range $\{1,2,\dots k \}$, and use index-raising to denote encoding qubits into the logical qubits of a code.

Iceberg QED codes can be promoted to QEC codes using concatenation (see, e.g. Ref.~\cite{Goto:2024xd}), which takes as input two iceberg codes, $\I{k_2}$ and $\I{k_1}$, and outputs a code, $\I{k_2} \circ \I{k_1}$, with parameters $\nkd{n_2n_1}{k_2k_1}{d_2d_1}$. This procedure can be visualized as in Fig.~\ref{fig:encoding}a by arranging the physical qubits into an $n_1 \times n_2$ grid with Pauli-$X$ and $Z$ operators labeled $X_{i,j}$ and $Z_{i,j}$ ($i\in [n_1], j\in [n_2]$) respectively. The qubits in each row are encoded into $\I{k_1}$ codes and then $\I{k_2}$ codes are encoded on each of the $k_1$ logical qubit columns. The $X$ and $Z$ stabilizers of the concatenated code consist of Pauli $X$ or $Z$ products along rows (the stabilizers of the $\I{k_1}$ codes) and pairs of columns (the stabilizers of the encoded $\I{k_2}$ codes), and the $\overline{X}$ and $\overline{Z}$ logical operators are rectangles of physical $X$ or $Z$ operators (see Fig.~\ref{fig:encoding}a). This concatenation procedure can be straightforwardly iterated to higher levels to achieve larger distances. 
In the following, we implement only the first two concatenation levels for iceberg codes, exploring a variety of iceberg QED codes with up to $94$ logical qubits and concatenated iceberg QEC codes with up to $48$ logical qubits (the largest for which a full QEC cycle can be executed on Helios).

\noindent\textbf{Encoding and readout:}
Encoding physical qubits into logical qubits is an essential first step for FT computation. To this end, we introduce and benchmark improved encoding circuits and logical readout protocols for iceberg codes with and without concatenation.

The logical zero state of the iceberg code $\I{k}$, $|\overline{0}\rangle^{\otimes k}$, is a GHZ state $|\text{GHZ}_n\rangle$ on all physical qubits \cite{Self:2024py}. In the $d=4$ concatenated iceberg code, the logical zero state $|\overline{0}\rangle^{\otimes k_1 k_2}$ is the tensor product of the logical GHZ states over each $\I{k_2}$ logical qubit column: $|\overline{\text{GHZ}}_{k_2}\rangle^{\otimes k_1}.$ To prepare the $d=2$ logical zero state, a physical GHZ state preparation circuit can be used, which starts with a $|+\rangle$ state and subsequently incorporates more qubits via a tree of CX gates. To prepare the $d=4$ logical zero state, we follow a similar procedure replacing physical qubits with iceberg codes. Starting with a $|\bar{+}\rangle^{\otimes k_1}$ $\I{k_1}$ block, other blocks are incorporated via a tree of transversal logical CX gates.
In both cases, $X$ errors during the circuit can cascade through the CX gates, causing a high-weight error on the output. To ensure fault-tolerance, we catch these errors by measuring $ZZ$ operators at the ends of certain pairs of tree branches \cite{Goto:2024xd}. In the $d=2$ case, we introduce a FT state preparation and measurement (SPAM) circuit using a log-depth tree, which uses up to $n/4$ ancilla qubits and up to $n/2$ additional CX gates. In the $d=4$ case, we introduce a two-branch circuit, with a modified tree structure and additional error-checking to catch additional errors which cannot occur in the $d=2$ case. We use repeat-until-success (RUS), decoherence free subspace (DFS) techniques, and dynamical decoupling (DD) to increase acceptance rates and suppress errors in some $d=2$ encoding and QED gadgets throughout this work (the third only for $XY$ model simulations as well as SPAM experiments extracting $S_X$); for additional details, see Supplementary Information (SI) Sec.~\ref{sec:state_prep_gadgets}. In some SPAM experiments, we only extract one syndrome destructively ($S_Z$), in contrast to the final syndrome measurement proposed in Refs.~\cite{Self:2024py, iceberg-beyond-the-tip}.

We benchmark the performance of fault-tolerantly preparing 48 logical qubits in the $d=2$ and $d=4$ cases in the $\I{48}$ and $\I{8}\circ \I{6}$ codes, respectively (see Fig.~\ref{fig:encoding}b, c, SI.~Sec.~\ref{sec:spam}). Both the $d=2$ case and $d=4$ case achieve beyond-break-even performance compared to the physical SPAM error on Helios \cite{helios} by at least an order of magnitude, with $d=4$ exhibiting improved fidelities and acceptance rates compared to $d=2$, demonstrating the additional error-correction power of the concatenated QEC codes compared to their QED counterparts. In some experiments, we use the leakage-heralded measurement available on Helios (Methods~\ref{sec:Helios}) and compare to the physical leakage-heralded measurement error correspondingly.

\noindent\textbf{Syndrome Extraction and QEC:}
Once encoded, quantum error detection and/or correction are required to prevent physical errors from spreading into logical ones.
Here, we introduce new methods for performing (QEC) QED on (concatenated) iceberg codes, which are optimized to Helios.
For the $d=2$ case, both mid-circuit syndrome extraction (SE) and readout of the logical state are performed
by extending the protocol~\cite{iceberg-beyond-the-tip} to use GHZ states in the $X$ or $Z$ basis to improve parallelization of gates (shown later in Fig.~\ref{fig:computing}b; see also SI Sec.~\ref{sec:ghz_se}). We use four-qubit GHZ states that allow syndrome extraction gates to be parallelized in groups of $8$, matching the number of gate zones in Helios~\cite{helios}. 

For concatenated iceberg codes, previous work on closely-related many-hypercubes codes~\cite{Goto:2024xd, goto2025optimized} proposed using Knill-style QEC that involves preparing two ancillary logical code blocks per QEC cycle, incurring significant spacetime overhead to the QEC cycle. Here, we instead introduce a lower-overhead scheme that generalizes the $d=2$ QEC protocol (with a Bell pair of ancilla) by replacing the physical ancilla qubits in that scheme with lower-concatenation-level $d=2$ iceberg code blocks, followed by regular $d=2$ syndrome extraction on each of the lower-level $\I{6}$ blocks. For the $\I{8}\circ\I{6}$ code, the resulting circuit, shown in Fig.~\ref{fig:encoding}d, uses only 18 physical ancilla qubits, and has depth $18$, representing a substantial reduction in the spacetime overhead compared to Knill-Style QEC. The QEC cycle in this work is similar to the improvement proposed in Ref.~\cite{Nakai_2026}, a follow-up to Ref.~\cite{Goto:2024xd} which appeared during the writing of this manuscript for a subsystem version of concatenated iceberg codes. The physical qubit overhead in Ref.~\cite{Nakai_2026} for the subsystem code would be the same, with the additional advantage of not requiring the lower-level ancilla to be encoded for the extraction of either high- or low- level stabilizers. Conversely, the QEC cycle and decoder introduced here is fault-tolerant for circuit-level noise as discussed in SI Sec.~\ref{sec:QEC}, whereas the QEC cycle in Ref.~\cite{Nakai_2026} was introduced only for phenomenological noise.

Implementing this procedure for the $\I{8}\circ \I{6}$ concatenated iceberg code with 48 logical qubits, we observe an average QEC-cycle infidelity per logical qubit (see Methods~\ref{sec:QEC Benchmark}) upper-bounded by $4\times 10^{-5}$ with no logical errors recorded in five thousand shots in each basis, and an acceptance rate (AR) of $0.62(2)$ when postselecting out uncorrectable errors. Because leakage-heralded measure as well as DD were used in all logical memory experiments, a natural point of comparison for this encoded logical memory experiment is the (non-leakage component of the) physical linear memory error accrued over a depth-1 time: $1.1(3)\times 10^{-4}$ (\cite{helios}, Methods~\ref{sec:QEC Benchmark}).
By this comparison, QEC improves the memory error over that for unencoded physical qubits, indicating that the hardware is below the pseudo-threshold for this $d=4$ concatenated iceberg code.

Next, we examine the effect of increasing the code distance on the logical memory performance. To this end, we shift our focus to smaller (concatenated) iceberg codes with 4 logical qubits: $\I{2}\circ \I{2}$ or two copies of $\I{2}$. 
The higher AR per QEC cycle of these smaller codes allow us to explore sequences with many repeated rounds of QEC (Fig.~\ref{fig:encoding}e) without a prohibitively high shot overhead.
We observe that increasing the distance from 2 to 4 by concatenation significantly boosts the AR (averaged over both initial states) from $0.945(2)$ to $0.972_{-2}^{+1}$.
The infidelities, $< 4\times 10^{-5}$ and $2^{+8}_{-2}\times10^{-5}$ for distances $2$ and $4$ respectively, are indistinguishable up to statistical uncertainty.
As the discard rate is several
orders of magnitude higher than the logical error rate, this demonstrates that the length of a FT computation can be scaled up by increasing the code distance through concatenation.

\begin{figure*}[!t]
\centering
\includegraphics[width=\textwidth]{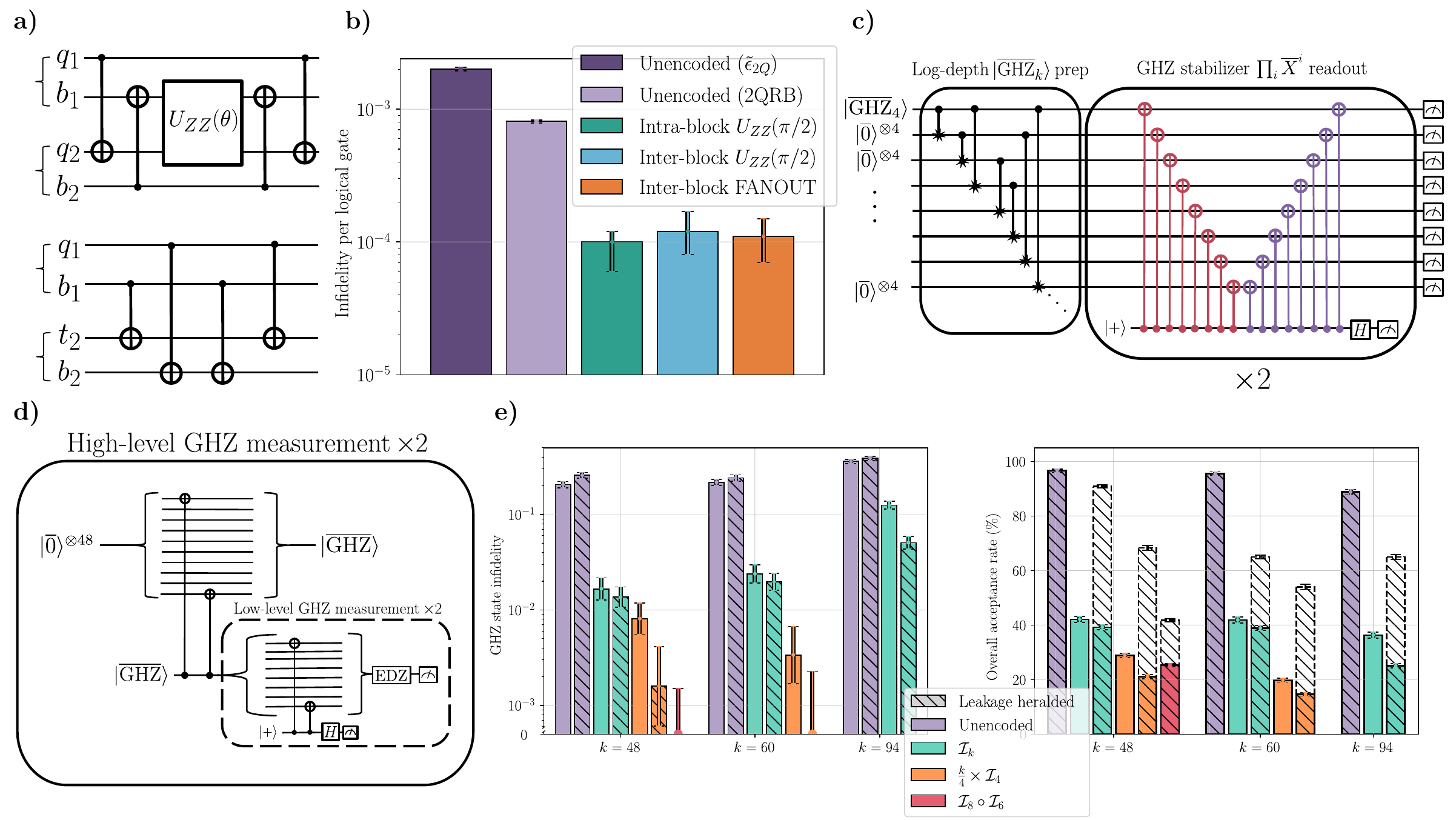}
\caption{
{\bf Logical gate benchmarking and GHZ state preparation}
\textbf{(a)} Circuit gadgets for an inter-block $U_{ZZ}$ gate and the FANOUT gate. Here $t = q_0$ is the ``top" qubit and $b = q_{n-1}$ is the ``bottom" qubit on each code block with indices labeling the code block. \textbf{(b)} Infidelity per logical gate for each gate considered in the main text, reported alongside the physical $U_{ZZ}(\pi/2)$ fidelity from \cite{helios}, both including ($\tilde{\epsilon}_{2Q}$) and excluding (2QRB) transport overheads. \textbf{(c)} Preparing a GHZ state distributed over many $\I{4}$ code blocks using the log-depth FANOUT tree, represented in shorthand by controlled gates with stars, followed by extraction of the $\prod_i \overline{X}^i$ GHZ stabilizer where red (purple) gates couple to qubit $q_0$ ($q_5$) alone on each block. \textbf{(d)} GHZ state preparation in the $\I{8}\circ\I{6}$ code as described in the text. \textbf{(e)} GHZ infidelity and acceptances across the various experiments. Note that the $\I{8}\circ\I{6}$ ($k=48$) and leakage-heralded $15\times \I{4}$ ($k=60$) experiments both recorded no errors.}
\label{fig:entangling}
\end{figure*}

\noindent\textbf{Benchmarking a universal logical gate set:}
Unlike physical benchmarking, logical gate benchmarking must incorporate the effects of syndrome extraction, state preparation, and readout at the encoded level. 
Here, we perform logical gate benchmarking to estimate the average fidelity of a universal logical gate set that entangles logical qubits within or between $\I{4}$ code block(s), demonstrating break-even performance over physical gates.
We employ a logical version of cycle benchmarking (CB)~\cite{Combes2017, Erhard2019} (see Methods~\ref{sec:methods_lcb}, SI Sec.~\ref{sec:si_lgb}), which involves interleaving repeated layers of a chosen multi-qubit Clifford operation (``cycle") with random Pauli gates. Fitting the (exponential) decay of the expectation value of a randomly sampled Pauli observable vs. depth---after postselecting via syndrome extraction rounds interspersed every 8 logical gates (chosen to roughly optimize gate performance)---extracts an estimate of the Pauli fidelity, which can be used to compute the fidelity of the logical gate.

Using this approach, we benchmark the following universal pFT gate set (Fig.~\ref{fig:entangling}a, see also SI Sec.~\ref{sec:gates}) consisting of: non-FT intra- and inter-code block $U_{ZZ}(\theta) = e^{-i\frac{\theta}{2}\overline{ZZ}}$ \footnote{We benchmark the gate at $\theta = \pi/2$; in general Helios can implement either gate with arbitrary $\theta$ and we expect that logical error decreases with $\theta$ similar to the physical gate, since the probability of a single fault causing a logical error is correspondingly scaled down.}, as well as a FT logical FANOUT gate.
The latter implements a global logical parity flip $\prod_i \logX^i$ on all logical qubits in one code block, controlled by a logical qubit in another code block. 
The resulting average logical gate infidelities, all between $(1.0 - 1.2) \times 10^{-4}$, are summarized in Fig.~\ref{fig:entangling}b.
Crucially, these logical fidelities significantly outperform the bare physical two-qubit gate errors of $\sim 8\times 10^{-4}$ on Helios~\cite{helios}, not to mention the implied two-qubit gate fidelities $\tilde{\epsilon}_{2Q} \sim 2\times 10^{-3}$ obtained from full-circuit benchmarks that include other substantial overheads.

\noindent\textbf{Preparing large logical entangled states:}
We now benchmark the high-fidelity preparation of entangled resource states of many logical qubits, specifically, GHZ states, in various sizes and combinations of iceberg codes.
Details of the unencoded GHZ state fidelity protocol can be found in SI Sec.~\ref{sec:physical_ghz}.
The logical GHZ state for a single $\I{k}$ code block factorizes into a physical GHZ state on physical qubits $1,\dots k$ and a Bell pair on qubits $0,k+1$, which can be directly prepared by a small modification of the usual encoding circuit (SI Sec.~\ref{sec:ghz_prep}).
To prepare logical GHZ states spanning multiple iceberg code blocks, we initialize a single code block in a logical GHZ state, and apply a log-depth tree of logical FANOUT gates (Fig.~\ref{fig:entangling}c) to extend the GHZ state across other code blocks that were initialized to $\ket{\bar{0}}^{\otimes k}$. This approach saves many physical 2Q gates compared to the standard implementation using transversal CX gates.
To benchmark the resulting GHZ state fidelity in both cases, we perform simultaneous (single-shot) FT measurements of the logical GHZ stabilizers ($\bar{Z}^i\bar{Z}^{i+1}$ and $\prod_i \bar{X}^i$)~(see Methods~\ref{sec:methods_ghz}).

To explore the effects of code concatenation, we develop a FT GHZ state preparation protocol on a concatenated iceberg code block that generalizes more straightforwardly to higher concatenation levels. The protocol (Fig.~\ref{fig:entangling}d) uses level-1 $\I{k_1}$ ancilla prepared in a logical GHZ state as described above, to fault-tolerantly measure the level-2 logical GHZ stabilizer $\prod_{i,j}\logX^{i,j}$ (defined in Methods~\ref{sec:methods_ghz}).
To ensure robustness to up to two measurement errors the logical GHZ stabilizer measurement is performed twice.
Additionally, the $X, Z$ stabilizers and $\logZ$ observables of the level-1 GHZ state are extracted, which allows us to correct one error occurring anywhere in the full $d=4$ protocol (see Methods~\ref{sec:methods_ghz} for further details). SI Sec.~\ref{sec:ghz_prep_d4} contains simulations of fault tolerance and a more detailed analysis. 

Results for logical GHZ state (in)fidelity and acceptance rates are displayed in Fig.~\ref{fig:entangling}e and reported in SI Sec.~\ref{sec:ghz_prep} for $\I{48}$, $\I{60}$, and $\I{94}$ codes, the $\I{8} \circ \I{6}$ code, and for $k=48,\,60$ logical qubits distributed among $\I{4}$ code blocks. Across the board, encoding yields large improvements to logical GHZ state fidelity with greater gains coming from the improved fault-tolerance of the $\I{4}$ code blocks and the $\I{8} \circ \I{6}$ code. In particular, we measure $\sim\! 95\%$ fidelity in the $\I{94}$ code and recorded no logical errors in the $\I{8} \circ \I{6}$ and leakage-heralded $k=60$, $\I{4}$ codes over several thousand shots.

\begin{figure*}[!t]
\centering
\includegraphics[width=\textwidth]{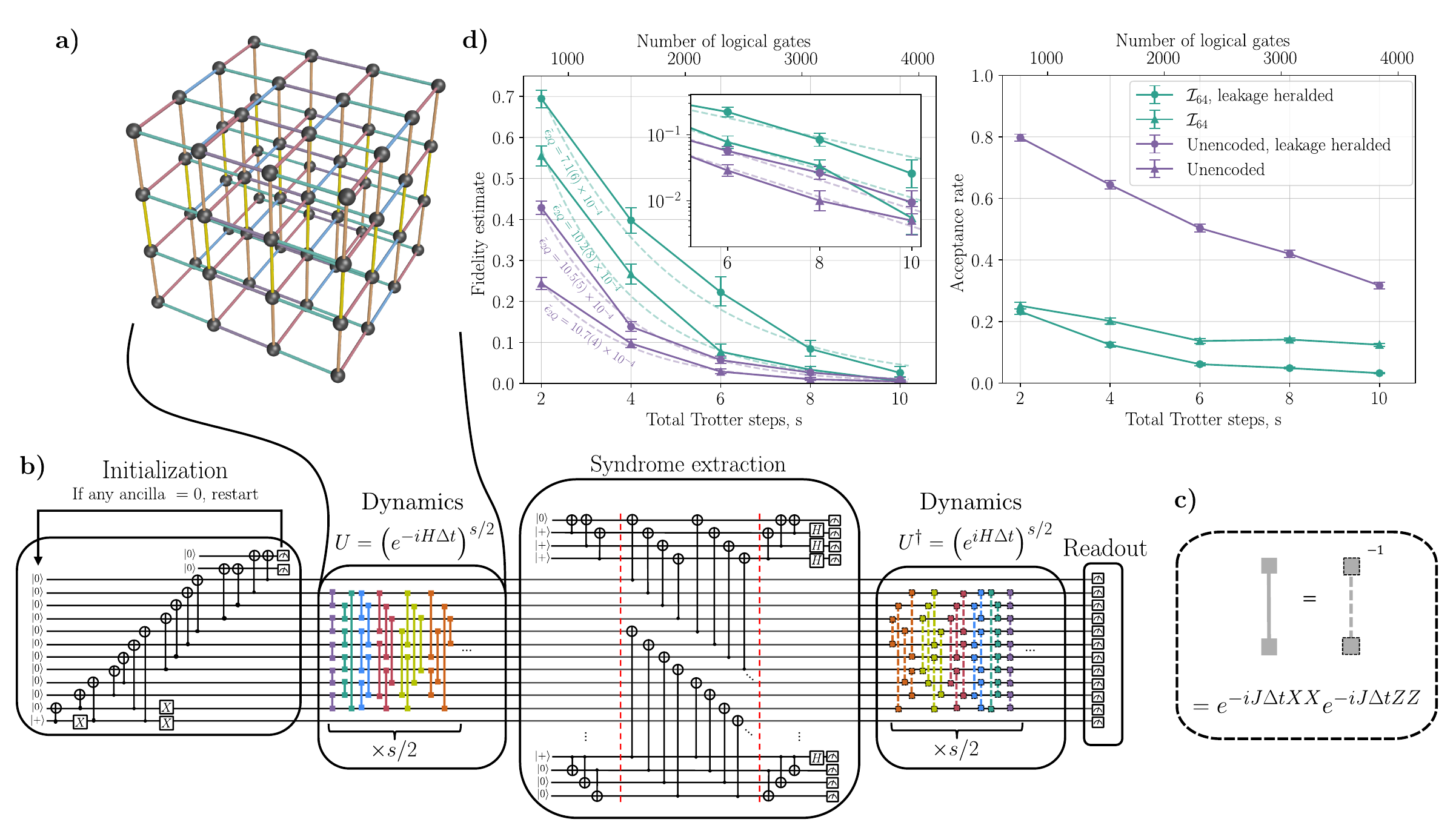}
\caption{
{\bf Mirror benchmarking of encoded Hamiltonian simulation}
\textbf{(a)} The natural edge-$6$-coloring of the $4\times 4\times 4$ periodic cubic lattice. Edges connecting the periodic boundaries are omitted for visual clarity, but the corresponding gates are included in the circuits executed. \textbf{(b)} Logical mirror benchmarking circuits used to assess the fidelity of Hamiltonian simulation. For the purposes of illustration, the circuit shown uses an $\I{10}$ code with dynamics governed by a 6-regular graph on 10 vertices. \textbf{(c)} Logical gates implementing the first-order Trotter circuit depicted in (b). In the ``reverse" half of the circuit, the inverse gates (dashed) are used in the opposite order. \textbf{(d)} Fidelity estimate and acceptance rate as a function of total Trotter steps in the full mirrored circuit. The fidelity estimate inset displays the last three points on a log scale for clarity. Each heuristic depolarizing error model (dashed) is labeled with its corresponding effective 2Q error rate $\tilde{\epsilon}_{2Q}$ (see Methods~\ref{sec:Helios}, SI Sec.~\ref{sec:64lq_data}).
}
\label{fig:computing}
\end{figure*}

\noindent\textbf{Computing:}
Finally, we put these ingredients together for an application-level benchmark of the pFT operation of an iceberg code, inspired by magnetic materials simulation. The iceberg code is particularly well-tailored to the simulation of quantum magnetism with Hamiltonians of the form $H=\sum_{\alpha \in x,y,z} \sum_{i,j} J^\alpha_{i,j} \sigma^\alpha_i \sigma^\alpha_j$, whose terms are each logical operators that are simple to implement in a pFT manner. Trotterized evolution under such Hamiltonians can be implemented with no extra overhead in an $\I{k}$ code using Helios' native physical gate set. This direct implementation of logical gates is not FT---a single 2Q gate error with Pauli components $XX$, $YY$, or $ZZ$ leads to a logical error. However, due to bias in the native 2Q gate errors, these dangerous errors occur with probability $\sim\!10\times$ lower than other detectable errors~\cite{PhysRevX.13.041052, floquet}.
This pFT computation scheme mixing FT state preparation, syndrome extraction, and measurement with non-FT Hamiltonian simulation gates enables a significant fidelity improvement over unencoded simulations of quantum computations at a scale that challenges classical simulation.

To benchmark this pFT scheme, we implement a logical version of the standard mirror-benchmarking (MB) protocol~\cite{Proctor22, proctor2022measuring, mayer2021theory}, that measures the survival probability $P(t) = |\langle \psi_0|U(\frac{t}{2})^\dagger U(\frac{t}{2})|\psi_0\rangle|^2$, where $U(\frac{t}{2})=e^{-iH t/2}$ is implemented by $s/2$ Trotter steps with size $\Delta t = t/s$ approximating the magnetic dynamics, and $|\psi_0\rangle$ is a reference state.
This protocol provides a sample-efficient way to estimate the global fidelity of time evolution for time $t$, without comparing to non-scalable classical simulations.
Specifically, we implement logical MB for a three-dimensional ($3$D) antiferromagnet with an easy-plane anisotropy to the spin-exchange, described by the paradigmatic $XY$-model Hamiltonian:
\begin{equation*}
    H_{xz} = \sum_{\langle ij \rangle} h_{ij} =  \sum_{\langle ij\rangle} J(X_i X_j + Z_i Z_j).
\end{equation*}
where $\langle ij\rangle$ denote nearest neighbors on a $3$D cubic lattice with dimensions $4\times 4\times 4$, and periodic boundary conditions.
Here, we have chosen the easy-plane anisotropy to lie in the $xz$ plane (rather than the conventional $xy$ plane) for convenience of implementation in the iceberg code. 
We choose to study a 3$D$ model to connect to realistic bulk materials, showcase the flexible programmability of
Helios' QCCD architecture in implementing models of arbitrary dimension, and because this high-dimensional geometry offers a stronger challenge to competing classical simulation methods for forward time evolution circuits.

We encode this system into a single $\I{64}$ iceberg code block, choosing a ferromagnetically polarized initial state $|\psi_0\rangle = |\bar{0}\rangle^{\otimes 64}$. The forward dynamics is implemented via $s/2$ first-order Trotter steps: $U(s\Delta t/2) \approx U_1(\Delta t)^{s/2}$, where $U_1(\Delta t) = \prod_{C}\prod_{\langle ij\rangle \in C} e^{-ih_{ij}\Delta t}$ and $C$ labels the 6 different sets of non-overlapping bonds indicated by the colorings shown in Fig.~\ref{fig:computing}a. We choose a step size $J\Delta t = \pi/8$ that is large enough to rapidly generate entanglement in a small number of Trotter steps, but small enough to be far from the Clifford limit of the system. Gates are Pauli-twirled using real-time random number generation \cite{guppy} to avoid coherent cancellation of errors.
A single SE round is inserted between the forward ($U$) and backward ($U^\dagger$) time evolutions. The MB protocol concludes with $Z$-basis measurements of all physical qubits, from which we reconstruct both the survival probability and $Z$-stabilizers.
The data is postselected on detecting no errors in the mid-circuit SE or final measurement.

The resulting survival probability and acceptance rates for MB sequences of up to $s=10$ total Trotter steps are shown in Fig.~\ref{fig:computing}d for both unencoded and iceberg QED encoded, with and without postselection on detecting leakage outside of the qubit subspace of any of the $k+2$ data qubits.
From exponential-decay fits to the survival probability versus logical gate number, we also extract an effective per-gate error $\tilde{\epsilon}_{2Q}$, which reflects a combination of gate errors and memory errors (see Methods~\ref{sec:Helios} for details and SI Sec.~\ref{sec:bilayer}, \ref{sec:64lq_data} for preliminary experimental results and more data).
We observe that the effective error rate $\tilde{\epsilon}_{2Q}$ is reduced $\approx 30\%$ by encoding in the iceberg code, at the expense of a reasonable increase in shot count incurred by postselecting out detected errors.
We note that obtaining high global state fidelity may be an overly conservative standard for quantum simulation applications that typically require the less stringent task of approximating local correlation functions.

While the MB protocol itself is trivial to simulate classically, 
for chaotic many-body Hamiltonians such as $H_{xz}$, it serves as a proxy for the global state fidelity for simulating a purely forward-time-evolved initial state: $U(s\Delta t)|0\rangle$. We anticipate that the pure forward time evolution may be difficult---if at all possible---to simulate classically, due to the large number of logical qubits (which renders direct state-vector simulation infeasible), high dimensionality, periodic boundary conditions, and rapid entanglement spreading (all of which challenge tensor-network methods).

\noindent\textbf{Outlook:}
We demonstrate computations below break-even on many logical qubits encoded in iceberg codes and their concatenations on Quantinuum Helios.
In particular, the favorable performance and overhead of the pFT simulation of the $XY$ model with 64 logical qubits at distance-2 motivates extending this approach to more demanding condensed-matter models, such as the Fermi–Hubbard or Kitaev models, of both academic and industrial relevance. 
When higher algorithmic fidelity or acceptance rates are required, fully FT, below-threshold protocols on higher-distance codes become necessary~\cite{Acharya2025googlethreshold}. 
We observe in this work that increasing the distance via concatenation already enhances acceptance rates without degrading fidelity.

We further present alternative fault-tolerant resource-state preparation schemes (SI Sec.~\ref{sec:alternative ft prep}) based on repeated stabilizer measurements, flag-at-origin~\cite{Forlivesi:2025ilq}, and an automatically generated flag circuit---the latter outperforming the gadgets of~\cite{Goto:2024xd} in depolarizing simulations. 
Decoding with look-up tables derived from circuit-level simulations is explored in SI Sec.~\ref{sec:decoding}.
They can provide a first layer of a hierarchical decoding pipeline, with other layers potentially leveraging the concatenated structure of iceberg codes~\cite{kakizaki2026concatenateddecoding}.
These codes additionally offer a low-overhead, naturally global logical gate set—including Clifford gates via automorphisms~\cite{liu2025coniq} and global magic-state preparation~\cite{Dasu:2025ktd} (as an alternative to distillation~\cite{SalesRodriguez2025queramagic})—that warrants further investigation.

These creative developments, together with advances in quantum architectures that offer arbitrary connectivity, including neutral-atom~\cite{Bluvstein2024queradevice} and trapped-ion~\cite{helios} platforms, are rapidly bringing fault-tolerant quantum computation with many logical qubits within reach.

\section*{Acknowledgments}

We acknowledge the entire Quantinuum team for innumerable contributions towards the design and operation of the Helios quantum computer, and we acknowledge Honeywell for fabricating the trap used in this experiment.
We thank Eli Chertkov, Michael Foss-Feig, Reza Haghshenas, Tyler LeBlond, Setso Metodi, Ciar\'{a}n Ryan-Anderson, Kartik Singhal, and David Stephen for useful discussions.

\section*{Author Contributions}

S.D., M.D., C.H.B., D.H., B.C., A.C.P., and D.A. devised the project and the protocol implementations.
S.D. devised the distance-4 QEC cycle, with improvements from B.C., and decoder as well as distance-4 GHZ preparation protocol and developed code for and executed all distance-4 experiments and distance-2 syndrome extraction experiments.
M.D. devised the GHZ-based distance-2 syndrome extraction as well as the FT FANOUT gate and its use in the GHZ experiments, and developed code for and executed distance-2 SPAM, GHZ state preparation, and $XY$ model experiments.
A.Y.G. developed code for and executed the logical gate benchmarking experiments.
Y.-H. C. developed code for and executed the unencoded GHZ state preparation experiments.
A.L., J.B., A.B., C.N.S., S.S., and B.S. performed additional simulations, QEC algorithm design, and circuit code development.
A.Y.G., Y.-H. C., K.M., and C.H.B. developed the logical gate benchmarking analysis.
M.S.A., J.A., J.G.B., N.Q.B., J.P.C. III, A.C., S.F.C., R.D.D., J.M.D., B.E., Car.F., Cam.F., J.P.G., A.H., C.A.H., A.A.H., A.I., C.J.K., Y.K., N.K., N.K.L., I.S.M., M.M., A.R.M., B.N., A.J.P., A.R., A.P.R., and S.J.S. designed, optimized, operated, and maintained the Helios hardware and software stack.
D.A. directed the overall project as lead principal investigator.
S.D., M.D., A.Y.G., A.L., J.B., A.B., Y.-H. C., S.S., N.Q.B., A.R., C.H.B., D.H., B.C., A.C.P., and D.A. contributed to the writing and editing of the manuscript.

\section*{Competing Interests}

The authors declare no competing interests.

\section*{Additional Information}
Supplementary Information is available for this paper. Correspondence and requests for materials should be addressed to David Amaro, \href{david.amaro@quantinuum.com}{david.amaro@quantinuum.com}.

\section*{Methods}
\label{sec:methods}

\subsection{Hardware specifications} \label{sec:Helios}

Helios is a 98-qubit trapped-ion quantum processor with $^{137}\text{Ba}^{+}$ hyperfine qubits that implements the quantum charge-coupled device (QCCD) architecture introduced in Refs.~\cite{Wineland1998-xz, Kielpinski2002}, first realized in Ref.~\cite{Pino2021} with $^{171}\text{Yb}^{+}$ qubits. Here we present an abbreviated summary of the technical specifications of the processor presented in Ref.~\cite{helios}. Qubit ions, along with their associated sympathetic coolant ions, are confined by the electromagnetic field generated by RF electrodes and physically transported around the ion trap by additional DC electrodes. The Helios architecture stores qubits in a large ring, sorts them by transport around the ring and through an X-junction \cite{PhysRevLett.130.173202} into an intermediate ``cache" region, and subsequently further transports the qubits into one of two ``legs" for cooling, quantum logic operations and temporary storage (see Ref.~\cite{helios}, Figs.~1 and 2). Enabled by a new classical control system and software stack, Helios implements real-time compilation of program operations (using the Guppy programming language \cite{guppy}) into physical hardware commands on the device.  This permits realization of many new algorithmic primitives, including dynamic allocation of qubits, classical control flow such as conditional statements and loops, early conditional termination of programs, and remote procedure calls to an external server that can perform arbitrary classical computations and return information needed for on-the-fly decisions about quantum logic (``gate streaming", demonstrated in Refs.~\cite{Niroula2025}, \cite{Liu2025}).

All primitive quantum logic operations (measurements, resets, and gates) on Helios are performed via interactions that are mediated by several lasers of different frequencies delivered to the QPU via free space; we refer the reader to Sec.~IIB of Ref.~\cite{helios} for further details. Qubits on Helios are initialized into the $|0\rangle$ state. Arbitrary one-qubit (1Q) gating operations are decomposed by the compiler into $U_{1Q} (\theta, \phi)= \exp \left(-i \frac{\theta}{2} (\cos(\phi) X + \sin(\phi) Y) \right)$ physical gates and $R_Z (\theta) = \exp \left( -i \frac{\theta}{2} Z\right)$ gates that are implemented fully in software as phase updates. The two-qubit (2Q) gate is mediated by the M{\o}lmer-S{\o}rensen interaction made phase-insensitive via wrapper pulses and implements $U_{ZZ} (\theta) = \exp \left(-i \frac{\theta}{2} ZZ \right)$ \cite{Lee_2005, Pino2021}, where $\theta$ is a user-specified angle and the gate error scales linearly with angle up to the maximal entangler $\theta = \pi / 2$ \cite{PhysRevX.13.041052}. Several types of measurements are available on Helios; in addition to the standard (possibly mid-circuit) measurement of isolated qubits, requests for simultaneous measurement of batches of qubits can trigger compilation into a ``protected" measurement operation that reduces crosstalk. Furthermore, users can request a ternary-valued measurement that may return $2$ to herald that the measured qubit has leaked outside the qubit subspace, at the expense of extra measurement error since this scheme requires several additional optical pulses. The leakage-heralded measurement may similarly be protected if ternary measurements are requested on batches of qubits simultaneously.

The infidelities associated to each of these quantum logic operations are reported in Table II of Ref.~\cite{helios}. For self-contained comparisons to the logical benchmarks in this work, we note in particular that the SPAM error associated to the standard measurement was measured at $p_{\text{SPAM}} = 4.8(6) \times 10^{-4}$ per qubit, and the 2Q gate error was measured to be $\epsilon_{2Q} = 8(2) \times 10^{-4}$ at the maximum gate angle. In the full-circuit benchmarks studied in that work, including random Clifford and mirror circuits, several heuristic formulas for circuit fidelity were used to extract an ``effective" 2Q gate infidelity $\tilde{\epsilon}_{2Q}$ roughly following a depolarizing-channel model. These estimates include additional overheads such as memory errors incurred during qubit idling and predict an effective 2Q gate infidelity at the maximum angle of $\tilde{\epsilon}_{2Q} = 20(2) \times 10^{-4}$. For comparison, in Fig.~\ref{fig:computing} we report $\tilde{\epsilon}_{2Q}$ for $XY$ model simulation experiments using \emph{partially}-entangling gates ranging from $\tilde{\epsilon}_{2Q} = 10.7(5) \times 10^{-4}$ (unencoded) to $\tilde{\epsilon}_{2Q} = 7.1(6) \times 10^{-4}$ (encoded in $\I{64}$, leakage-heralded). Here, and similarly to Ref.~\cite{helios}, these values are extracted from a simple depolarizing model $F(N_{\text{gates}}) = A(1-\frac54 \tilde{\epsilon}_{2Q})^{N_{\text{gates}}}$ with $N_{\text{gates}}$ the number of two-qubit gates in the circuit, so that $\tilde{\epsilon}_{2Q}$ is an effective average fidelity and $A$ is an overall normalization capturing SPAM errors.

\subsection{QEC benchmarking}\label{sec:QEC Benchmark}
For all $k > 1$ QEC/QED experiments, due to the application of one or more rounds of QEC/QED, we assume that the logical memory channel on all $k$ can be written as a stochastic logical Pauli channel \cite{Beale2018}. Therefore, we can write this channel in the form of
\begin{equation}
    \mathcal{E}(\rho) = \sum_i p_i P_i\rho P_i, \label{eq:pauli mult}
\end{equation}
where the sum is over all $k$ qubit Paulis. For the $\I{8}\circ\I{6}$ experiments, in which a single round of QEC was performed on $\ket{\overline{0}}^{\otimes 48}$ and $\ket{\bar{+}}^{\otimes 48}$, we apply a union bound to the process fidelity. Let us denote by $p_{\ket{\overline{0}}^{\otimes 48}}$ and $p_{\ket{\bar{+}}^{\otimes 48}}$ the probability of any logical error occurring to the $\ket{\overline{0}}^{\otimes 48}$ and $\ket{\bar{+}}^{\otimes 48}$ states. The process infidelity is $I = \sum p_i$, where the $p_i$ are the non-identity Pauli error probabilities in Eq.~(\ref{eq:pauli mult}). Since any $P_i$ containing an $X$ or a $Y$ will cause a logical error on $\ket{\overline{0}}^{\otimes 48}$ and any $P_i$ containing a $Y$ or a $Z$ will cause a logical error on $\ket{\bar{+}}^{\otimes 48}$, we have that 
\begin{equation}\label{eq:proc mult bound}
    \frac{1}{2}(p_{\ket{\overline{0}}^{\otimes 48}} + p_{\ket{\bar{+}}^{\otimes 48}}) \leq I \leq p_{\ket{\overline{0}}^{\otimes 48}} + p_{\ket{\bar{+}}^{\otimes 48}}.
\end{equation} The average infidelity for this $48$-qubit process is essentially identical to the process infidelity in this case, differing by a factor of $\frac{2^{48}}{2^{48}+1}$, and we report the upper bound on the average infidelity divided by the number of logical qubits as the average infidelity of a QEC cycle per logical qubit for these experiments.

For the $k=4$ experiments, repeated Pauli errors which may occur on a single qubit due to repeated applications of QEC/QED will cancel, so the above analysis does not apply directly. To deal with these cancellations, we restrict our attention to single-qubit logical memory channels by marginalizing over all but one qubit. For logical qubit $j$, the reduced channel for $\tilde{\rho} = \tr_{[k]\backslash j} \rho$ is
\begin{equation*}
    \mathcal{E}(\tilde{\rho}) = \tilde{p}_I \tilde{\rho} + \tilde{p}_X X \tilde{\rho} X + \tilde{p}_Y Y \tilde{\rho} Y + \tilde{p}_Z Z \tilde{\rho}  Z
\end{equation*}
where $\tilde{p}_P$ is the probability of the single qubit Pauli $P$ being applied to $\tilde{\rho}$, which can be obtained by summing over the terms $p_i$ in the full channel which apply a $P$ operator to logical qubit $j$. The process infidelity $\tilde{I_j}$ of the reduced channel is $\tilde{p}_X + \tilde{p}_Y + \tilde{p}_Z$.

To estimate the $\tilde{p}_P$, we can make use of a formula which takes into account possible cancellations of Pauli errors for the survival probability of a single logical qubit under repeated rounds of QEC \cite{RyanAnderson2021} given by
\begin{equation}\label{eq:qec sq}
    p_{\psi}(c) = \frac{1}{2} + \left(s_{\psi} - \frac{1}{2} \right) \left(1 - 2 q_{\psi} \right)^c.
\end{equation}
where $p_\psi(c)$ is the survival probability of an initial state $\ket{\overline{\psi}}$ after $c$ rounds of QEC. The logical SPAM error of producing state $\ket{\overline{\psi}}$ is $s_{\psi}$, and $q_{\psi}$ is the probability of introducing a logical error to $\ket{\overline{\psi}}$ after one round of QEC. Since both logical $X$ and $Y$ operators will flip $\ket{\overline{0}}$, we have that $p_{\ket{\overline{0}}} = \tilde{p_X} + \tilde{p_Y}$ and, similarly,  $p_{\ket{\bar{+}}} = \tilde{p_Y} + \tilde{p_Z}.$ Therefore, we can estimate $\tilde{I}$ from the survival probabilities of the input states over repeated rounds of QEC by again applying a union bound,
\begin{equation}\label{eq:proc sing bound}
    \frac{1}{2}(p_{\ket{\overline{0}}} + p_{\ket{\bar{+}}}) \leq \tilde{I} \leq p_{\ket{\overline{0}}} + p_{\ket{\bar{+}}}.
\end{equation}

We obtain estimates for the $p_\psi$ by fitting experimental data to the curve in Eq.~\eqref{eq:qec sq} through maximum likelihood estimation to take into account non-normality of the sample distributions due to rare events. Uncertainties were derived from Markov Chain Monte Carlo sampling of the curve-fitting parameters. For both $k=4$ experiments, we report the average of the upper bounds for these single-qubit infidelities as the infidelity of a QEC cycle per logical qubit. This quantity is the analog of the depth-1 memory error infidelity reported in \cite{helios}, which does not capture correlations in the noise between different qubits and instead treats the memory error as a single-qubit process. We note that when taking the average of the infidelity over the reduced channels, we add the statistical uncertainties of the infidelities linearly and not in quadrature, since the estimators for the different logical qubits are not independent. For all experiments, confidence intervals were obtained by adding the $1\sigma$ statistical uncertainty (68\%) of all estimators to the lower and upper bounds in Eq.~\eqref{eq:proc mult bound} and Eq.~\eqref{eq:proc sing bound}.

The average infidelity $1.1(3)\times 10^{-4}$ corresponding to non-leakage component of the depth-1 linear memory error reported in the main text is obtained from the transport-1QRB benchmarking data from \cite{helios}. We fit the non-leakage-heralded memory error to the quadratic curve $ax^2 + bx - c$ where $x$ is the number of transport rounds and $b$ is the depth-1 linear memory error and $a$ is the quadratic memory error. For those experiments, $a = 6(3)\times 10^{-5}$, $b = 1.1(3)\times10^{-4}$, and $c = 2(5)\times 10^{-5}$.

We compare to the linear memory error because a limited amount of dynamical decoupling (DD) is used in all QEC/QED memory experiments in the second half of each cycle when the $Z$ syndromes are extracted, since $Z$-type noise on the data qubits commutes with the gates therein and can therefore build coherently. In the $d=4$ case, X pulses are inserted before, at the midpoint of, and after the high-level syndromes are extracted, while in the $d=2$ case, X pulses are applied to the qubits that are gated in the second half of the Z syndrome extraction circuit, then to all the qubits at midpoint of Z syndrome extraction, and then to the qubits gated in the first half of syndrome extraction at the end of the QED cycle. This ensures that X pulses are applied opportunistically in a way that does not induce additional transport while commuting with the extracted stabilizers. In the first half of each cycle, no DD is used since $X$ syndrome extraction should convert $Z$-type noise into incoherent errors.

For all $k=4$ experiments, the QED/QEC cycle rejection rates over $c \in \{1,2,4,8\}$ cycles were fitted to the exponential $(1-r)^c$, where $r$ is the rejection rate per cycle. In the $X$ basis, $r = .060^{+2}_{-2}$ and $r = .027_{-1}^{+2}$, for $d=2$ and $d=4$ respectively. In the $Z$ basis, $r =.070_{-3}^{+2}$, and $r=.030_{-2}^{+2}$, for $d=2, 4$ respectively.

For all $d=2$ and $d=4$ experiments, leakage-heralded measurement is used to postselect on any leakage, since leakage at certain locations can potentially violate fault-tolerance in the absence of gadgets which introduce additional overheads \cite{brown2020}. An alternative solution is leakage repump \cite{Hayes_2020}, which converts leakage into depolarizing noise that can be fault-tolerantly handled by QEC/QED without additional circuit gadgets. All QEC/QED acceptance rates reported are conditioned on the absence of leakage and successfully preparing the data and ancilla qubit code blocks. Leakage acceptance rates, and preacceptance rates for state preparation, are given in SI Sec.~\ref{sec:QEC}.

\subsection{Logical gate benchmarking}\label{sec:methods_lcb}
To perform logical cycle benchmarking on gates in/between $\I{4}$ code blocks, we sample a random logical Pauli string $P$ over all logical qubits and initialize the system in one of its corresponding logical eigenstates. For inter-block gates between two $\I{4}$ blocks, Pauli strings of length $8$ are used, reflecting the fact that faults can propagate across all logical qubits in both code blocks due to the nonlocal structure of syndrome-extraction gadgets. We repeat the Pauli-twirled logical gate (cycle) $L$ times, perform syndrome extraction at a fixed frequency, and then measure in the $P$ basis to obtain an estimate of the expectation value of $P$. We choose a maximum cycle depth of $L = 64$ for intra-block gates and $L=48$ for inter-block gates. Syndrome-extraction gadgets are interleaved at a fixed cadence of one round per eight logical cycles, enabling detection of single-qubit faults and two-qubit Pauli errors not in the set $\{XX,YY,ZZ\}$ that occur during gate execution or syndrome extraction. Any shot with a detected error is discarded, and the remaining shots are used for estimating the logical observable.

Pauli twirling is implemented using shot-level randomness via the real-time random number generation implemented in the Guppy programming language. The Pauli twirling allows the aggregate error channel on the input state $\rho$ to be written as a stochastic Pauli channel
\begin{equation*}
    \mathcal{E}(\rho) = \sum_i p_i P_i\rho P_i,
\end{equation*}
where the sum is over all Pauli operators (modulo an overall phase) and the Pauli error probabilities $p_i \in[0,1]$ sum to 1.
Taking a Walsh-Hadamard transform of these error probabilities, we obtain the \emph{Pauli fidelities}, $f_j$, which are equivalent to the eigenvalues of the Pauli error channel. Specifically, 
\begin{equation*}
    f_j = \sum_i (-1)^{\langle i,j\rangle} p_i,
\end{equation*}
where the symbol $\langle i,j\rangle$ is 0 or 1 depending on whether the Paulis $P_i$ and $P_j$ commute or anti-commute.

Denoting a noisy circuit of length $L$ as $C_j$, the output of the Pauli-basis measurement is $\mathbb{E}_L[P_j] = \text{Tr}({P_j C_j(P_j)})$. 
We estimate the Pauli fidelity by fitting the expectation value decay with increasing $L$ to an exponential model, i.e. $\mathbb{E}_L[P_i] = A_j f_j^L$.
The fit is performed using maximum likelihood estimation (see SI Sec.~\ref{sec:si_lgb} for more details).
We then average the Pauli fidelities to obtain an estimate of the process fidelity $F_{\mathrm{pro}}$ of the logical cycle, which, for the case of a single logical 2Q gate per cycle, can be converted to the gate-averaged fidelity via $F_{\mathrm{avg}} = (4F_{\mathrm{pro}} + 1)/5$. 
The $F_\text{avg}$ for each of the three experiments are reported in Fig.~\ref{fig:entangling}, along with error bars representing the $1\sigma$ (68\%) confidence intervals obtained via standard parametric bootstrap.
We discuss the details of the statistical analysis along with the postselection rates from error detection in SI Sec.~\ref{sec:si_lgb}.

\subsection{Logical GHZ state preparation} \label{sec:methods_ghz}

We estimate the GHZ state preparation fidelity in the $d=2$ codes via a single-shot protocol, making use of the fact that the GHZ state is uniquely identified by its stabilizers $\prod_i \logX^{i}$ and $\logZ^i \logZ^{i+1}$. The former can be measured in a quantum non-destructive (QND) fashion while the latter can be obtained by destructive readout of all qubits in the $Z$ basis. In contrast to other recent large-scale GHZ state experiments \cite{Wei2025}, we do not postselect on any of the GHZ stabilizers.
The logical GHZ state in a single global $\I{k}$ code is particularly simple; as described in the main text, it factorizes into a physical GHZ state on $n-2$ qubits and a Bell pair. As such, it can be directly prepared fault-tolerantly by a circuit similar to the code initialization circuit (see SI Sec.~\ref{sec:ghz_prep} for an illustration). The GHZ stabilizer $\prod_i \logX^{i}$ can be measured efficiently since it is stabilizer-equivalent to $X_0 X_{n-1}$. We measure $X_0 X_{n-1}$ twice to ensure robustness against measurement errors by coupling to a single ancilla each time that is dynamically allocated before reading out the parity.

Break-even is relatively easy to achieve for logical GHZ state preparation on a global $\I{k}$ code block, since the protocol prepares a physical GHZ state of essentially comparable size and then enhances the logical fidelity via postselection. By contrast, preparing a logical GHZ state across many small iceberg code blocks is significantly more general and nontrivial, corresponding to a complicated entangled state of all physical qubits. We describe here the case of many $\I{4}$ iceberg code blocks, which uses the same FANOUT gate and code structure as in our logical gate benchmarking protocol. However, the preparation for any number and size of code blocks follows a similar structure. This protocol is displayed in Fig.~\ref{fig:entangling}c and begins with preparation of one code block in the $|\overline{\text{GHZ}}_4\rangle$ state by the procedure described above while all other code blocks are initialized in the $|\overline{0}\rangle^{\otimes 4}$ state. We then extend the logical GHZ state on the first code block to the full logical GHZ state using the logical FANOUT gate described above. This gate maps the state $|\overline{\text{GHZ}}_{k_1}\rangle \otimes |\overline{0}\rangle^{\otimes k_2 }$ to the state $|\overline{\text{GHZ}}_{k_1 + k_2}\rangle$ when controlled on any logical qubit in the $\I{k_1}$ code. One then proceeds analogously to the usual GHZ state preparation (in which the $|+\rangle$ state on one qubit is extended to the $|\text{GHZ}\rangle$ state on all qubits by means of a tree of CX gates), by applying the logical FANOUT gate in a log-depth tree spanning iceberg code blocks. Each application of the logical FANOUT only requires four two-qubit gates compared to $n$ when instead applying transversal CX gates between each pair of code blocks, and straightforwardly applies to the case that the code blocks are of incommensurate size. Subsequently, the GHZ stabilizer $\prod_i \logX^{i}$ can still be measured relatively efficiently using the product of $X_0 X_{n-1}$ across all code blocks. We order the gates so that first all top qubits $q_0$ in each block are coupled to the ancilla followed by all bottom qubits $q_{n-1}$ in the opposite order, c.f. Fig.~\ref{fig:entangling}c, so that $X$ hook errors are almost always detected on at least one code block. We again repeat the measurement of $\prod_i \logX^{i}$ to be FT to measurement errors on the ancilla. We note that although many applications of the $\I{k}$ codes use the pFT continuously-parameterized two-qubit Pauli rotation gates, both GHZ state preparation protocols detailed above are actually fully FT to single weight-2 faults.

The logical operators of $\I{k_2} \circ \I{k_1}$ are rectangles and we denote them by $\logX^{i, j} = X_{0, 0} X_{i, 0} X_{0, j} X_{i, j}$, $\logZ^{i, j} = Z_{i, j} Z_{n_1-1, j} Z_{i, n_2-1} Z_{n_1-1, n_2-1}$. To prepare a global GHZ state on the $48$ logical qubits of the $\I{8}\circ\I{6}$ code, we begin by initializing a code block in the $\ket{\overline{0}}^{\otimes 48}$ state. The stabilizers of this state already contain the Z-type GHZ stabilizers $\overline{Z}^{i_1,j_1}\overline{Z}^{i_2,j_2}$ for $1\leq i_1,i_2 \leq 6, 1 \leq j_1,j_2, \leq 8$. To enforce the condition that $\prod_{i=1,j=1}^{6,8}\overline{X}^{i,j}$ stabilizes the state, we measure this operator using an $\I{6}$ ancilla. We observe that the operator $X_{0,0}X_{7,0}X_{0,9}X_{7,9}$ commutes with all the $\I{8}\circ\I{6}$ stabilizers and anticommutes with all $\overline{Z}^{i,j}$. Therefore, this operator is equal to $\prod_{i=1,j=1}^{6,8}\overline{X}^{i,j}$ up to stabilizers. We can measure this operator by initializing an $\I{6}$ ancilla with stabilizer $X_0X_7$, performing transversal CX gates from this ancilla to the top and bottom blocks of the $\I{8}\circ\I{6}$ code, and then measuring this ancilla in the X basis, see Fig.~\ref{fig:entangling}d. This will cause the state to collapse to a $\pm 1$ eigenvalue of $\prod_{i=0,j=0}^{5,7}\overline{X}^{i,j}$ with random sign. We flip the sign of this eigenvalue in software to $+1$ by updating the Pauli frame. Since we are preparing the GHZ state in a $d=4$ code, we must measure the eigenvalue of $\prod_{i=1,j=1}^{6,8}\overline{X}^{i,j}$ a second time to protect against two measurement/gate errors which could have caused a logical error in the first $\I{6}$ measurement. If the two $\I{6}$ measurements differ by a logical operator, we postselect. By postselecting in the presence of two faults, we will prepare the state with logical error rate of $O(p^3)$, and, by correcting a single error anywhere in the preparation procedure, the postselection probability will be $O(p^2)$. See SI Sec.~\ref{sec:ghz_prep_d4} for simulations demonstrating these scalings with the error rate.

\clearpage
\onecolumngrid
\appendix

\setcounter{equation}{0}
\setcounter{figure}{0}
\renewcommand{\theequation}{S\arabic{equation}}
\renewcommand{\thefigure}{S\arabic{figure}}
\renewcommand{\thesection}{\Roman{section}}
\renewcommand{\thesubsection}{\arabic{subsection}}
\renewcommand{\appendixname}{}

\begin{center}
{\large \bf Supplementary Information  for:\\ Computing with many encoded logical qubits beyond break-even}
\end{center}

\section{State preparation and error mitigation techniques} \label{sec:state_prep_gadgets}

In the $\I{k}$ codes, we prepare the logical zero state $|\overline{0}\rangle^{\otimes k} = |\text{GHZ}_k\rangle$ via a tree of CX gates that originate at a $|+\rangle$ state with target qubits initialized in the $|0\rangle$ state. This tree can be extended to contain up to $n/2$ branches where $n=k+2$, in which case the state preparation proceeds in log depth. Pairs of branches are coupled to an ancilla qubit that reads out the $ZZ$ operator of the qubits at the end of each branch; therefore, up to $n/4$ ancilla qubits and $n/2$ additional CX gates are required, although the ancilla qubits may be reused at the expense of reduced circuit parallelization.
The simplest such case is the two-branch preparation presented in Ref.~\cite{iceberg-beyond-the-tip} and for which the corresponding gate connectivity graph is displayed in Fig.~\ref{fig:halfdepth}.

\begin{figure}[h!]
\centering
\includegraphics[width=\textwidth]{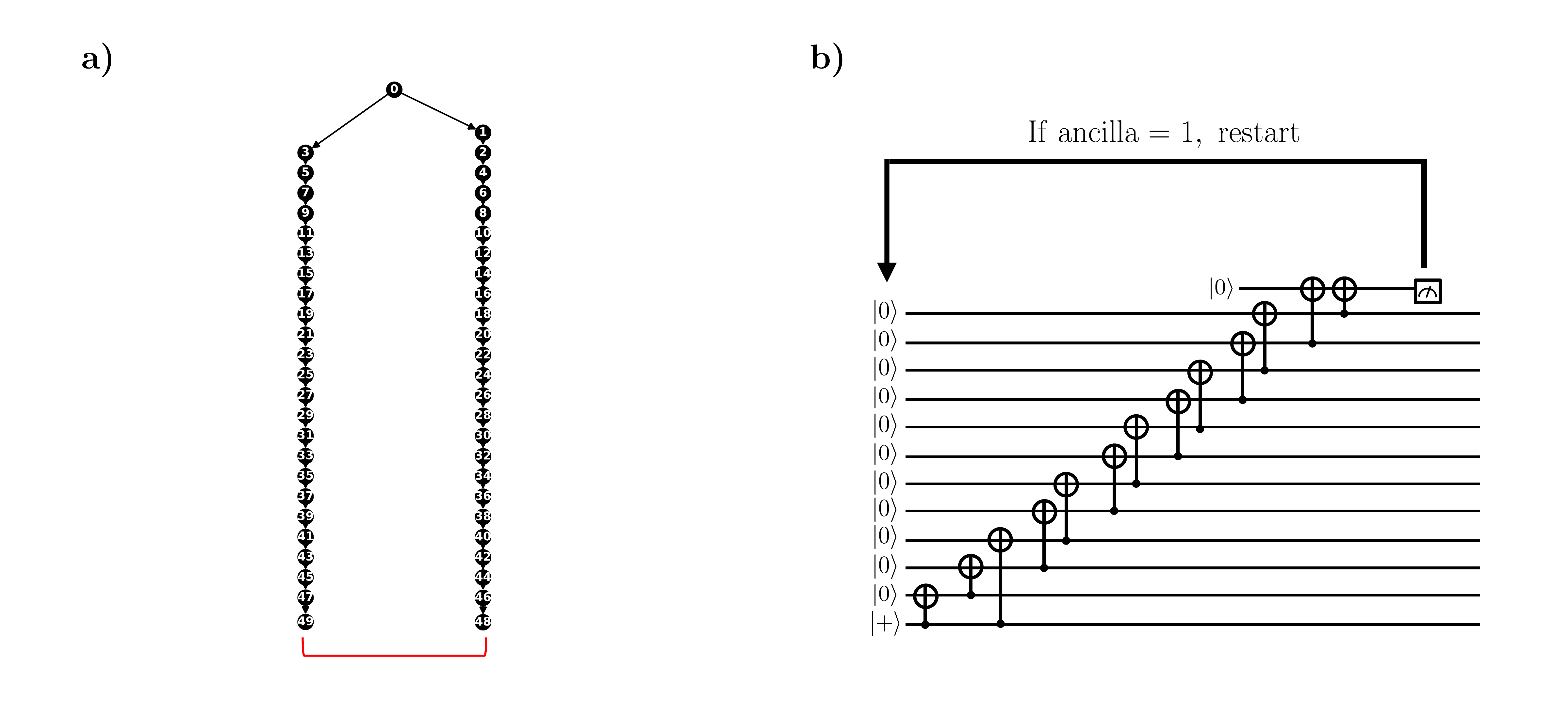}
\caption{\textbf{Two-branch state preparation gadget in the $\I{k}$ codes.}
\textbf{(a)} Two-branch tree specifying preparation of the logical zero state of the $\I{48}$ code. Nodes are labeled with the index $i$ corresponding to qubit $q_i$ in the code. Qubit $q_0$ is initialized in the $|+\rangle$ state while all other qubits are initialized in the $|0\rangle$ state; directed arrows indicate the control and target of CX gates between the qubits and all arrows terminating on nodes at the same vertical height are performed in parallel. The red bracket connecting the terminal nodes indicates that the corresponding operator $Z_{48}Z_{49}$ will be read out by coupling to an ancilla. \textbf{(b)}  The circuit gadget instantiating the two-branch tree in the $\I{10}$ code for illustration.
}\label{fig:halfdepth}
\end{figure}

This gadget is fault-tolerant: all single faults are detectable or do not act on the logical state, as we now explain. First, a weight-1 $Z$ error does not propagate down the CX tree at all, so it will be detected by the $X$ stabilizer later. A weight-1 $X$ error does propagate down the tree; however, it either occurs on $q_0$ in which case it applies the code stabilizer $\prod_i X_i$ to the physical state, or it flips the value of the $ZZ$ operator measured at the end of the branches, since it can only spread down one branch. Gate errors may also cause single weight-2 faults. Any Pauli error on a single gate that does not involve $q_0$ either flips the $ZZ$ operator at the end of the branches (if it applies an $X$ or $Y$ to the target), causes a detectable syndrome, or applies a $ZZ$ operator that acts as the identity on a physical GHZ state. The situation is similar for gates that do involve $q_0$; although these can cause Pauli errors to potentially propagate down both tree branches if they apply an $X$ or $Y$ to $q_0$, either this results in an overall detectable $Y$ error, or it applies the $\prod_i X_i$ code stabilizer.

In Fig.~\ref{fig:logdepth} we illustrate the log-depth version of this state preparation protocol. The key observation is that the same arguments above for fault-tolerance apply as long as we specify the tree such that the $ZZ$ operators measured always correspond to two qubits whose only common ancestor is the root of the tree $q_0$. Conceptually, one builds up such a tree iteratively by separating nodes into two sets: those with $q_1$ as an ancestor, and those without. Each time a new successive node is added to the tree, it is either added to the descendants of $q_1$ or to the side of the tree that does not descend from $q_1$ depending on which set is larger, to maintain a balance between the size of the two sets. Following this procedure, one obtains a log-depth tree in which half of the terminal nodes descend from $q_1$ and half do not; pairing off nodes in each set obtains the $ZZ$ operators that should be checked for fault-tolerance. It is possible, depending on the desired depth of the state preparation protocol, for the tree generated in this fashion to be imbalanced by at most one surplus/deficit qubit descending from $q_1$; the leftover qubit may be paired with $q_0$ in checking $ZZ$ operators of the pairs of terminal nodes.

\begin{figure}[h!]
\centering
\includegraphics[width=\textwidth]{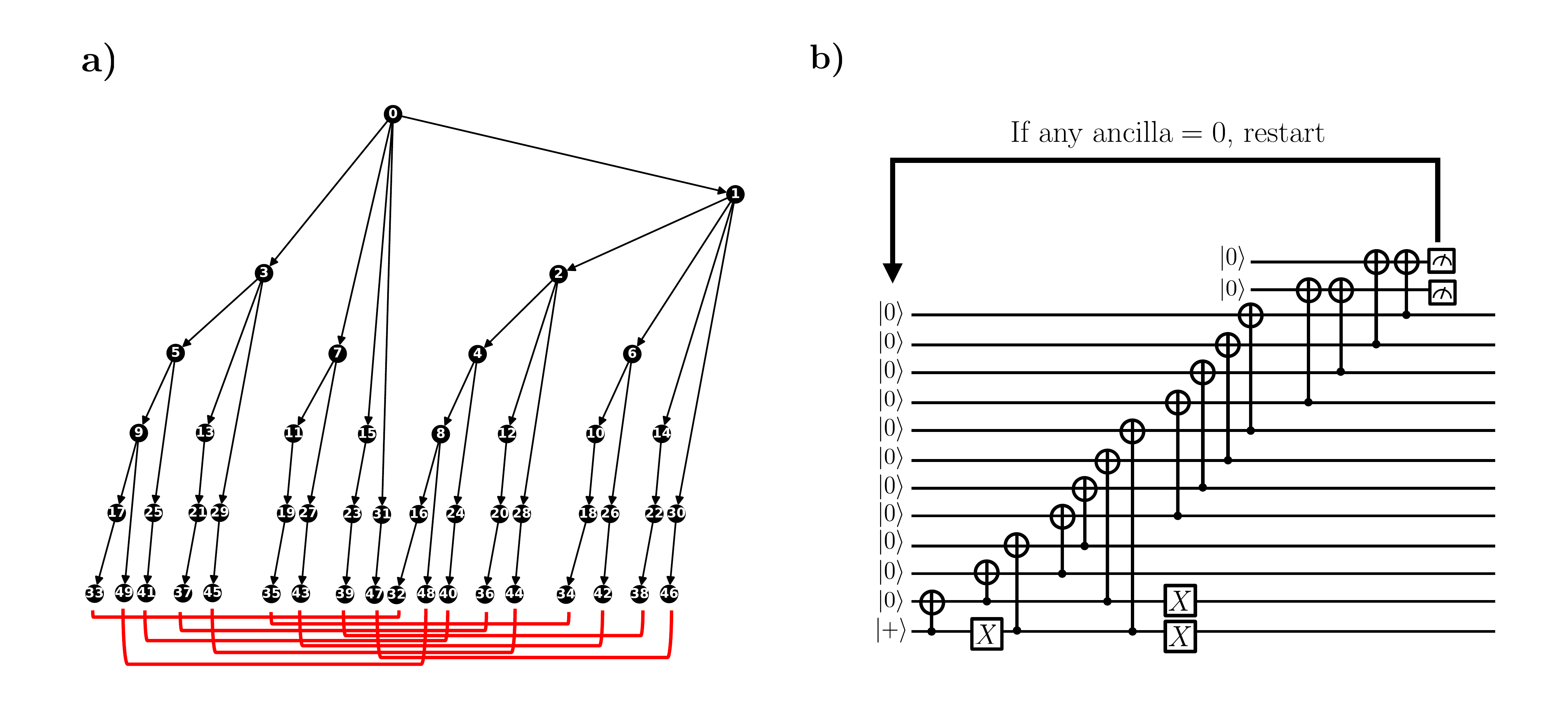}
\caption{\textbf{Log-depth state preparation gadget for $\I{k}$ codes.}
\textbf{(a)} The balanced tree structure specifying preparation of the logical zero state in the $\I{48}$ code. As before, red brackets indicate pairs of qubits whose $ZZ$ operator will be read out by coupling to an ancilla (note that in this particular case, the qubit with index $m$ on the left connects to the qubit with index $m-1$ on the right). \textbf{(b)} Circuit instantiation of the log-depth tree in the $\I{10}$ code, as also depicted in Fig.~\ref{fig:computing}. Pauli $X$ gates in the circuit implement the DFS-based error mitigation discussed below; note that this flips the noiseless value of the ancilla measurements that is checked as part of the repeat-until-success protocol.
}\label{fig:logdepth}
\end{figure}

Rather than postselect on errors detected on the ancilla measurements that record the various $ZZ$ operators, we instead use a repeat-until-success (RUS) protocol in which any detected errors trigger a reset of all qubits followed by a new attempt at initializing the code. In Fig.~\ref{fig:rus} we report statistics on the number of loop iterations for RUS to pass, across all experiments in which it was utilized. Even for the largest code studied, $\I{94}$, initialization succeeds with exceptionally high probability in only one or two loop iterations, with a maximum of six loop iterations required across all experiments of any size in this work. We group the results in Fig.~\ref{fig:rus} for each experiment type (SPAM, GHZ, $XY$ model simulation) and aggregate results within each category across experiments using the same-size codes since other differences do not affect the state preparation. For GHZ state preparation experiments using many $\I{4}$ code blocks we report the success with which we can prepare each individual $\I{4}$ code block and aggregate the statistics across all $\I{4}$ code blocks in a given shot. However, we separate the results reported by total number of code blocks since, e.g., the memory error experienced by the qubits may depend on the total number of qubits involved in the experiment.

\begin{figure}[h!]
\centering
\includegraphics[width=\textwidth]{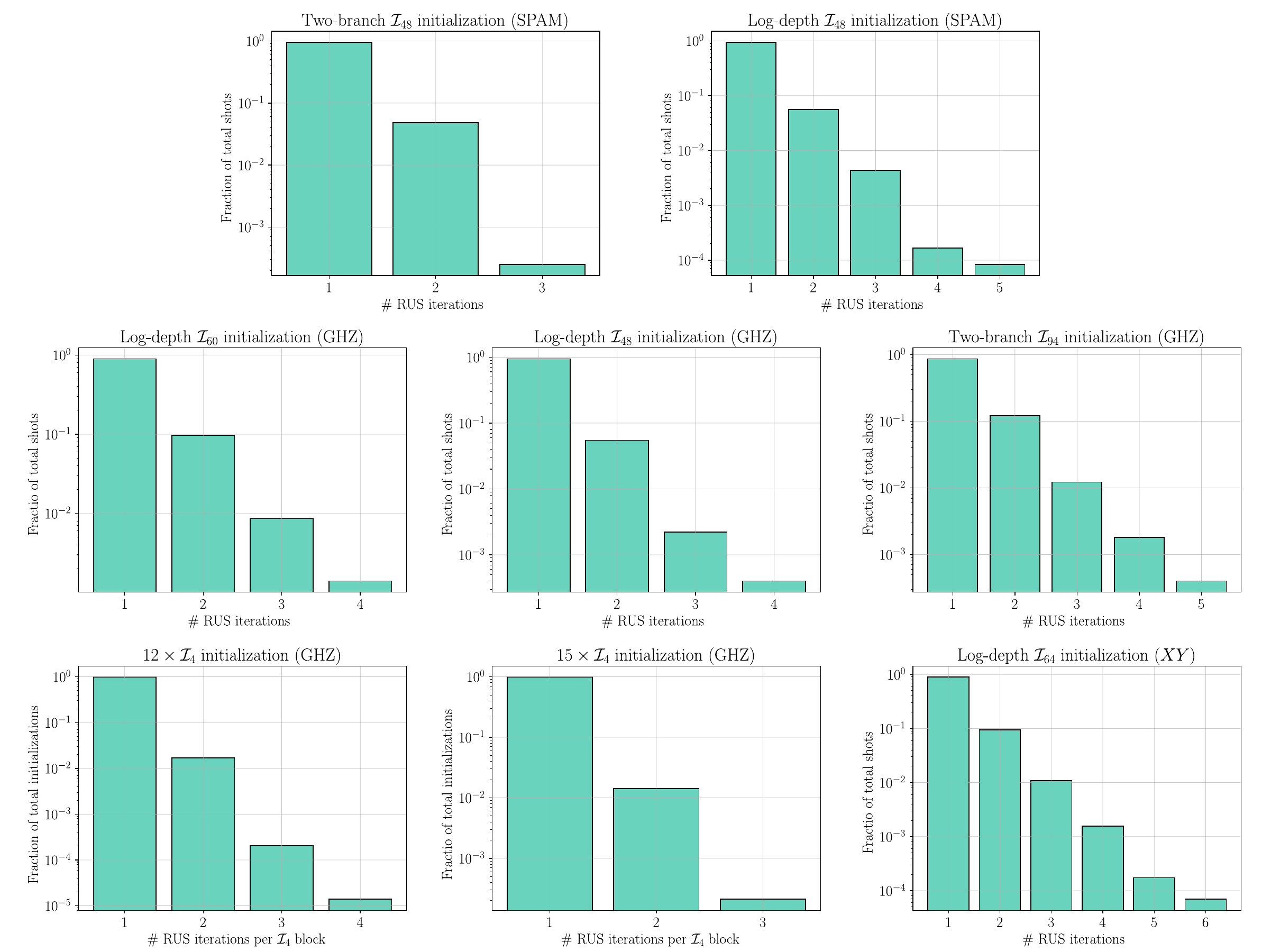}
\caption{
\textbf{Repeat-until-success statistics for distance-2 experiments.} We report the fraction of the time that a code block required each number of RUS iterations before no error was detected, corresponding to the fraction of total shots for experiments using a single $\I{k}$ code block, or fraction of total initializations $=$ (total shots times number of code blocks) for experiments using multiple code blocks. Experiments of the same type (i.e., SPAM, GHZ state preparation, or $XY$ model simulation) are aggregated on the plots above when they only differ in aspects that do not affect state preparation. Shot counts for each experiment are reported in the relevant SI sections below.
}\label{fig:rus}
\end{figure}

Although the log-depth initialization gadget substantially reduces the circuit depth of state preparation, qubit idling remains a concern due to the potential for the buildup of coherent dephasing errors (for instance, due to spatial and temporal variance in magnetic fields) which can grow quadratically at small idle times. Below, we discuss additional error mitigation techniques that were employed only in the log-depth gadget to further reduce coherent memory errors.

To the extent that coherent memory errors are \emph{spatially} uniform across the device, motivated by \emph{decoherence-free subspace} (DFS) codes \cite{PhysRevLett.81.2594, Kielpinski2002} we insert a Pauli X operator into the log-depth tree so that after each full layer of the tree is applied, the qubits are in a ``decoherence-free state'' that is a superposition of computational basis states with an equal mixture of 0s and 1s. The net effect at the end of state preparation is equivalent to the application of $X_i$ operators on half of the physical qubits, so one can simply flip the relevant bits in postprocessing with the understanding that in the case that $n = 2j$ with $j$ odd, the code space throughout the circuit is now defined by the $S_Z = -1$ subspace. In Fig.~\ref{fig:logdepth} we show the inserted Pauli $X$ operator along with two additional Pauli $X$ operators that we insert at the end of initialization so that $q_0$ and $q_{n-1}$ are in opposite states; since these two qubits correspond to the pivot qubits, they may idle for long periods during computations in a single $\I{k}$ code. Consequently, this may provide additional protection against coherent dephasing on those two qubits in particular. In deep circuits, we expect that the code does not remain in a DFS state; however, at the beginning and end of the circuit, and in particular during any terminal syndrome extraction, we might expect the state to remain approximately DFS, so this technique provides a low-overhead way to mitigate some coherent error both during the initialization gadget and potentially during an encoded circuit.

To the extent that coherent memory errors may be spatially varying but relatively \emph{temporally} uniform, we insert two dynamical decoupling (DD) $X$ pulses on the qubits idling for the longest period during state preparation. As in Ref.~\cite{floquet}, we alternate the sign of the DD pulses on each qubit between $+X$ and $-X$ (which have physically distinct implementations on Helios) so as to cancel any potential coherent over-rotation on the single-qubit gate itself in the positive vs. negative directions. We do not depict these gates in Fig.~\ref{fig:logdepth} but they are arranged to occur during the long idle period on half the system qubits on the right-hand side of the gadget, with roughly equal temporal spacing.

\begin{figure}
\includegraphics[width=\textwidth]{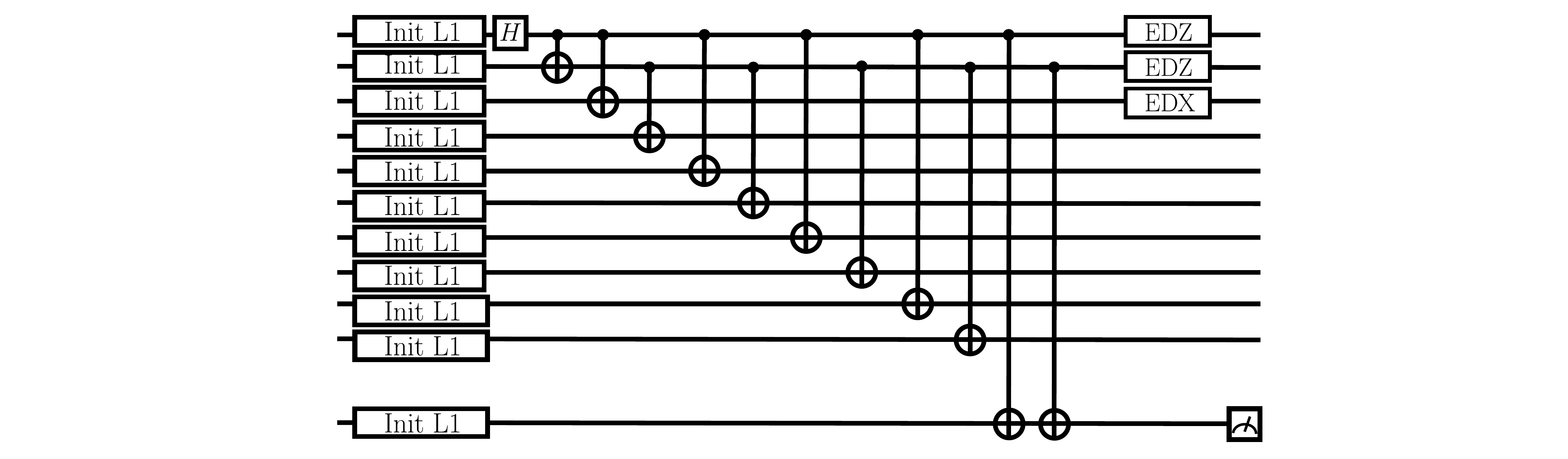}
\caption{\textbf{Two-branch state preparation gadget in the $\I{8} \circ \I{6}$ code.}
All wires are logical, representing $\I{6}$ blocks and all CX gates are transversal. The ``init L1'' gadgets initialize the $\I{6}$ using the $d=2$ two-branch state preparation circuit of Fig.~\ref{fig:halfdepth}. 
}\label{fig:halfdepthd4}
\end{figure}

Finally, we introduce a new, FT two-branch $\ket{0}^{\otimes k_1k_2}$ state preparation gadget for concatenated iceberg codes $\I{k_2}\circ \I{k_1}$ which generalizes the $d=2$ circuit in Fig.~\ref{fig:halfdepth}. This circuit is shown for the $\I{8}\circ\I{6}$ code in Fig.~\ref{fig:halfdepthd4}. First, we initialize the lower-level $\I{k_1}$ blocks using the $d=2$ two-branch circuit in Fig.~\ref{fig:halfdepth}. Then, we perform a two-branch circuit comprised of transversal CX gates with a logical ancilla which measures the logical $ZZ$ operators of two $\I{6}$ blocks to catch all cascading logical $X$ errors, analogous to the $d=2$ case. The main difference is that we arrange the CX gates so that they have only two controls, so that any weight-two Z error resulting from a single failure will be caught by the $d=2$ EDZ gadgets, which extract $X$ syndromes of the $\I{k_1}$ code using a Bell-pair syndrome extraction circuit similar to the non-destructive part of Fig.~\ref{fig:se_protocols}b. Similarly, the EDX gadget will extract $d=2$ Z syndrome of the third $\I{k_1}$ block to catch high weight $X$ errors which may result from an $XX$ error on one of the CX gates in the first transversal CX between the top two code blocks. We postselect on any non-trivial measurements in these EDZ/EDX gadgets as well as in the bottom transversal $Z$ measurement of the $\I{k_1}$ ancilla. If the $\I{k_2}\circ\I{k_1}$ code passes these checks, a single failure will lead to at most a weight-one error on the code block. Pairs of gate failures can lead to columns of errors on the output code block of weight possibly larger than two. However, since these error strings intersect the support of any weight-four logical operator on at most two qubits, we shall show in Sec.~\ref{sec:QEC} that such error are ``generalized weight-two," in the sense that they can be treated as weight-two errors and the usual proofs that computations are FT to distance four will still apply. 

\section{Data from SPAM experiments} \label{sec:spam}

In Table~\ref{tab:spam} we report explicitly all distance-2 data appearing in Fig.~\ref{fig:encoding}b and Fig.~\ref{fig:encoding}c.

We report separately the distance-4 data, which does not use RUS in its state preparation protocol, so that we can distinguish the pre-acceptance and post-acceptance rates of the experiment, where the \emph{pre-acceptance rate} corresponds to the fraction of leakage-postselected shots that pass state preparation and the \emph{post-acceptance} rate refers to the fraction of remaining shots that are not postselected due to a weight-two error. In the $\I{8} \circ \I{6}$ SPAM experiment the leakage acceptance rate was $0.775^{+7}_{-6}$, the pre-acceptance rate was $0.654(8)$, and the post-acceptance rate was $0.987(2)$. No logical errors were recorded after performing 4000 shots of the experiment, which is consistent with an upper bound (Wilson 68\% confidence interval) on the SPAM infidelity per logical qubit of $8.3 \times 10^{-6}$.

\begin{table}[h!]
\begin{ruledtabular}
\begin{tabular}{lccccc}
Experiment & Infidelity ($\times 10^{-5}$) & Total shots & Leakage accept rate &  Overall accept rate  \\[.3em]  \hline \\[-.6em]
Unencoded & $48(6)$ & $4000$ & -- & --\\[.1em]
$\I{48}$, two-branch & $6^{+3}_{-2}$ & $2000$ & -- & $.942(5)$\\[.1em]
$\I{48}$, log-depth & $3^{+3}_{-1}$ & $2000$ & -- & $.943(5)$\\[.1em]
$\I{48}$, two-branch, $S_X$ extraction & $31^{+9}_{-7}$ & $2000$ & -- &$.47(1)$\\[.1em]
$\I{48}$, log-depth, $S_X$ extraction & $13^{+6}_{-4}$ & $2000$ &-- & $.49(1)$\\[.1em]
Unencoded, LH & $170(10)$ & $4000$ & $.82(3)$ & $.82(3)$\\[.1em]
$\I{48}$, log-depth, LH & $12^{+3}_{-2}$ & $4000$ & $.916(4)$ & $.829(6)$\\[.1em]
$\I{48}$, log-depth, LH, $S_X$ extraction & $21(4)$ & $4000$ & $.907^{+4}_{-5}$ & $.695(7)$\\[.1em]
\end{tabular}
\end{ruledtabular}
\caption{Logical distance-2 and unencoded SPAM experiment results. Here, and elsewhere, LH refers to the leakage-heralded measurement available on Helios (Methods~\ref{sec:Helios}). Infidelity is reported per logical qubit as in Fig.~\ref{fig:encoding} as the fraction of total shots with no logical error, divided by number of logical qubits, and (leakage) acceptance rate is computed as the probability of preparing $48$ qubits with no detected (leakage) error in both the encoded and unencoded case. Uncertainties for all encoded experiments in this work are obtained from Wilson 68\% confidence intervals unless otherwise indicated, and are reported with parenthetical notation when symmetric and $\mbox{}^+_-$ notation where asymmetric, with values referring to the trailing significant digits.}
\label{tab:spam}
\end{table}

\section{GHZ-state based syndrome extraction and readout in $\I{k}$ codes} \label{sec:ghz_se}

The standard gadgets for $d=2$ syndrome extraction presented in Refs.~\cite{Self:2024py, iceberg-beyond-the-tip} use two ancilla qubits to extract the $S_Z$ and $S_X$ syndromes non-destructively in the middle of the circuit. These gadgets are made FT by careful ordering of the SE gates so that each ancilla serves as a flag for errors that may occur in extracting the syndrome measured by the other ancilla. This assumes that any nontrivial measurement of the corresponding ancilla qubit results in that circuit's results being discarded via postselection, and it is not possible to distinguish errors that occurred during the rest of the circuit that resulted in a nontrivial syndrome from flagged errors that occurred during the SE gadget itself. More importantly, the gates involved in these gadgets are highly serial, with at most two performed simultaneously throughout the gadget. This presents a challenge for efficient qubit transport on Helios, both because it can parallelize batches of up to 8 two-qubit gates at a time in its gate zones and because this type of ``two-to-all" circuit structure requires repeatedly moving the same ancilla qubits back and forth from the ring storage to the quantum logic regions, which may result in large time overheads.

\begin{figure}[h!]
\centering
\includegraphics[width=\textwidth]{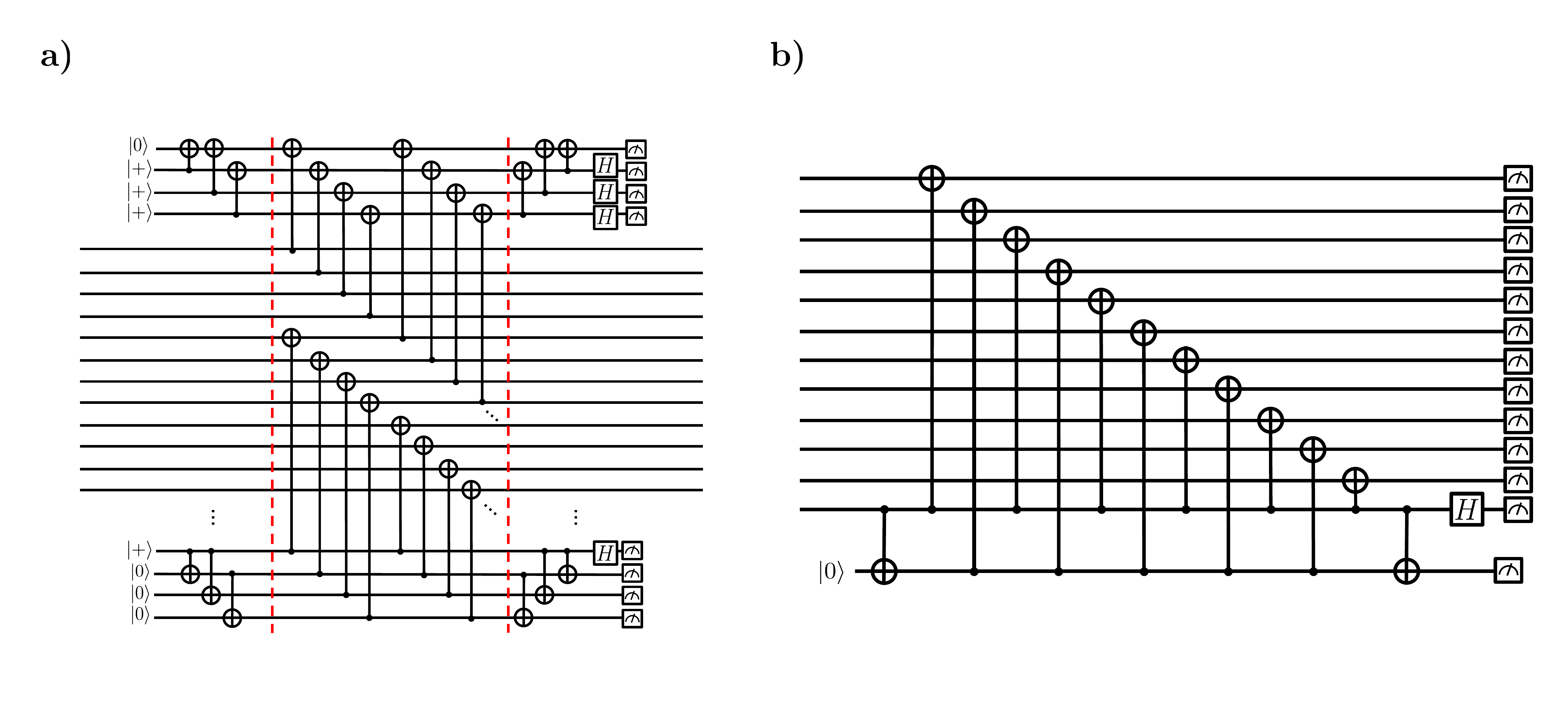}
\caption{\textbf{Syndrome extraction and readout gadgets in the $\I{k}$ codes} \textbf{(a)} GHZ-based syndrome extraction illustrated on an $\I{10}$ code. Both four-qubit ancillary GHZ states read in parities across all qubits; the lower ladder of CX gates wraps around to include the top four physical qubits as well. \textbf{(b)} Readout gadget used in $\I{94}$ GHZ state preparation experiments. The information encoding the measured syndromes is described in the text below.
}\label{fig:se_protocols}
\end{figure}

These standard gadgets are designed with efficiency of qubit number required in mind. However, for nearly all $d=2$ experiments in this work, there is a surplus of qubits in the device that we may employ for mitigation of errors. The parallelization of the standard gadgets can already be improved simply by assigning a flag qubit for each ancilla rather than using them as flags for each other. This allows twice as many gates to be implemented in parallel, since each ancilla and its corresponding flag are prepared in a Bell pair and syndrome extraction gates can be broadcast over both qubits in the Bell pair. 

In this work we take this idea further, noting that we can achieve further parallelization of SE gates by entangling ancillas into GHZ states in the appropriate basis; gates can then be broadcast over all qubits in the ancillary GHZ states and one qubit detects the syndrome while the remaining qubits flag any potential error that may occur during the gadget. In Fig.~\ref{fig:se_protocols}a we illustrate this scheme on an $\I{10}$ code. 
Using four-qubit GHZ states takes full advantage of the gate-zone parallelization on Helios while introducing minimal delays for encoding and decoding of the ancillary states themselves. We expect this is close to the optimal size GHZ state for performing SE on Helios, because larger states will introduce larger encoding/decoding delays, will be more sensitive to coherent dephasing errors on the ancillas, and the gate zones are already being fully utilized in parallel when 8 gates are performed simultaneously (larger states will require additional transport rounds to perform more gates). However, one can take the GHZ states to be of any intermediate size up to potentially $n$ (the Steane-style syndrome extraction limit \cite{Steane_1997}) which may be useful on future devices.

To be explicit, we use an $X$-basis GHZ state to detect $X$ errors (detected by $S_Z$) and vice versa with a $Z$-basis GHZ state as shown in Fig.~\ref{fig:se_protocols}a. We arrange the gates that encode/decode the GHZ states in a careful order as shown, so that all errors which can flip the syndrome propagate back to the same ancilla (i.e., the qubit initialized in $|0\rangle$ for the $X$-basis GHZ state and vice versa), whose measured value records the relevant syndrome. All other qubits flag any hook errors that might occur (i.e., a $Z$ error on the $X$-basis GHZ state from gates or memory errors that could propagate back onto the system), and we correspondingly discard any shot where any qubit in either GHZ state records a nontrivial measurement. We note that when the DFS-based error mitigation discussed in Sec.~\ref{sec:state_prep_gadgets} is used, the code space is defined by the $S_Z = -1$ subspace so the expected noiseless value of the relevant ancilla measurement is flipped.

Readout proceeds similarly to syndrome extraction, but the $S_Z$ syndrome does not need to be extracted non-destructively; it can be reconstructed by the destructive measurement of each qubit in the $Z$ basis. As a result, we only extract $S_X$ non-destructively during readout, using a four-qubit $Z$-basis GHZ state. In general, it is not clear that extracting $S_X$ at the end of a circuit is always desirable, as the gadget itself can introduce coherent idling errors in the $Z$ basis that do not affect the destructive measurements of the qubits but suppress the acceptance rate since they are detected by $S_X$ (and potentially introduce extra logical errors). For the state preparation and measurement (SPAM) experiments that are described in the main text and reported in Fig.~\ref{fig:encoding}b and Fig.~\ref{fig:encoding}c, we report the fidelity of the distance-2 experiments both with and without terminal $S_X$ extraction. It is clear from the results reported there that terminal $S_X$ extraction appears to typically increase both discard rates and logical error. For this reason, we do not extract $S_X$ at the end of the circuit in our $XY$ model experiments.

The results reported in Fig.~\ref{fig:encoding}b and Fig.~\ref{fig:encoding}c also include some experiments that used the leakage-heralded measurement available on Helios (Methods~\ref{sec:Helios}). As expected, those experiments also show slightly larger logical error rates and discard rates, while the leakage-heralded measurement substantially improves the results presented in Fig.~\ref{fig:computing}. This is primarily due to the fact that the leakage-heralded measurement has larger associated physical SPAM error but the SPAM experiments, being relatively short circuits, do not result in much qubit leakage. Consequently, we expect leakage-heralding not to enhance logical error in these very short circuits. 

During both syndrome extraction and readout, we employ additional dynamical decoupling pulses following the same method as discussed for state preparation in Sec.~\ref{sec:state_prep_gadgets}, applying $\pm X$ pulses of alternating sign to idling qubits during the gadget, with an even parity of pulses applied on any given qubit before and after SE gates to avoid affecting the code stabilizers. In the SPAM experiments of Fig.~\ref{fig:encoding}b,c, these pulses are only applied to the experiments that extracted $S_X$. The full impact of these DD pulses is not clearly resolved and may require further study to refine pulse timing.

In the $\I{94}$ GHZ state preparation experiments, there are not sufficiently many qubits on the device to use the log-depth initialization gadget and the GHZ-based readout protocol. For this reason, we default to the two-branch initialization presented in Fig.~\ref{fig:halfdepth}, and we use the readout gadget shown in Fig.~\ref{fig:se_protocols}b, which further improves on the gadget of Ref.~\cite{iceberg-beyond-the-tip} and requires only one ancilla qubit. In terms of the physical operators measured at the end of this gadget, we still have $\logZ^i = Z_{i} Z_{n-1}$; however, the stabilizers are now encoded into the operators $S_X = Z_0$ and $S_Z = \prod_{i=1}^{n-1} Z_i = Z_1Z_2 \ldots Z_k Z_{n-1}$. Any nontrivial measurement of the ancilla flags errors that occur during the gadget.

\section{Concatenated iceberg QEC cycle details}\label{sec:QEC}

In this section, we give fault-tolerance proofs, decoding algorithms, and experimental data from Helios for the QEC cycle for $d=4$ concatenated iceberg codes introduced in the main text. Recall that in the $\I{k_2}\circ\I{k_1}$ code, stabilizers consist of rows and pairs of columns, see Fig.~\ref{fig:encoding}a. We call the former low-level stabilizers, denoted by $S_{X/Z, i} = \prod_{j=0}^{n_1} (X_{ij}/Z_{ij})$ for $i \in \{0, ... , n_2-1\}$, since they are stabilizers of the $d=2$ $\I{k_1}$ codes and the latter, high-level stabilizers, denoted by $S_{X/Z}^{i} = \prod_{j=0}^{n_2} (X_{0,j}X_{i,j}/Z_{i,j}Z_{n_1-1,j})$ for $i \in \{1, ...,n_1-2\}$ since they are the stabilizers of the encoded $\I{k_2}.$

The full QEC cycle for $\I{k_2} \circ \I{k_1}$ is given in Fig.~\ref{fig:General QEC}. To summarize, a $\ket{\bar{+}}^{\otimes k_1}$ state encoded in an $\I{k_1}$ ancilla is used to extract the high-level $X$ stabilizers, $S_X^{i}$. The qubits used to flag this ancilla are unencoded, initialized in the $\ket{0}^{\otimes n_1}$ state. Then the lower-level stabilizers $S_{X,j}$ are extracted using Bell pair $d=2$ syndrome extraction (denoted as EDZ in Figs.~\ref{fig:General QEC} and \ref{fig:encoding}d), the circuit for which is shown in Fig.~\ref{fig:se_protocols}b. The second half of the QEC cycle is conjugate to the first via the Hadamard gate and extracts the $S_Z^{i}$ and $S_{Z,j}$ syndromes.

Before discussing decoding the QEC cycle under circuit-level noise, we give a decoding algorithm in the case of correcting input errors with a noiseless QEC cycle. In this case, we must correct a single input error and halt if uncorrectable errors of weight $> 1$ occur. The intuition behind decoding in this case is that the high-level stabilizers, corresponding to pairs of columns, will give information about the $x$ coordinates of errors, and the low-level stabilizers will give information about the $y$ coordinates. If a weight-one error occurs, we can deduce the $x$ and $y$ coordinates of the error from the stabilizers and correct it. If a detectable error occurs, we can deduce from the stabilizers that there will either be multiple $x$ coordinates, or multiple $y$ coordinates, or both. Pseudocode for correcting Z errors given the X stabilizers is shown in Fig.~\ref{fig:Noiseless Dec}. Correcting X errors given the Z stabilizers is analogous.

We see that strings of vertical and horizontal errors can be differentiated from correctable errors by the presence of multiple $y$ coordinates or multiple $x$ coordinates in Algorithm 1. We will refer to such errors as ``generalized weight-two,'' due to both their detectability as well as to the size of their common support with any of the weight-four logical operators being  $\leq 2$ (because lines intersect rectangles on at most two vertices).

\begin{definition}\label{def:gen wt 2}
    We call $X$ error strings of the form $\prod_{j \in S} X_{x,j}$ and $\prod_{i \in S} X_{i,y}$ generalized weight-two. Similarly, we call $Z$ error strings of the form $\prod_{j \in S} Z_{x,j}$ and $\prod_{i \in S} Z_{i,y}$ generalized weight-two. A product of a generalized weight-two $X$ error and a generalized weight-two $Z$ error is also generalized weight-two, due to independent correction of $X$ and $Z$ errors.
\end{definition}

Additionally, suppose we know that a vertical or horizontal string of errors has occurred, and, moreover, we know the common $x$ or $y$ coordinate of the string. Algorithm 1 will then compute the $y$ or $x$ locations of the string respectively, and this error can be corrected. Therefore, we call such strings for which the common coordinate is known ``generalized weight-one'' errors. There is one additional $X$ error which is ``generalized weight-one:'' an $X$ error on all the qubits of a column. Since pairs of columns are stabilizers, these are all equivalent up to a stabilizer: $\prod_{j \in [n_2]}X_{x_1,j} \sim \prod_{j \in [n_2]}X_{x_1,j}$. Therefore, if in Algorithm 1, the length of the array `Ycoords' is $n_2$, we may apply $\prod_{j \in [n_2]}X_{0,j}$ and this will correct an error string on any column.

\begin{definition}\label{def:gen wt 1}
    We call the tuple $(\prod_{j \in S} X_{x,j},\, x)$, representing a vertical error string with a known $x$ coordinate generalized weight-one. Similarly, we call the tuple $(\prod_{i \in S} X_{i,y}, y)$, representing a horizontal error string with a known $y$ coordinate generalized weight-one. Finally, we also denote the vertical error string $\prod_{j \in [n_2]}X_{x_1,j}$ with $x_1$ arbitrary generalized weight-one. Similarly, $(\prod_{j \in S} Z_{x,j}, x)$, $(\prod_{i \in S} Z_{i,y}, y)$, and $\prod_{j \in [n_2]}Z_{x_1,j}$ are generalized weight-one. A product of a generalized weight-one $X$ error and a generalized weight-one $Z$ error is also generalized weight-one, due to independent correction of $X$ and $Z$ errors.
\end{definition}

\begin{figure}[!h]
\centering
\includegraphics[width=0.95\textwidth]{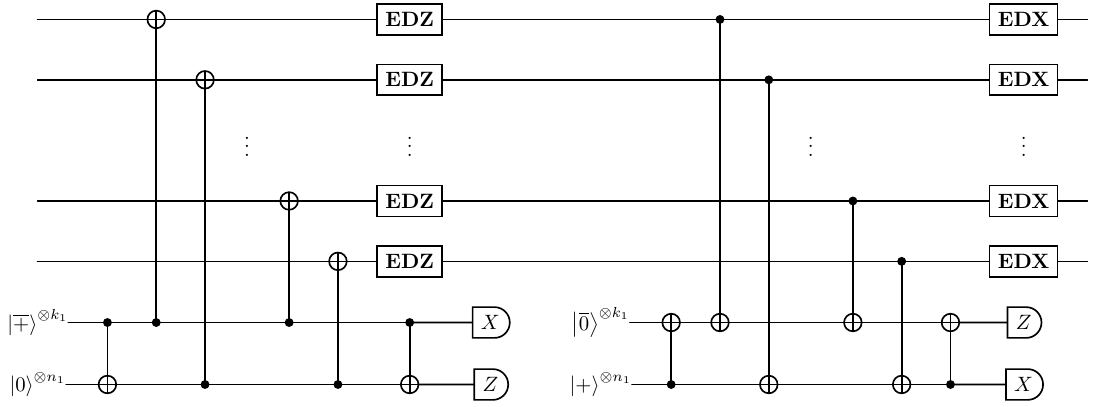}
\caption{
{\textbf{The QEC cycle for $\mathbf{\I{k_2}\circ\I{k_1}}$.} The first half of the QEC cycle extracts $X$ stabilizers and the second half extracts $Z$ stabilizers. EDZ refers to $d=2$ X syndrome and EDX to $d=2$ Z syndrome extraction respectively. All CX gates shown are transversal, including the ones between the bottom two ancilla wires. The top ancilla is an $\I{k_1}$ code block and the bottom is $n_1$ physical qubits, so they have the same size, and ``transversal" in this case still means $n_1$ physical CX gates.}
}\label{fig:General QEC}
\end{figure}

We will prove versions of the fault-tolerance criteria in \cite{Prabhu_2024} for $d=4$ codes, broadened to take into account the notions of generalized weight-1 and -2 errors, as well as the handling of measurement errors, which are natural for this QEC cycle. As a consequence, the conditions we prove are weaker in the sense of allowing more general output errors, but stronger in that the QEC cycle must correct more general input errors. To demonstrate that these conditions are sufficient for fault-tolerance, we will show that they imply the desired $d=3$ QEC and $d=4$ QED ExRec conditions (an ExRec or ``extended rectangle" is a region of a circuit containing two adjacent QEC cycles; see Definition 3). Namely, for $d=3$, we show that an encoded computation using an $\I{k_2} \circ \I{k_1}$ code will succeed so long as there is at most one fault per ExRec. Furthermore, we show for $d=4$, that if there are at most two faults per ExRec, the computation will either succeed or halt, but will not result in an undetectable logical error. 

\begin{proposition}\label{prop: qec correctness}
The QEC cycle for two-level concatenated iceberg codes satisfies the following properties:
\begin{enumerate}
    \item If the QEC cycle contains no fault, it will take an input with a generalized weight-1 (or weight-0) error to an error-free output, and will halt the computation in the case of a generalized weight-2 input.
    \item  Suppose the QEC cycle contains a single internal fault. Then it will take an input with no errors to an output with error of generalized weight at most 1. Also, it will take an input with a generalized weight-1 error to a generalized weight-1 output, or if it does not do so, it will pass flags to the next QEC cycle which will cause it to halt.
    \item If the QEC cycle contains at two internal faults, it will take an error-free input to an output of at most generalized weight-2, or it will halt.
\end{enumerate}

\end{proposition}

\begin{proof}
    Let us discuss the correction of $Z$ errors as correcting $X$ errors is analogous. The proof of 1.~follows from Algorithm 1 since, in this case, the extraction of the stabilizers is noiseless. Also, it is clear that, in the absence of noise, Algorithm 1 will correct a weight-one $Z$ error. In order to proceed further, we must analyze and correct circuit-level noise in terms of generalized weight-one and weight-two errors. $Z$ hook errors can only occur in the second half of the QEC cycle. $Z$ errors on the iceberg ancilla for extracting the $S_Z^i$ or its flags will lead to vertical $Z$ strings. An example is shown in Fig.~\ref{fig:Generalized weight example}, where an error on the flag for the first qubit of the iceberg ancilla will lead to a vertical string of $Z$ errors with $x$ coordinate 0. In the absence of further errors, the $0$th flag will have a non-trivial measurement, and this $x$ coordinate is known, making the error generalized weight-one. This information can be passed to the next QEC cycle, which can then correct this error (this is the ``HighZFlags'' measurement register in Algorithms 2 and 3, shown in Figs.~\ref{fig:Gen weight 1 alg} and \ref{fig:Circuit-level Dec}).
    
    Any single internal error, for it to be a hook error of $Z$ type, must affect an ancilla in the second half of the QEC cycle and will therefore cause a non-trivial flag measurement, and lead to an error of at most generalized weight-1 if the input is error-free, proving the first half of 2. To prove the second half, we must discuss what happens in the case of measurement errors in the $X$ stabilizer registers (``HighX" and ``LowX'') which, ideally, will correct generalized weight-1 input errors. If the error is strictly weight-1, in the absence of measurement errors, Algorithm 2 would deduce the $x$ and $y$ coordinate of such an error and be able to correct it. If a single measurement error occurs, then one of these two coordinates may be missing. In this case, we allow the error to pass through to the next QEC cycle, but, in Algorithm 3, we also pass the flags ``CurrInternalFaults''' and ``NextInputErrorAllowed" to the next QEC cycle and postselect if there are subsequent internal faults or input errors respectively. If the error is generalized weight-1 with a flagged $x$ or $y$ coordinate from the previous QEC cycle, and a single internal-fault in the current QEC cycle introduces any uncertainty in any of the $y$ or $x$ coordinates of the string respectively, we can postselect, as we know two errors in this ExRec have occurred (the hook error from the previous QEC cycle and the current internal fault). We can tell if such an internal fault has occurred by looking at, in the case of a vertical hook error, the parity of the number of detected $y$ coordinates, which should be accompanied by a detected $x$ error coordinate in the case that this parity is odd, or no detected $x$ error coordinate, in the case that this parity is even (see the comments ``Halt if there is another error in addition to the hook error" in Algorithm 3). We computationally enumerate all single internal-faults and all generalized weight-1 inputs in the case of the $\I{2} \circ \I{2} = \nkd{16}{4}{4}$ QEC cycle, and verify that, in all cases, if a generalized weight-1 input is turned into a higher weight output, the subsequent QEC cycle will halt, completing the proof of 2, as the decoding for other sizes of iceberg layers is entirely analogous.
    
    Finally, if a hook error occurs, a subsequent internal fault, such as a measurement error in the measurement of flag $0$ in Fig.~\ref{fig:Generalized weight example} can turn a generalized weight-1 vertical error, for instance, into an error where the $x$ coordinate of this error string is not detected. In this case, the error becomes generalized weight-two. We can verify 3.~through an exhaustive enumeration of all two $s$ internal faults for $r+s = 2$ in the case of the $\I{2} \circ \I{2} = \nkd{16}{4}{4}$ QEC cycle.
\end{proof}

Note that Prop.~\ref{prop: qec correctness} implies that the scaling in Fig.~\ref{fig:QEC X Basis Scaling Plot} for circuit-level simulations of the $\I{2}\circ\I{2}$ QEC cycle should have an $O(p^3)$ infidelity and $O(p^2)$ rejection rate, consistent with the numerically fitted slopes.

Now, we prove that Prop.~\ref{prop: qec correctness} implies the $d=3$ and $d=4$ ExRec conditions.

\begin{definition}
    In a memory experiment with repeated QEC cycles, we define an ExRec to be any pair of consecutive cycles. An ExRec is good to distance 3 if it contains $\leq 1$ errors.
\end{definition}

\begin{proposition}
    Suppose $\ket{\overline{\psi}}$ is in the code space of $\I{k_2} \circ \I{k_1}$ and is input to $m$ rounds of QEC. If all ExRecs are good to distance 3, then the computation will not halt and the output will have an error of generalized weight at most one.
\end{proposition}

\begin{proof}
    First, suppose that all ExRecs are good to distance three. Then the success of the computation follows from Prop.~\ref{prop: qec correctness} since QEC cycles with errors are interleaved with perfect QEC cycles, which will correct any generalized weight-one errors produced by the faulty QEC cycles.
\end{proof}

Now, we prove the $d=4$ ExRec condition for $Z$ errors (the proof for $X$ errors is analogous). Again, suppose the computation is $m$ rounds of QEC and the input is in the code space of $\I{k_2}\circ\I{k_1}.$

\begin{definition}
    An ExRec is good to distance 4 if it contains $\leq 2$ errors.
\end{definition}

\begin{proposition}
If all ExRecs are good to distance 4, then the computation will either halt or the output will have an $X$ error of generalized weight at most two.
\end{proposition}

\begin{proof}
    If all ExRecs are good to distance 4, then we examine the first ExRec where after the first QEC gadget, there is an error of generalized weight-two. There are two ways in which this can occur (by Prop.~\ref{prop: qec correctness}). The first is that the input to the ExRec was error-free and there were two errors in the first QEC cycle. In this case, the second QEC cycle must be error-free for the ExRec to be good, and the computation will halt because this noiseless QEC cycle has a weight-two input error. The second possibility is that the input to the ExRec had a generalized weight-one error and there was a single internal fault in the first QEC cycle. In this case, since after the first QEC cycle of the ExRec, the error is no longer generalized weight $\leq 1$, by Prop.~\ref{prop: qec correctness}, the second QEC cycle will halt, completing the proof.
\end{proof}

We note that while we prove the ExRec conditions here for memory experiments, we believe they could be straightforwardly generalized to, for example, a transversal CX gate, by looking at all the flags in both inputs to the $CX$ gate and halting if their combined non-trivial measurements are too numerous (for instance, if both had a ``CurrFailures" value of 1 in Algorithm 3).

For the $\I{2} \circ \I{2}$ and $2 \times \I{2}$ codes, we performed memory experiments on Helios where $c$ cycles of QEC/QED were performed for $c \in \{1,2,4,8\}$. Leakage heralding was used and the leakage acceptance, preacceptance, and QEC/QED acceptance rates are given in Table~\ref{tab:k=4 qec results}. We decode over multiple ExRecs using Algorithm 3 (and analogous decoding for $X$ errors) with the extra postselection condition that correctable hook errors must occur on every other $x$ or $y$ coordinate due to the fact that the ancillas in syndrome extraction interact with every other qubit. We note that the preacceptance rate, which denotes the fraction of all unleaked shots for which all the state preparations in the protocol succeed, could be made arbitrarily close to 1 through repeat-until-success, and leakage postselection could be solved through the conversion of leakage noise to depolarizing noise in hardware using the technique of Ref.~\cite{Hayes_2020}. Therefore, it is the QEC acceptance rate, which should scale as $O(p^2)$, which is fundamentally the limiting factor and must be increased by scaling the code distance, as we demonstrate. For the single $\I{8}\circ \I{6}$ QEC cycle experiment, the decoding problem over a single ExRec is significantly easier, as we can simply iterate over all possible single failures and apply corrections whenever the flags and syndromes are compatible with such failures, and postselect otherwise. The results for this experiment are given in Table~\ref{tab:k=48 qec results}.

\begin{table}[!h]
\begin{tabular}{lcllll}\hline \hline\\[-.25cm]
\multicolumn{1}{c}{Experiment} & \multicolumn{1}{c}{Cycles} & \multicolumn{1}{c}{Leakage acceptance} & \multicolumn{1}{c}{Preacceptance} & \multicolumn{1}{c}{QEC/QED acceptance}  \\[.3em]  \hline \\[-.6em]
$\I{2} \circ \I{2}$ (X Basis) & [$1$, $2$, $4$, $8$] & [$.930_{-9}^{+8}$, $.89(1)$, $.81(1)$, $.67_{-2}^{+1}$ ] & [$.88_{-2}^{+1}, .83_{-2}^{+1}, .88(1), .84_{-2}^{+1}$] & [$.989^{+3}_{-4}, .970^{+6}_{-7}, .90(1), .75(2)$]\\[.1em]

$\I{2} \circ \I{2}$ (Z Basis) & [$1$, $2$, $4$, $8$] & [$.917_{-9}^{+8}$, $.89(1)$, $.77(1)$, $.64_{-2}^{+1}$ ] & [$.87(1), .84(1), .80^{+1}_{-2}, .82_{-1}^{+2}$] & [$.991^{+3}_{-4}, .952^{+7}_{-8}, .88_{-3}^{+1}, .76(2)$]\\[.1em] 

$2 \times \I{2}$ (X Basis) & [$1$, $2$, $4$, $8$] & [$.983_{-5}^{+4}, .984(4), .967_{-6}^{+5}, .937_{-8}^{+7}$] & [$.993^{+2}_{-3}$, $.991_{-4}^{+2}$, $.993_{-4}^{+2}$, $.996_{-3}^{+1}$ ] & [$.960^{+6}_{-7}, .912_{-10}^{+8}, .71_{-2}^{+1}, .64(2)$]\\[.1em]

$2 \times \I{2}$ (Z Basis) & [$1$, $2$, $4$, $8$] & [$.983_{-5}^{+4}, .984(4), .967_{-6}^{+5}, .937_{-8}^{+7}$] & [$.993_{-3}^{+2}, .990_{-4}^{+3}, .981_{-5}^{+4}, .994_{-4}^{+2}$] & [$.960_{-7}^{+6}, .917_{-9}^{+8}, .74_{-2}^{+1}, .50_{-2}^{+1}$ ]\\[.1em] \hline \hline

\end{tabular}
\caption{Acceptance rates for $k=4$ memory experiments over $1,2,4,$ and $8$ cycles of QEC or QED. 1000 shots were submitted for each cycle length, for the $\ket{\overline{0}}^{\otimes 4}$ and $\ket{\bar{+}}^{\otimes 4}$ states for the $\I{2} \circ \I{2}$ and $2 \times \I{2}$ codes. In all experiments, leakage heralding was used and shots with detected leakage events were discarded. Leakage acceptance is the fraction of shots for which no detected leakage errors occurred. The preacceptance is the fraction of leakage postselected shots for which all state preparations (both for the input states and, in the case of the concatenated code, the lower-level ancilla used in QEC) succeeded. The QEC/QED acceptance rates are computed as the fraction of unleaked, preaccepted shots for which the QEC/QED cycles do not halt due to uncorrectable configurations of $d/2$ circuit failures. In the accepted shots, one error occurred in the $4$-cycle $\I{2}\circ\I{2}$ X basis experiment which induced an error on one of the logical qubits. No errors were found in any of the other $k=4$ memory experiments.}
\label{tab:k=4 qec results}
\end{table}

\begin{table}[!h]
\begin{tabular}{lcccclcc}\hline \hline\\[-.25cm]
\multicolumn{1}{c}{Experiment} & \multicolumn{1}{c}{Leakage acceptance} & \multicolumn{1}{c}{Preacceptance} & \multicolumn{1}{c}{QEC acceptance} & Accepted shots & Errors \\[.3em]  \hline \\[-.6em] 
$\I{8} \circ \I{6}$ (X Basis) & .431(7) & $.57(1)$ & $ .58_{-2}^{+1}$ & 708 & \quad 0 \\[.1 em]
$\I{8} \circ \I{6}$ (Z Basis) & .453(7) & $.64(1)$ & $ .66(1)$ & 958 & \quad 0 \\[.1 em]\hline \hline

\end{tabular}
\caption{Experimental data for the $\I{8}\circ\I{6}$ QEC cycle experiments. 5000 shots were submitted in each basis. Leakage, preacceptance and QEC acceptance rates are as before. Of the 708 accepted shots in the X basis and 958 accepted shots in the Z basis, no errors occurred.}
\label{tab:k=48 qec results}
\end{table}

\begin{figure}[!t]
\centering
\includegraphics{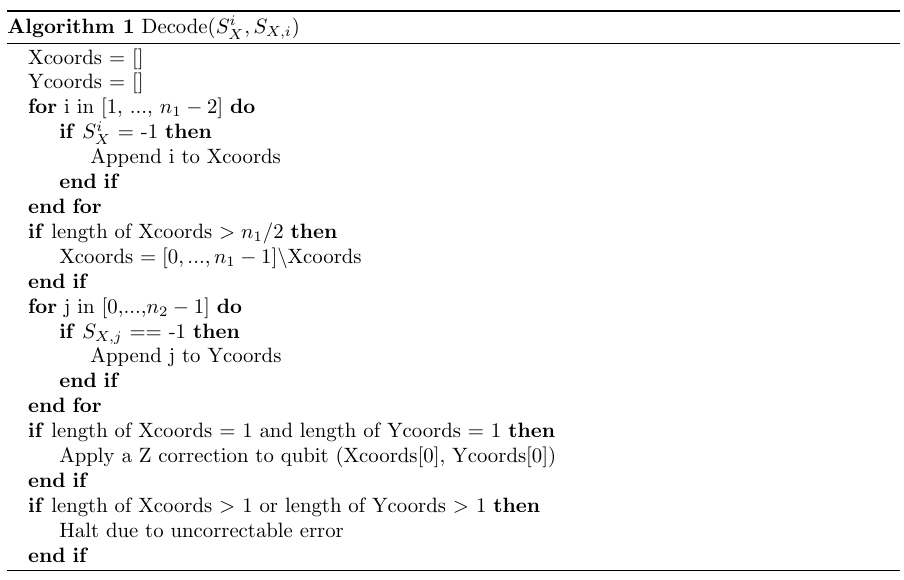}
\caption{
{\textbf{Noiseless decoding of $X$ stabilizer measurements.} Given two arrays of noiseless X stabilizer measurements $S_X^i$ and $S_{X,i}$, in the presence of non-trivial syndromes, this algorithm will either return a compatible weight-one Z correction or determine that an uncorrectable error has occurred.}
}\label{fig:Noiseless Dec}
\end{figure}

\begin{figure}[!h]
\centering
\includegraphics[width=0.95\textwidth]{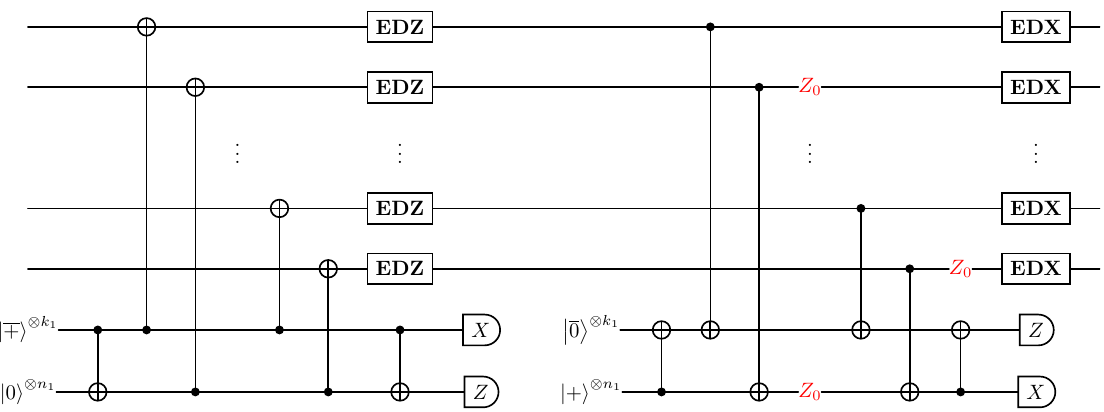}
\caption{
{\textbf{Error propagation during high-level $Z$ syndrome extraction.} An example of a $Z_0$ error on the ancilla during high-level $Z$ syndrome extraction spreading to a vertical string of $Z_0$ errors on the data iceberg code blocks.}
}\label{fig:Generalized weight example}
\end{figure}

\begin{figure}[!h]
\centering
\includegraphics{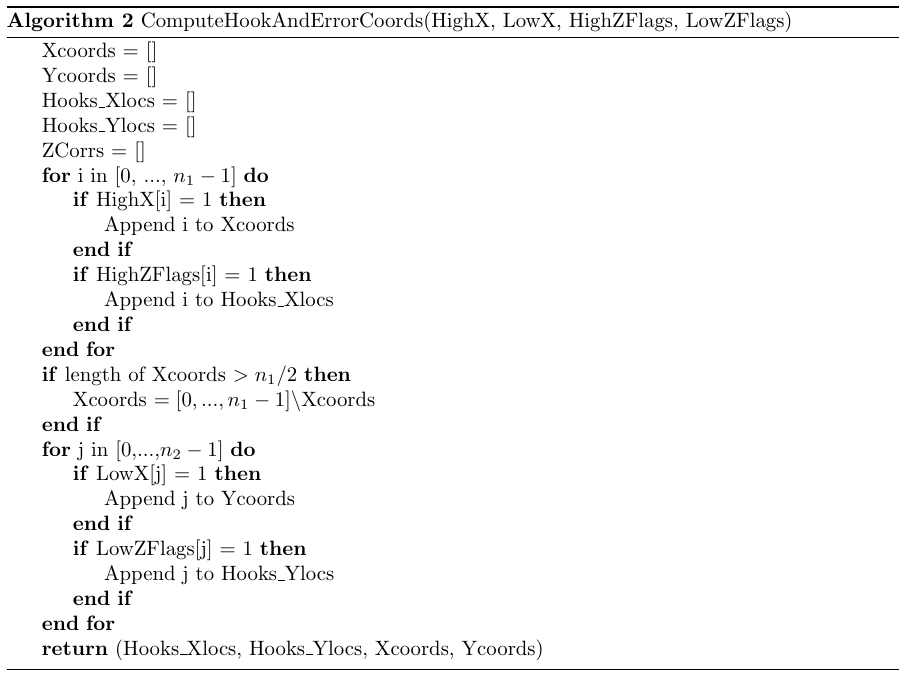}
\caption{
{\textbf{Helper Function for decoding circuit-level $Z$ errors from QEC Cycle Measurements.} Given the four relevant (possibly faulty) measurement registers from the current and previous QEC cycle, this algorithm computes the possible hook error locations as well as non-trivial $x$ and $y$ coordinates of Z errors. The HighX and LowX measurement registers are from the transversal $X$ measurement of  $\ket{\bar{+}}^{\otimes k_1}$ and the stabilizer measurements of the EDZ gadgets in the first half of the QEC cycle respectively. The HighZFlags and LowZFlags are the flags for the $\ket{\bar{0}}^{\otimes k_1}$ ancilla and the flags from the EDX gadgets in the second half of the previous QEC cycle, which record the locations of the vertical and horizontal $Z$ hook errors that the current QEC cycle must correct.}
}\label{fig:Gen weight 1 alg}
\end{figure}

\clearpage

\begin{figure}[h!]
\centering
\includegraphics[scale = .85]{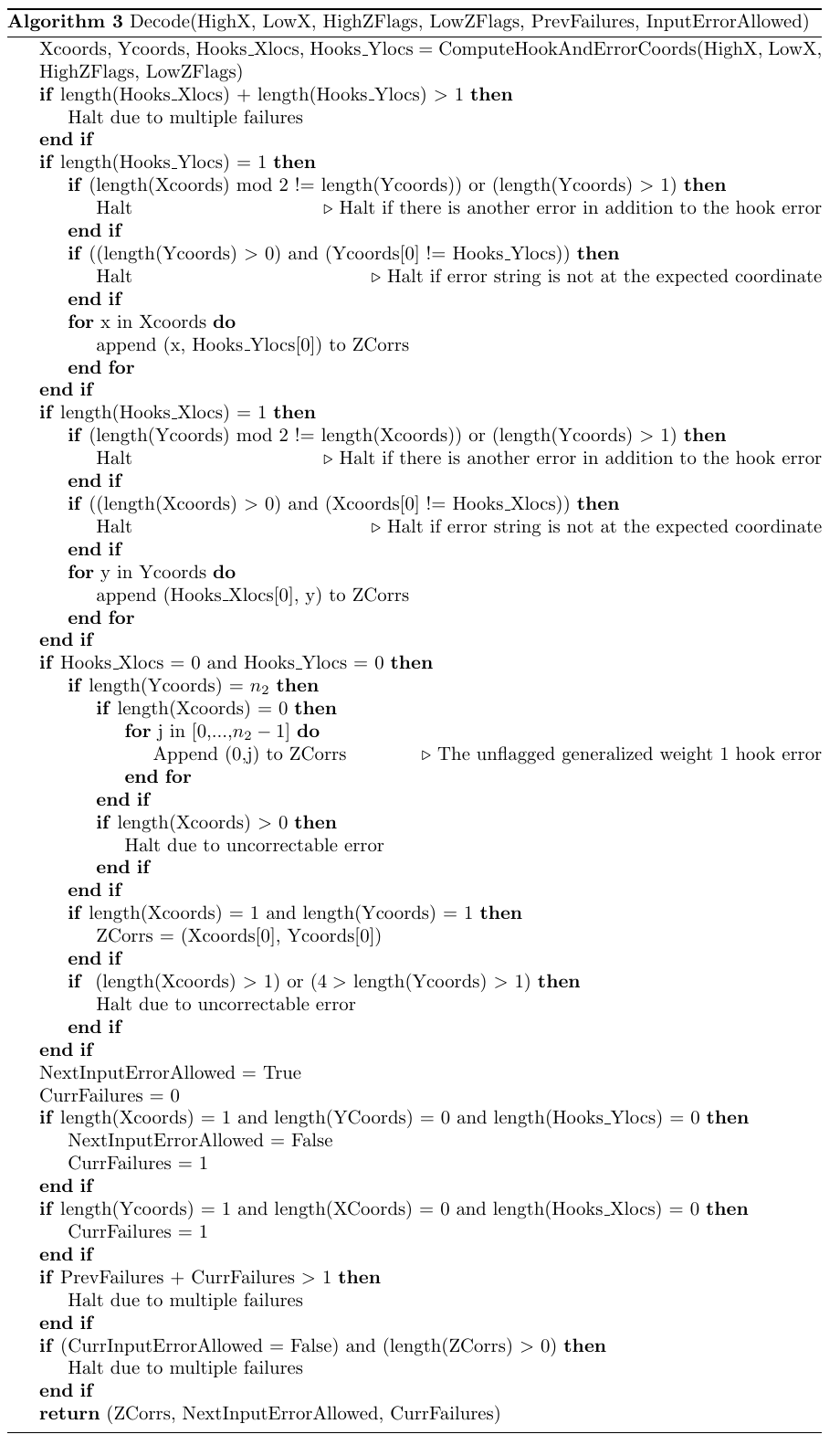}
\caption{\textbf{Circuit-level decoding for the $d=4$ concatenated iceberg QEC cycle.} The decoder for $Z$ errors which satisfies Prop.~\ref{prop: qec correctness}. `PrevFailures' and `InputErrorAllowed' are passed from the previous QEC cycle, and the current QEC cycle will return the analogous quantities `NextInputErrorAllowed' and `Curr Failures' to pass to the next QEC cycle, as well as `Zcorr', the Z corrections for this QEC cycle.}\label{fig:Circuit-level Dec}
\end{figure}

\clearpage

\begin{figure}[!h]
\centering
\includegraphics[width=1.0\textwidth]{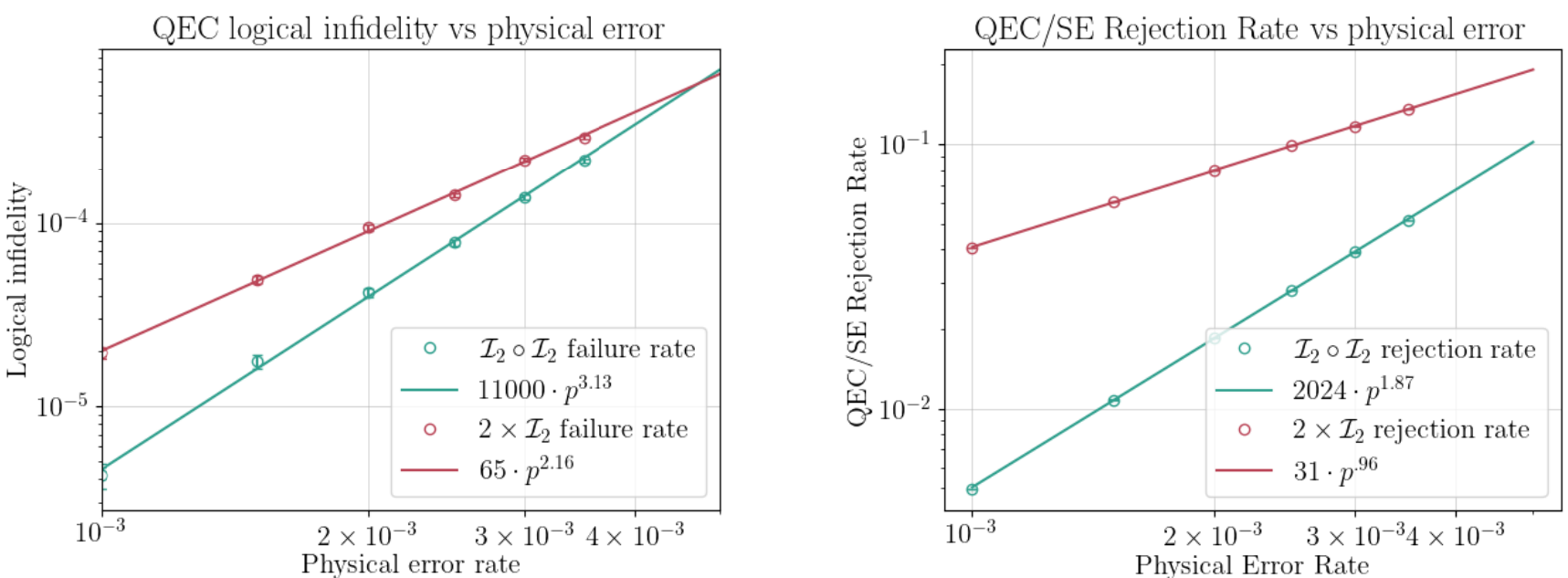}
\caption{
{\textbf{Simulated QEC infidelity and reject rate as a function of physical noise rate $p$.} Circuit-level \texttt{Stim} \cite{Gidney:2021wm} simulations with two-qubit gate errors and measurement errors of a single round of QEC on the $\ket{+}^{\otimes 4}$ input state. The properties in Prop.~\ref{prop: qec correctness} imply that the QEC rejection rate should scale as $O(p^2)$ and the infidelity as $O(p^3)$. Also shown are the rejection rates and infidelities for two copies of the $\I{2}$ code. We can see that for an effective two-qubit infidelity of $\leq 4\times 10^{-3}$, the fidelity and acceptance rate for $d=4$ are improved relative to $d=2$.}
}\label{fig:QEC X Basis Scaling Plot}
\end{figure}

\section{Logical gate benchmarking analysis and statistical procedure} \label{sec:si_lgb}

Here, we provide more details about the logical cycle-benchmarking protocol described in the main text and in Methods \ref{sec:methods_lcb} and discuss the methods used for fitting and statistical analysis.

\begin{figure}[!h]
    \centering 
    \includegraphics[width=\linewidth]{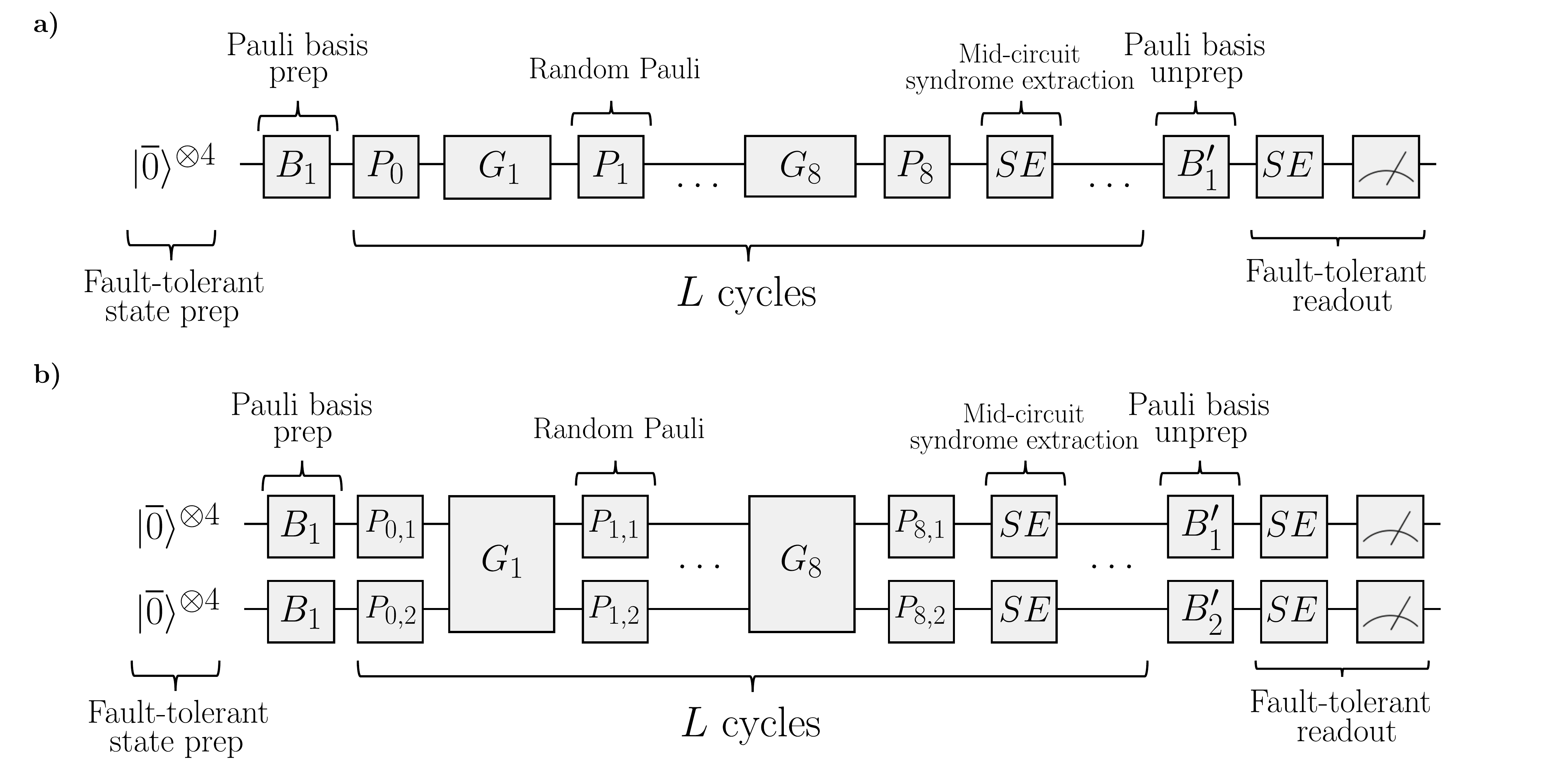}
    \caption{\textbf{Circuit diagrams for logical cycle-benchmarking experiments of $\I{4}$ code blocks.} Both \textbf{(a)} intra-block and \textbf{(b)} inter-block experiments are displayed with cycles consisting of Pauli-twirled 2Q-gate layers and syndrome extraction rounds inserted every 8 cycles.}
    \label{fig:logical-cb-circuit-diagrams}
\end{figure}

The general protocol operates as follows: suppose we have $k$ logical qubits distributed among one or more code blocks on which the logical gates act (i.e., $k = 8$ for two $\I{4}$ code blocks, even if the gates only couple one logical qubit in each block). First, sample a set of logical Pauli observables  $\mathcal{P}_\text{samp}$ from the uniform distribution over $k$-qubit Paulis, and choose a set of cycle depths $\mathcal{L}_\text{cyc}$ with which to perform the experiment. For each corresponding fixed cycle depth $L \in \mathcal{L}_\text{cyc}$, initialize in a random eigenstate $\ket{\psi}_P$ of $P \in \mathcal{P}_\text{samp}$ and apply $L$ Pauli-twirled gates along with periodic syndrome extraction before measuring in the eigenbasis of $P$. 
From the parity of the resulting bitstring, extract an estimate of the expectation value of the output state with respect to the basis of $P$. 
We view the constituent logical cycle of this experiment as being the Pauli-twirled logical gates, with the syndrome extraction essentially just a mechanism for detecting errors that occur during the cycle. However, note that the syndrome extraction can cause up to $k$-qubit logical errors, which is why we sample over Pauli observables of size $k$ correspondingly.

The circuit diagrams for the experiments are displayed in Fig.~\ref{fig:logical-cb-circuit-diagrams}. For the inter-block experiments, we use two $\I{4}$ iceberg code blocks with a single 2Q gate across the blocks ($U_{ZZ}(\pi/2)$ or FANOUT) as the cycle.
For the intra-block experiments, we used a single $\I{4}$ code block and perform two $U_{ZZ}(\pi / 2)$ on both pairs of logical qubits per cycle.
For each set of experiments, we randomly sample 20 Pauli observables across all 4 or 8 logical qubits (for intra-block or inter-block respectively). 
Each shot includes mid-circuit syndrome extraction round after every eight gate cycles and concludes with a terminal SE round. 
This frequency of syndrome extraction was chosen as a rough optimization for logical gate fidelity considering the competing factors of (a) very-frequent syndrome extraction introducing extra memory and gate errors versus (b) infrequent syndrome extraction leading to combinatorially more opportunities for faults to occur per detector region.

In order to obtain an estimate of the Pauli fidelity $f_P$ from the measured Pauli expectation values, we employ a maximum-likelihood based approach. 
For each depth $L\in \mathcal{L}_\text{cyc}$ and Pauli observable $P$, we perform $N$ repetitions of the corresponding logical CB circuit with Pauli layers randomized per shot. 
After postselecting on outcomes without detected errors, we accept $N_L$ shots.
For each accepted shot, we record the parity of the final measured bits to obtain the Pauli eigenvalue outcome
\begin{equation*}
    X_{P,s}(L)\in \{+1,-1\}, \qquad s=1,\dots,N_L.
\end{equation*}
In order to extract the Pauli fidelity, we fit the distribution of $X_{P,s}$ to a Bernoulli model. 
First, define a ``match'' event as when $X_{P,s}(L)$ is equal to the eigenvalue of the initial Pauli eigenstate, $\lambda_{\ket{\psi}}$, and let $K_P(L)$ be the total number of such events across the $N_L$ retained shots.
For fixed $P$ and $L$, we model the observed outcomes as i.i.d.\ Bernoulli trials with probability
\begin{equation*}
    q_P(L) = \Pr\!\big(X_{P,s}(L)=\lambda_{\ket{\psi}}\big).
\end{equation*}
Equivalently, $K_P(L)$ is a binomial random variable with mean $N_L$ and variance $q_P(L)$.
We then relate the Pauli expectation value $\mu_P(L)$ to the ``match'' probability by
\begin{equation*}
    \mu_P(L) = \mathbb{E}[X_{P,s}(L)] = (+1)\cdot q_P(L) + (-1)\cdot[1-q_P(L)] = 2q_P(L) - 1.
\end{equation*}

In cycle benchmarking, the expectation value of each Pauli decays exponentially with the number of cycles $L$ (assuming the probability of each cycle introducing an error is approximately i.i.d.) \cite{Erhard2019}, that is,
\begin{equation*}
    \mu_P(L) = A_P f_P^L,
\end{equation*}
where $f_P \in [0,1]$ represents the Pauli fidelity associated with the effective logical Pauli channel of one cycle, and $A_P\in[0,1]$ is the amplitude parameter that captures the SPAM error for preparing an eigenstate of $P$.
Since $q_P(L) = (1+\mu_P(L))/2$, this gives the binomial success probability
\begin{equation*}
    q_P(L) = \frac{1+A_P f_P^L}{2}.
\end{equation*}
We fit each Pauli individually with its own amplitude parameter $A_P$ and Pauli fidelity $f_P$.

Given observed match counts $\{K_P(L)\}_{L\in\mathcal{L}_\text{cyc}}$ and accepted shot counts $\{N_L\}_{L\in\mathcal{L}_\text{cyc}}$, define the likelihood function for parameters $(A_P,f_P)$ by
\begin{equation*}
    \mathcal{L}(A_P,f_P) = \prod_{L\in\mathcal{L}_\text{cyc}} \binom{N_L}{K_L{(P)}} \left(\frac{1+A_P f_P^L}{2}\right)^{K_P(L)} \left(\frac{1-A_P f_P^L}{2}\right)^{N_L - K_P(L)}.
\end{equation*}
Dropping binomial coefficient terms that do not depend on the fit parameters, the negative log-likelihood is 
\begin{equation*}
\begin{split}
    -\log \mathcal{L}(A_P,f_P)
    =
    -&\sum_{L\in\mathcal{L}_\text{cyc}}
    \Bigg[
    K_P(L)\log\!\left(\frac{1+A_P f_P^L}{2}\right)
     +(N_L-K_P(L))\log\!\left(\frac{1-A_P f_P^L}{2}\right)
    \Bigg].
\end{split}
\end{equation*}
We obtain the maximum-likelihood estimates
\begin{equation*}
(\hat A_P,\hat f_P)=\text{argmin}_{A_P,f_P}\big[-\log \mathcal{L}(A_P,f_P)\big],
\end{equation*}
subject to constraints that ensure that $0 < q_P(L) < 1$ for all fitted cycle depths $L$.  
This yields the fitted Pauli fidelities $\hat f_P$ that we plot in Fig.~\ref{fig:logical-cb-survivals-plots}.

The average of all fitted Pauli fidelities $\hat f_P$ reconstructs the estimated process fidelity of the cycle:
\begin{equation*}    
    F_\text{pro} = \frac{1}{|\mathcal{P}_\text{samp}|}\sum_{P \in \mathcal{P}_\text{samp}} \hat f_P.
\end{equation*}
From the (cycle) process fidelity, we compute the process fidelity of the two-qubit sub-cycle by taking $F_\text{pro, 2Q} = (F_\text{pro})^{1/n}$, where $n$ is the number of 2Q gates in the cycle ($n=1$ for inter-block gates and $n=2$ for intra-block $U_{ZZ}(\pi / 2)$).
We estimate the gate average fidelity from the two-qubit process fidelity through the formula $F_\text{avg} = (4F_\text{pro, 2Q} + 1)/5$ \cite{Nielsen2002} and report the results for each experiment in Table~\ref{tab:logical-cb-fidelities}.
Corresponding results for acceptance rates are reported in Fig.~\ref{fig:logical-cb-accept-plots} and Table~\ref{tab:logical-cb-accept-rates}.

In order to obtain uncertainties on the estimate of the gate-average fidelities, we resample from the binomial distribution for each sequence depth and Pauli and fit the bootstrapped expectation-value decays using the MLE method. From the fitted Pauli fidelities across 500 resamples, we obtain a bootstrap distribution of the gate average fidelity and quantify uncertainty on the empirical estimate using the $1\sigma$ (68\%) quantile around the bootstrap mean.
We report the results in Table~\ref{tab:logical-cb-fidelities}.

\begin{figure}[!h]
    \centering
    \includegraphics[width=.95\linewidth]{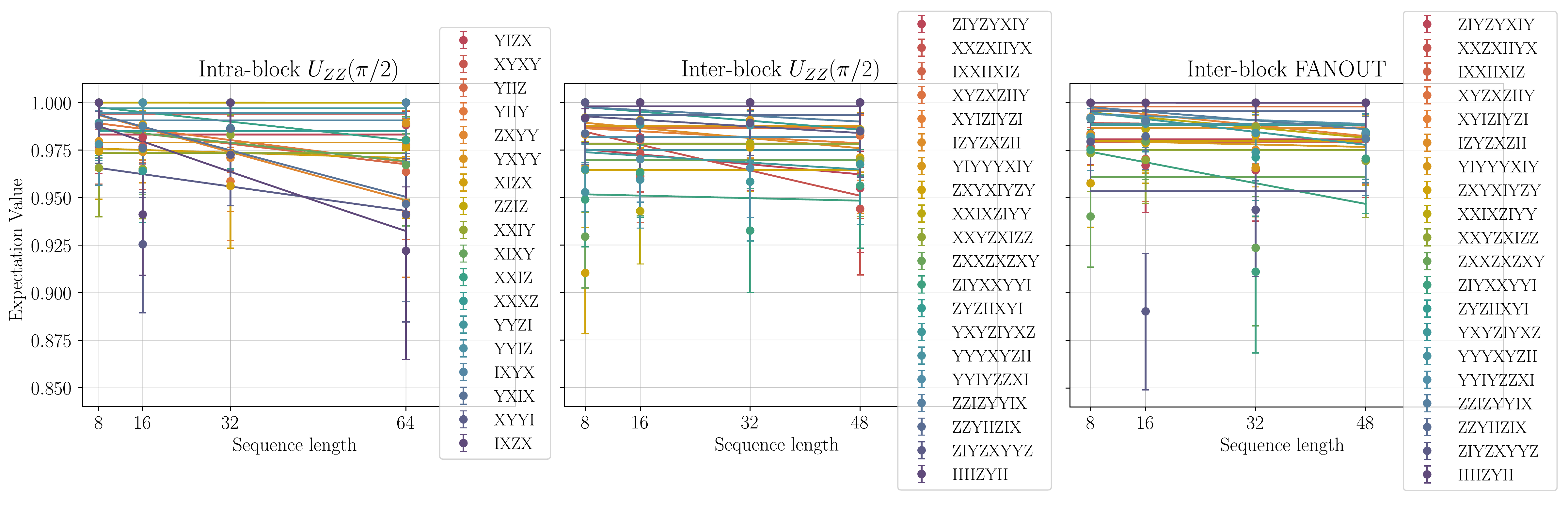}
    \caption{
    Fitted exponential decays of logical Pauli expectation values as a function of sequence depth for the intra-block $U_{ZZ}(\frac \pi 2)$, inter-block $U_{ZZ}(\pi / 2)$, and inter-block FANOUT cycle-benchmarking experiments. Measured expectation values and fitted exponential decays are shown for each sampled 4-qubit Pauli string (for intra-block tests) and 8-qubit Pauli string (for inter-block tests). Each cycle consists of Pauli-twirled logical two-qubit gate layers followed by a round of syndrome extraction every 8 layers. Error bars indicate the 68\% Wilson confidence intervals.
    }
    \label{fig:logical-cb-survivals-plots}
\end{figure}

\begin{figure}[!h]
    \includegraphics[width=.95\linewidth]{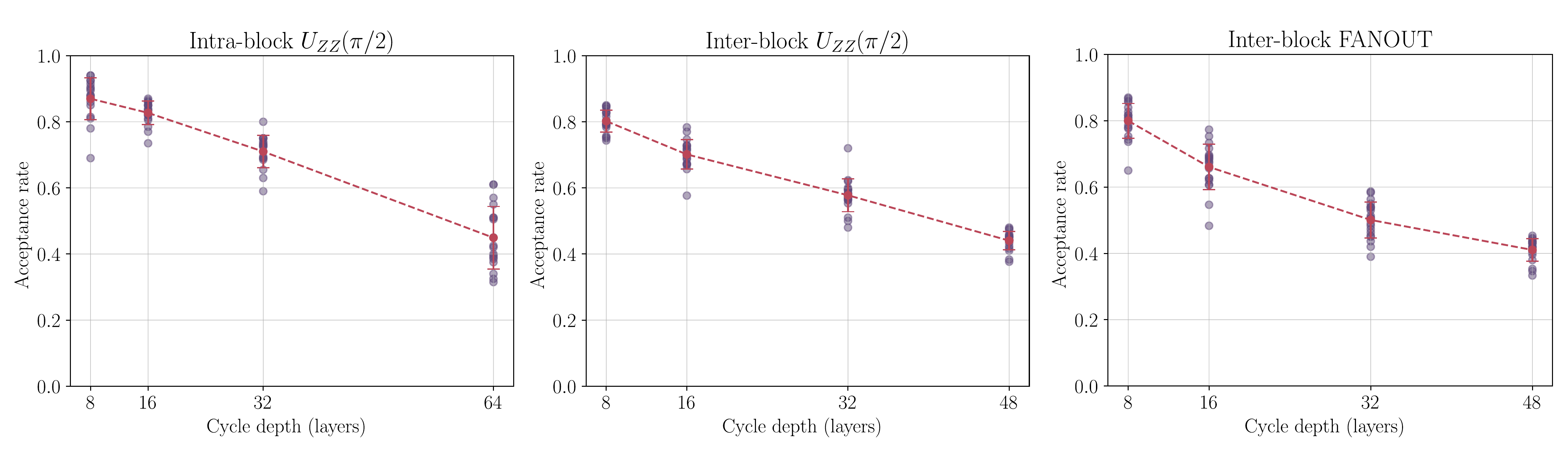}
    \caption{\textbf{Acceptance rates of logical gate benchmarking experiments.} Acceptance rate as function of cycle depth for the intra-block $U_{ZZ}(\pi / 2)$, inter-block $U_{ZZ}(\pi / 2)$, and inter-block FANOUT logical cycle-benchmarking tests. Numerical results are tabulated in Table~\ref{tab:logical-cb-accept-rates}. Pink dashed curves indicate the average acceptance rate at each depth. 
    }
    \label{fig:logical-cb-accept-plots}
\end{figure}

\begin{table}[!h]
\centering
\begin{tabular}{lc}\hline \hline\\[-.25cm]
Experiment & Gate average infidelity   \\[.3em]  \hline \\[-.6em]
Intra-block $U_{ZZ}(\frac \pi 2)$ 
& $1.0^{+2}_{-4} \times10^{-4}$ \\[.1em] 
Inter-block $U_{ZZ}(\frac \pi 2)$ 
& $1.2^{+5}_{-4} \times10^{-4}$ 
\\[.1em] 
Inter-block FANOUT 
& $1.1^{+4}_{-4}\times10^{-4}$  \\[.1em] \hline\hline
\end{tabular}
\caption{Estimated gate average infidelity for various logical cycle-benchmarking tests with 1$\sigma$ confidence intervals obtained through bootstrap resampling.}
\label{tab:logical-cb-fidelities}
\end{table}

\begin{table}[!h]
\centering
\begin{tabular}{lcccc}\hline \hline\\[-.25cm]
Experiment & Cycle depths & Total shots (each depth) & Overall accept rate  \\[.3em]  \hline \\[-.6em]
Intra-block $U_{ZZ}(\frac \pi 2)$ 
& [$8$, $16$, $32$, $64$] 
& $4000$
& [$.87(6)$, $.83(4)$, $.71(5)$, $.45(9)$] 
\\[.1em] 

Inter-block $U_{ZZ}(\frac \pi 2)$  
& [$8$, $16$, $32$, $48$] 
& $6000$ 
& [$.80(3)$, $.70(4)$, $.58(5)$, $.44(3)$] 
\\[.1em] 

Inter-block FANOUT 
& [$8$, $16$, $32$, $48$] 
& $6000$
& [$.80(5)$, $.66(7)$, $.50(5)$, $.41(3)$] 
\\[.1em] \hline\hline
\end{tabular}
\caption{Acceptance rates across various logical cycle-benchmarking (CB) tests showing depths, total shots per depth, and acceptance rates.}
\label{tab:logical-cb-accept-rates}
\end{table}

\section{Partially fault-tolerant gates between $\I{k}$ code blocks} \label{sec:gates}
A convenient feature of the iceberg code is that logical $U_{ZZ}(\theta)$ and $U_{XX}(\theta)$ gates can be implemented with single physical two-qubit gates within an iceberg code block. However, the same cannot be said for logical two-qubit rotations between two disjoint code blocks. The underlying reason is that, if one thinks of the two-qubit Pauli operators appearing in each rotation as being composed of the product of two one-qubit Pauli operators, one has $\logZ{}^i \logZ{}^j = Z_{i}Z_{n-1}Z_{j}Z_{n-1} = Z_{i} Z_{j}$ within a single code block; whereas the corresponding $Z_{n-1}$ operators do not cancel when the logical qubits $q_i$ and $q_j$ are in different code blocks and therefore the logical rotation is four-body.

One method to execute this four-body rotation with some amount of fault-tolerance is to use the well-known decomposition of the $U_{ZZ}(\theta)$ gate in terms of $\text{CX}$ gates and the operator $U_Z (\theta) = \exp\left(-i\frac{\theta}{2} Z \right)$, that is, $U_{ZZ}(\theta) = \CX_{12} U_Z (\theta)_2 \CX_{12}$ (indexing the two qubits involved as $1$ and $2$). Treating all operators in this expression as logical and using the fact that fact that logical CX gates can be applied transversally in the iceberg code (that is, the logical CX operator can be implemented by broadcasting the physical CX operator between the respective physical qubits in each code block), this decomposition can be implemented by the circuit gadget shown on the left-hand side of Fig.~\ref{fig:between-gates}a.

\begin{figure*}[!h]
\centering
\includegraphics[width=\textwidth]{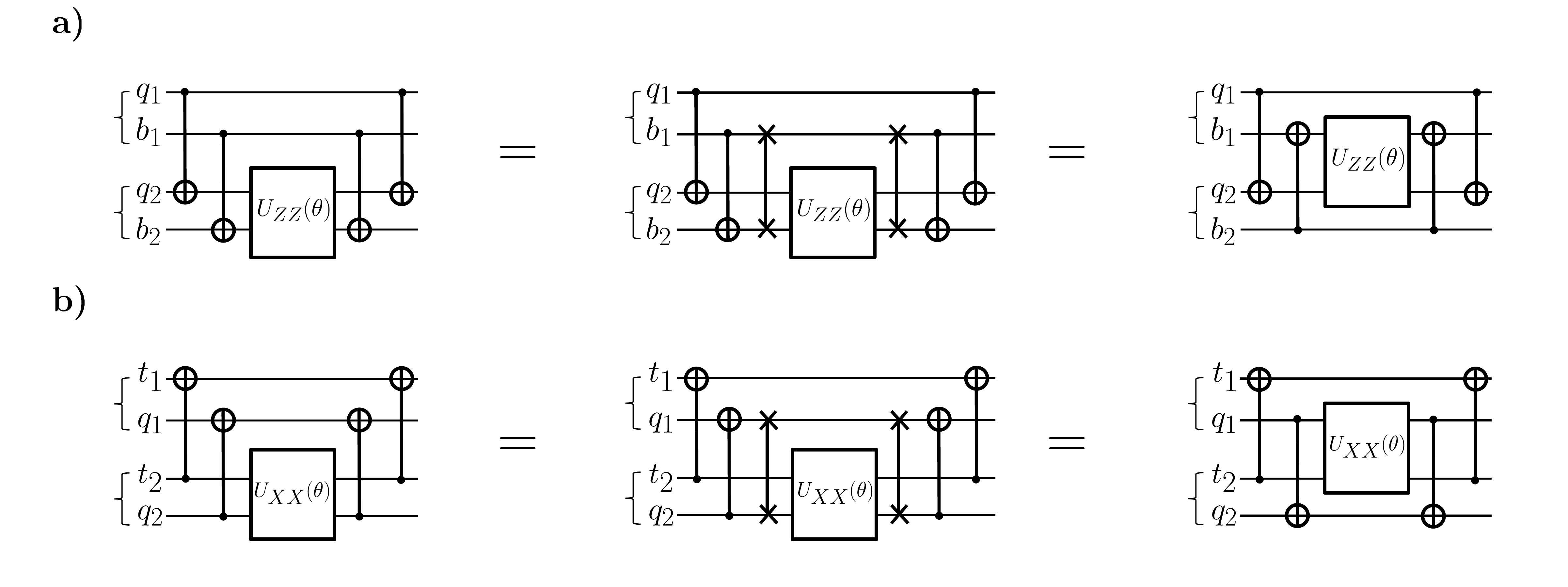}
\caption{
{\bf Partially fault-tolerant gates between code blocks.}
\textbf{(a)} Starting on the left-hand-side with the circuit gadget corresponding to the $\CX$-$U_Z(\theta)$-$\CX$ decomposition of the $U_{ZZ}$, one can insert SWAP gates to move the $U_{ZZ} (\theta)$ operator so that it spans between blocks. With a judicious choice of decomposition of the SWAP into CX gates, and using the fact that Pauli $Z$ operators on the control qubit commute with the CX gate, one arrives at the decomposition on the right-hand side. \textbf{(b)} The analogous series of steps for the circuit gadget implementing the logical $U_{XX} (\theta)$ between code blocks, noting that Pauli $X$ operators on the target qubit also commute with the CX gate. In all circuits above, $t = q_0$ is the iceberg top qubit and $b = q_{n-1}$ is the iceberg bottom qubit on each code block, with brackets indicating qubits that are in the same code block and indices labeling the first or second code block.
}\label{fig:between-gates}
\end{figure*}

This gadget is pFT: any single Pauli error on the CX gates is detected in one of the two code blocks. However, it fails to be fully FT for the same reason that the $U_{ZZ}(\theta)$ gate alone fails to be FT in the iceberg code; namely, weight-2 $XX$, $YY$, or $ZZ$ errors on the $U_{ZZ}(\theta)$ are undetectable. The fault-tolerance can be improved by rewriting the circuit in a logically equivalent way by using SWAP gates to move the location of the $U_{ZZ}(\theta)$ gate so that it is supported on both code blocks, as in Fig.~\ref{fig:between-gates}. The resulting circuit gadgets on the right-hand side for $U_{ZZ}(\theta)$ and $U_{XX}(\theta)$ have the desirable property that the non-FT rotation gate now spans between code blocks, with the effect that only the $ZZ$ errors remain undetectable on the physical $U_{ZZ}(\theta)$ gate. Note that $U_{ZZ}(\theta)$ is the only physically implemented two-qubit gate in the QCCD architecture, so the corresponding undetected $XX$ error on the $U_{XX} (\theta)$ gate is physically a $ZZ$ error on the underlying $U_{ZZ}(\theta)$ gate as well.

The Pauli error model measured for the two-qubit gates on Helios \cite{helios} has negligible contribution from the $XX$ and $YY$ channels, with essentially all weight-2 errors being of $ZZ$ type. Therefore, this construction does not necessarily substantially improve the fault-tolerance of a practical implementation of these gates on Helios. However, we expect that they may be useful for iceberg code implementations on other devices that may have a different Pauli error model.

In Fig.~\ref{fig:entangling}a in the main text we also present the circuit decomposition of the inter-block FANOUT gate, which enacts a global logical parity flip $\prod_i \logX^i$ on one code block controlled by a logical qubit in another code block. The intuition underlying this circuit gadget is that the global logical parity flip can be enacted by applying $X_0 X_{n-1}$ on the second code block, and that controlling on logical qubit $c$ is equivalent to physically controlling on both $c$ and $q_{n-1}$. It is straightforward to verify that this circuit gadget applies $X_0 X_{n-1}  = \prod_i \logX^i$ if and only if qubit $c$ is in the logical $|\overline{1}\rangle$ state.

The FANOUT gate is interesting for several reasons: first, that it is FT even to all single weight-2 gate errors, unlike the pFT $U_{ZZ}$ gates. Second, it allows distribution of entanglement between code blocks with a number of gates (four) that does not scale with the size of the code blocks, unlike a transversal CX construction. A convenient aspect of the construction is that the precise way in which it distributes entanglement among code blocks aligns exactly with efficient logical GHZ state preparation among many $\I{k}$ code blocks as described in the main text and explained further in Sec.~\ref{sec:ghz_prep}.

Although we do not provide constructions here, there are clear extensions of the FANOUT gate to other interesting FT gates that have constant overhead (i.e., with number of gates not scaling with code block size): namely, one can control on, or target, a specific subset of qubits with simple modifications of the choices of $\CX$ gates applied.

\section{Unencoded GHZ state fidelity estimation} \label{sec:physical_ghz}
To measure the fidelity of the physical GHZ state, we use a state verification scheme introduced in Ref.~\cite{Li2020}. A state verification involves constructing an ensemble of pass/fail tests such that the average passing probability reveals information of the closeness to the ideal state. Specifically, one constructs a verification operator $\Omega=\sum_lp_l\Pi_l$ where $p_l$ is the probability of performing a pass/fail test $\{\Pi_l,\mathrm{I}-\Pi_l\}$, which is a two-outcome Positive Operator-Valued Measure (POVM). As shown in Ref. \cite{Li2020}, when the verification operator is homogeneous, i.e., $\Omega=|\Psi\rangle\langle\Psi|+\beta(\mathrm{I}-|\Psi\rangle\langle\Psi|)$ where $|\Psi\rangle$ is the target state and $\beta$ is the second largest eigenvalue of $\Omega$, the fidelity $F:=\langle\Psi|\rho|\Psi\rangle$ between the prepared state $\rho$ and the target state $|\Psi\rangle$ is related to the average passing probability as 
\begin{align}
F=\frac{\Tr[\Omega\rho]-\beta}{1-\beta}. \label{eq:F_to_omega}
\end{align}
Therefore, measuring the fidelity amounts to measuring the average passing probability of the ensemble of tests $\Tr[\Omega\rho]$. 

For the GHZ state $|\text{GHZ}\rangle=(1/\sqrt{2})(|0\cdots0\rangle+|1\cdots1\rangle)$, Ref. \cite{Li2020} found that 
\begin{align*}
    \Omega&=|\text{GHZ}\rangle\langle\text{GHZ}|+\frac{1}{3}\left(\mathrm{I}-|\text{GHZ}\rangle\langle\text{GHZ}|\right) \nonumber\\
    &=\frac{1}{3}P_Z+\frac{2}{3}\left(\frac{1}{2^{N-1}}\sum_{XY}P_{XY}\right),
\end{align*}
where $P_Z=|0\cdots0\>\langle0\cdots0|+|1\cdots1\rangle\langle1\cdots1|$ and $P_{XY}$ is a projector onto the +1 eigenspace of an XY-Pauli stabilizer of $|\text{GHZ}\rangle$ as defined in Eq.~(12) of Ref.~\cite{Li2020}. Passing a $P_Z$ test means the outcome is either $|0\rangle^{\otimes N}$ or $|1\rangle^{\otimes N}$ while passing a $P_{XY}$ test means the corresponding Pauli expectation value is $+1$. Operationally, $\Tr[\Omega \rho]$ is obtained by performing the $P_Z$ test with probability $1/3$ and with probability $2/3$, uniformly sampling a $P_{XY}$ test to perform (there are $2^{N-1}$ of them). Since $\beta=1/3$ in this case, we can obtain the fidelity as $F=(3/2)\Tr[\Omega\rho]-(1/2)$. 

It is worth noting that the sampling cost of this procedure is slightly lower than that of Direct Fidelity Estimation (DFE) \cite{Flammia2011}. Recall that DFE samples a stabilizer uniformly at random and measures the expectation value of that stabilizer. $F$ is obtained by averaging over all measurements. Since each observable is a Pauli operator that takes value $\pm1$ for every shot, sampling and measuring a stabilizer for each shot corresponds to an i.i.d. Bernoulli trial with outcomes $\pm1$. Hence, the variance of the estimate is $(1-F)(F+1)$. However, as shown in Ref. \cite{Li2020}, the variance of the fidelity obtained via Eq.~\eqref{eq:F_to_omega} is $(1-F)\left(F+\tfrac{\beta}{1-\beta}\right)$. In this case where $\beta=1/3$, the variance is $(1-F)(F+0.5)$, marginally lower than that of DFE.

In the experiment, we used a total of 1000 shots for each fidelity estimation, where $1/3$ of the shots were allocated to $P_Z$ tests and the rest for $P_{XY}$ tests. For each shot allocated to the $P_{XY}$ test, we sampled and measured a different XY-Pauli stabilizer of $|\text{GHZ}\rangle$. The average passing probabilities of the $P_Z$ and $P_{XY}$ tests are weighted by $1/3$ and $2/3$ correspondingly to obtain $\Tr[\Omega\rho]$, which is then used to deduce fidelity via Eq.~\eqref{eq:F_to_omega}. 

\section{Logical GHZ state preparation in $d=2$ iceberg codes} \label{sec:ghz_prep}

In a single global $\I{k}$ code, the physical state on the $n$ qubits encoding the logical $|\overline{\text{GHZ}}_k\rangle$ state takes a particularly simple form $|\overline{\text{GHZ}}_k\rangle = |\Phi\rangle_{tb} \otimes|\text{GHZ}_{n-2}\rangle$  where $|\Phi\rangle_{tb} = \tfrac{1}{\sqrt{2}}(|00\rangle_{tb} + |11\rangle_{tb})$ is a Bell pair of the top and bottom qubits $t=q_0$ and $b=q_{n-1}$ respectively. As a result, the logical GHZ state can be directly prepared fault-tolerantly by preparing the physical GHZ state on the system qubits (via the same log-depth initialization gadget that would be applied to prepare the $|\overline{0}\rangle^{\otimes k}$ state) and the Bell pair on $q_0, q_{n-1}$ independently. We demonstrate the state preparation gadget in Fig.~\ref{fig:ghz-global}. We note that the Bell pair state preparation is FT since the undetectable weight-2 errors on the $\CX$ gate stabilize the Bell state.

\begin{figure*}[!h]
\centering
\includegraphics[width=\textwidth]{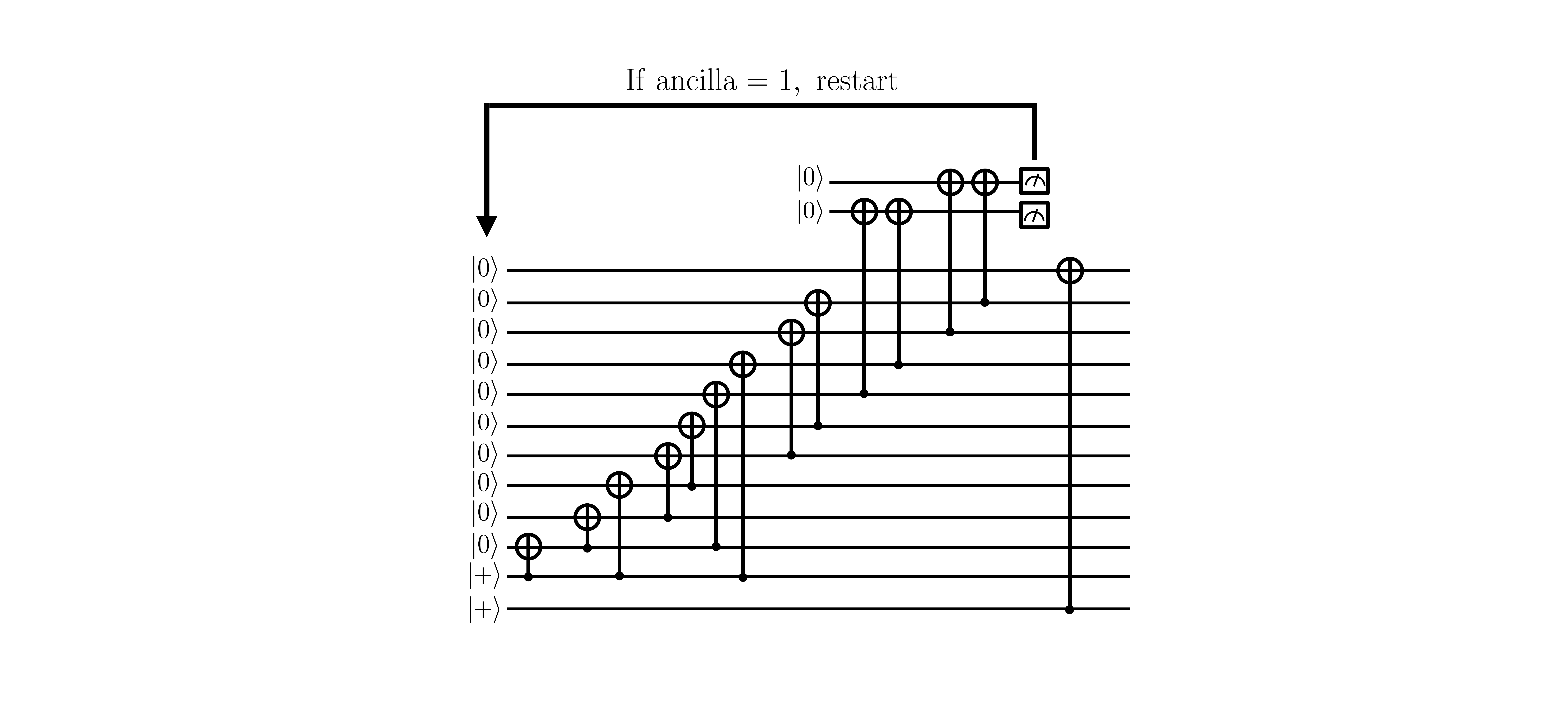}
\caption{
{\bf Logical GHZ state preparation in a global $\I{k}$ code.} The state preparation separately prepares a Bell pair between $q_0$, $q_{n-1}$ and a physical GHZ state on the remaining qubits in log-depth following the same construction as Fig.~\ref{fig:logdepth}.
}\label{fig:ghz-global}
\end{figure*}

As discussed, we perform single-shot fidelity estimation by reading out the GHZ stabilizer $\prod_i \logX^i$ non-destructively simultaneously with obtaining all other stabilizers $\logZ^i \logZ^{i+1}$ from destructive readout in the $Z$ basis, using the fact that $\prod_i \logX^i$ commutes with $S_X$. We measure $\prod_i \overline{X^{i}}$ by noting that it is stabilizer-equivalent to $X_0 X_{n-1}$, which we measure directly by coupling to a single ancilla qubit. To make the implementation FT against state preparation and measurement (SPAM) errors, we repeat the non-destructive measurement of $\prod_i \logX^i$ twice and discard shots in which the parity does not agree between the two measurements. Subsequently measuring $S_X$ non-destructively and reading out each qubit in the $Z$ basis obtains $S_Z$ and all $\logZ^{i} \logZ^{i+1}$. Shots with $S_X = -1$ or $S_Z = -1$ are postselected out and the logical GHZ state fidelity is then estimated as the fraction of accepted shots that prepared $+1$ eigenstates of both $\prod_i \logX^i$ and all $\logZ^i \logZ^{i+1}$. This procedure is demonstrated in the general case of possibly multiple code blocks in Fig.~\ref{fig:entangling}c.

In Table~\ref{tab:ghz} we report the values that appear in Fig.~\ref{fig:encoding}e, excluding results from GHZ state preparation in concatenated iceberg codes, to be discussed in Sec.~\ref{sec:ghz_prep_d4}.

\begin{table}[h!]
\begin{ruledtabular}
\begin{tabular}{lccccc}
Experiment & Fidelity & Total shots & Leakage accept rate &  Overall accept rate  \\[.3em]  \hline \\[-.6em]
$N = 48$ unencoded  & $.79(2)$ & $1000$ & -- & --  \\[.1em] 
$N = 48$ unencoded, LH  & $.74(2)$ & $1000$ & $.966^{+5}_{-6}$ & $.966^{+5}_{-6}$ \\[.1em]
$\I{48}$ & $.983^{+4}_{-5}$ & $2000$ &  -- & $.42(1)$ \\[.1em] 
$\I{48}$, LH & $.986^{+3}_{-4}$ & $3000$ & $.908(5)$ & $.392(9)$  \\[.1em] 
12$\times \I{4}$  & $.992(3)$ & $3000$ & -- & $.288(8)$  \\[.1em] 
12$\times \I{4}$, LH  & $.998^{+1}_{-3}$ & $3000$ & $.683^{+8}_{-9}$ &  $.210^{+8}_{-7}$   \\[.1em]\hline \\[-0.8em]
$N = 60$ unencoded  & $.78(2)$ & $1000$ & -- & --   \\[.1em] 
$N = 60$ unencoded, LH  & $.76(2)$ & $1000$ & $.956^{+6}_{-7}$ & $.956^{+6}_{-7}$  \\[.1em]
$\I{60}$  & $.976^{+5}_{-6}$ & $2000$ & -- & $.42(1)$ \\[.1em] 
$\I{60}$, LH  & $.980(4)$ & $3000$ & $.650(6)$ & $.389(9)$  \\[.1em] 
15$\times \I{4}$ & $.997^{+2}_{-3}$ & $3000$ & -- & $.198(7)$   \\[.1em] 
15$\times \I{4}$, LH & $1.000^{+0}_{-2}$ & $3000$ & $.540(9)$ & $.146^{+7}_{-6}$  \\[.1em] \hline \\[-0.8em]
$N = 94$ unencoded  & $.64(2)$ & $1000$ & -- & --   \\[.1em] 
$N = 94$ unencoded, LH & $.61(2)$ & $1000$ & $.89(1)$ &  $.89(1)$  \\[.1em] 
$\I{94}$  & $.87(1)$ & $2000$ & -- & $.36(1)$   \\[.1em]
$\I{94}$, LH  & $.949^{+7}_{-9}$ & $3000$ & $.650(9)$ & $.25(1)$   \\[.1em] 
\end{tabular}
\end{ruledtabular}
\caption{Logical distance-two and unencoded GHZ experiment results. As before, LH refers to leakage-heralded measurement and the corresponding leakage accept rate is reported in the third column.}
\label{tab:ghz}
\end{table}

\section{Logical GHZ state preparation in the $\I{8}\circ\I{6}$ code} \label{sec:ghz_prep_d4}

The stabilizers of the 48-logical-qubit GHZ state $\ket{\overline{\text{GHZ}}_{48}}$ encoded in $\I{8}\circ \I{6}$ are generated by the $\I{8}\circ \I{6}$ stabilizers, the pairwise product of all logical Z operators $\logZ^{i_1,j_1}\logZ^{i_2,j_2}$, and the product of all logical $X$ operators $\prod_{i=1}^{i=6}\prod_{j=1}^{j=8} \logX^{i,j}.$ This last stabilizer can be rewritten as $\prod_{i=1}^{i=6}\prod_{j=1}^{j=8} \logX^{i,j} = X_{0,0}X_{7,0}X_{0,9}X_{7,9}.$ To see this, observe that the product of X operators forming a rectangle on the four corners of the $8 \times 10$ grid will anti-commute with all the $\logZ^{i,j}.$ To prepare the global GHZ state, we begin by fault-tolerantly preparing the $\ket{\overline{0}}^{\otimes 48}$ state via the ``Many Hypercubes'' state preparation circuit \cite{Goto:2024xd}. This state possesses all the $Z$ stabilizers of $\ket{\overline{\text{GHZ}}_{48}}$ and all $X$ $\I{8} \circ \I{6}$ stabilizers. To transform the state into one stabilized by $\prod_{i=1}^{i=6}\prod_{j=1}^{j=8} \logX^{i,j}$ as well, we fault-tolerantly measure this operator, which will return $\pm 1$ eigenvalue. If the eigenvalue is $-1$, we virtually apply a Pauli correction to flip it to $+1$ by updating the Pauli frame, completing the GHZ preparation procedure. 

We now describe how to fault-tolerantly measure the global logical $X$ operator. For an $\I{8}$ code, the logical $X$ operator $\prod_{i=1}^{6}\logX^i = X_0X_7$ is the stabilizer of a GHZ state on all six logical qubits. Therefore, we can use an $\I{8}$ ancilla encoding a GHZ state to measure the $\I{8}\circ\I{6}$ operator $\prod_{i=1}^{i=6}\prod_{j=1}^{j=8} \logX^{i,j} = X_{0,0}X_{7,0}X_{0,9}X_{7,9}$ by performing transversal CX gates between the $\I{6}$ GHZ state and the top and bottom data qubit iceberg blocks of $\I{8}\circ\I{6}$, shown in Fig.~\ref{fig:non_ft_GHZ_d4}. However, this is not FT because $X$ errors on the lower-level GHZ state can undetectably propagate to weight-two errors on the data qubits. To fix this, we extract the $Z$ stabilizers of $\ket{\overline{\text{GHZ}}_{6}}$, which allows us to locate single-qubit $X$ errors. For instance, $X_1$ will flip the stabilizers $Z_1Z_i$ for $2 \leq i \leq 6$. This will allow us to turn generalized weight-two errors into generalized weight-one errors (see Defs.~\ref{def:gen wt 2} and~\ref{def:gen wt 1}). The only single-qubit $X$ errors that are indistinguishable are $X_0$ and $X_7$, but these are equivalent up to a stabilizer of the GHZ state. We extract these stabilizers by measuring transversally in the $Z$ basis and extract the $X$ GHZ stabilizer via physical ancillae and the $\I{6}$ stabilizer via the EDZ gadget to get the circuit in Fig.~\ref{fig:se_protocols}b. Altogether, we get the FT measurement circuit in Fig.~\ref{fig:ft_GHZ_d4}, which must be repeated once to detect a logical measurement error in the lower-level $\I{6}$ block. If the two measurements disagree by a logical operator of $\I{6}$, we postselect.

After preparing the GHZ state, we measure in either the $Z$ and $X$ bases and use the expectation values of the logical observables to estimate the fidelity of this logical GHZ state via a union bound as in Ref.~\cite{Hong_2024}. The entire circuit consisting of the state preparation of zero, the measurement of the global logical $X$ operator, and a transversal $X$- or $Z$-basis measurement has a fault-distance of four, so we perform correlated decoding over the entire circuit volume, correcting a single error occurring anywhere in the circuit and postselecting if two or more errors occur. This implies that this procedure should have an infidelity of $O(p^3)$ and a discard rate of $O(p^2)$ with respect to the physical error rate. This is illustrated in circuit-level simulations of this protocol in Stim, shown in Fig.~\ref{fig:GHZScalingPlot} in the case of the $Z$ basis. The scaling in the $X$ basis is similar.

On Helios, we prepared the 48 logical qubit GHZ state and then either measured in the $X$ basis and the $Z$ basis. We ran 4000 shots in either basis and used eq. (4) from \cite{Hong_2024} to upper bound the infidelity of the state. Since we correct any single error including during state preparations (as opposed to preselecting on these preparations) there are only two acceptance rates: leakage acceptance and total acceptace. These are given in Table~\ref{tab:ghz d=4 results}.

\begin{figure}[!h]
\centering
\includegraphics[width=0.5\textwidth]{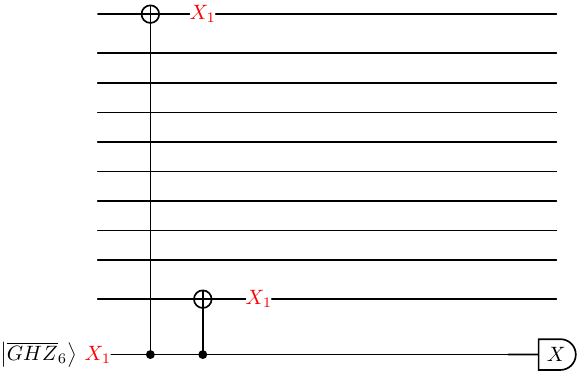}
\caption{
\textbf{Non-fault-tolerantly measuring the global logical $X$ stabilizer of $\ket{\overline{\text{GHZ}}_{48}}$.} A GHZ state encoded in $\I{6}$ can be used to measure the global logical $X$ stabilizer of $\ket{\overline{\text{GHZ}}_{48}}$. We show an example of non-fault-tolerance where a single $X_1$ error on the lower-level GHZ state can undetectably spread to a weight-2 error on the data code block.
}\label{fig:non_ft_GHZ_d4}
\end{figure}

\begin{figure}[!h]
\centering
\includegraphics[width=0.5\textwidth]{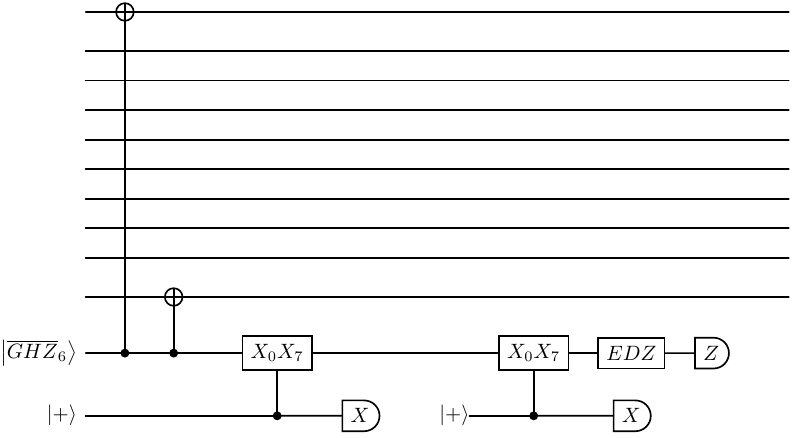}
\caption{
{\textbf{Fault-tolerantly measuring the global logical $X$ stabilizer of $\ket{\overline{\text{GHZ}}_{48}}$.} The transversal $Z$ measurement of the lower-level GHZ state will allow for the correction of an $X$ hook error from the ancilla. To be robust to two errors, this circuit must be repeated again since two measurement errors could cause a logical measurement error of the $X_0X_7$ operator of the lower-level $\I{6}$ block.}
}\label{fig:ft_GHZ_d4}
\end{figure}

\begin{figure}[!h]
\centering
\includegraphics[width=1.0\textwidth]{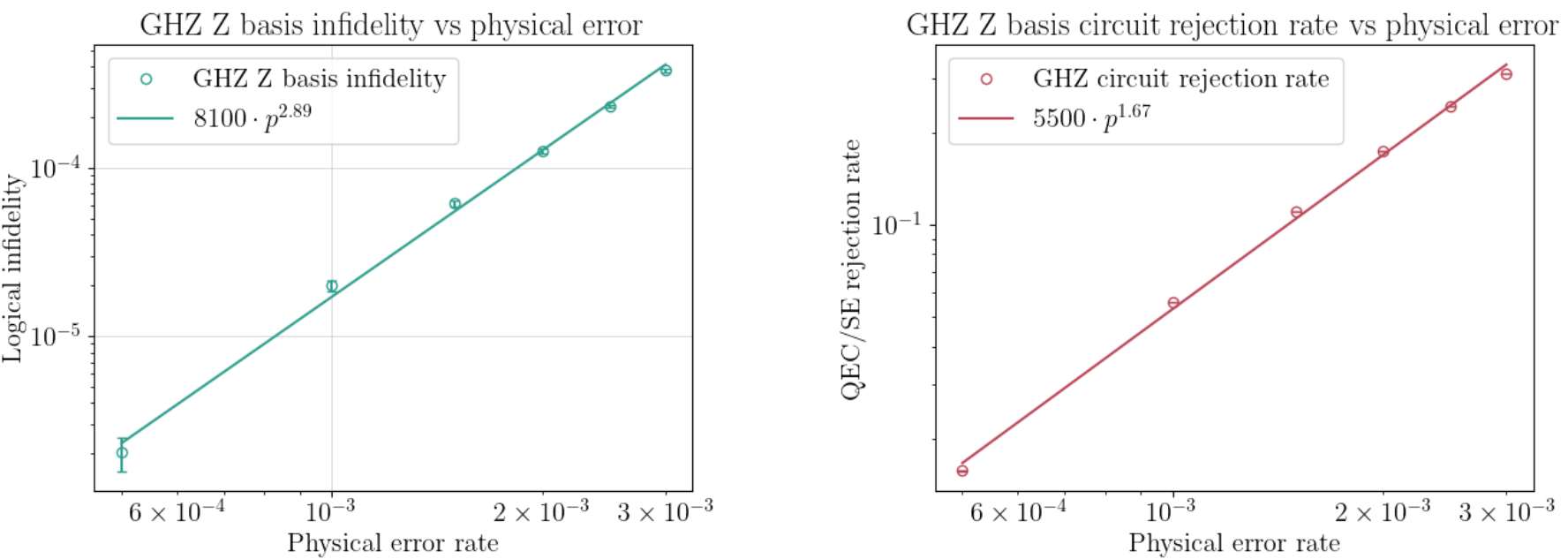}
\caption{
{\textbf{Simulated GHZ Z basis infidelity and reject rate as a function of physical noise rate $p$.} This data is derived from circuit-level simulations performed in Stim~\cite{Gidney:2021wm}. A two-qubit depolarizing error and bit flip error are inserted after every two-qubit gate and measurement respectively with strength $p$. Because one error anywhere in the circuit is corrected, the total rejection rate should scale as $O(p^2)$. Fault tolerance to distance $4$ further implies that the infidelity should scale as $O(p^3)$. Performing linear regression on the simulated rejection rate and infidelity to estimate these exponents demonstrates this scaling.}
}\label{fig:GHZScalingPlot}
\end{figure}

\begin{table}[!h]
\begin{tabular}{lcccclcc}\hline \hline\\[-.25cm]
\multicolumn{1}{c}{Experiment} & \multicolumn{1}{c}{Leakage acceptance} & \multicolumn{1}{c}{Overall accept rate} & Accepted shots & Errors \\[.3em]  \hline \\[-.6em] 
$\I{8} \circ \I{6}$ (X Basis) & $.372_{-6}^{+8}$ & $.53(1)$ & 790 & 0 \\[.1 em]
$\I{8} \circ \I{6}$ (Z Basis) & $.462_{-8}^{+7}$ & $.67_{-1}^{+2}$ & 1233 & 0 \\[.1 em] \hline \hline
\end{tabular}
\caption{Experimental data for $d=4$ 48 logical qubit GHZ state experiments in the $\I{8} \circ \I{6}$ code. 4000 shots were submitted in each basis. Leakage and overall acceptance rates are given. Of the 790 accepted shots in the X basis and 1233 accepted shots in the Z basis, no errors occurred.}
\label{tab:ghz d=4 results}
\end{table}

\section{Data from preliminary XY model experiments in the $\I{50}$ code} \label{sec:bilayer}

Prior to executing the $\I{64}$ code $XY$ model experiments in the main text, we performed several preliminary experiments using the $\I{50}$ code for the $XY$ model on a $2\times 5\times 5$ lattice with periodic boundary conditions (Fig.~\ref{fig:50lq_results}a). These experiments were aimed at assessing the optimal configuration (in terms of achievable circuit fidelity) of a relatively sparse number of syndrome extraction rounds. Three different configurations were studied. In the first, we extracted both syndromes $S_X$ and $S_Z$ at both the midpoint of the circuit as well as at the end of the circuit ($S_X$ extraction $+$ midpoint SE). In the second, both syndromes were extracted at the midpoint, but only $S_Z$ was extracted at the end of the circuit (midpoint SE). In the third and final, no midcircuit syndrome extraction was performed but both syndromes were extracted at the end of the circuit ($S_X$ extraction). We did not use leakage-heralded measurement in any of the $\I{50}$ experiments.

We report the results of these experiments in terms of the obtained estimates of fidelity and acceptance rates in Fig.~\ref{fig:50lq_results}b, Fig.~\ref{fig:50lq_results}c and Tables~\ref{tab:50lq_results},\ref{tab:50lq_results2},\ref{tab:50lq_fit_params}. The first two configurations demonstrated break-even performance across the range of Trotter steps studied, with the third performing similarly at $s=2$ to $s=6$ but failing to demonstrate any clear improvement in fidelity at $s=8$ and $s=10$. For this reason, we chose the second configuration (midpoint SE) for the $\I{64}$ experiments in the main text, which performed similarly to the first configuration ($S_X$ extraction $+$ midpoint SE) but with approximately a factor of two improvement in acceptance rate (implying that terminal $S_X$ extraction tends to typically detect many Pauli $Z$-type memory errors induced by the gadget itself that do not affect terminal measurement in the $Z$ basis). 

\begin{figure}[htbp]
  \centering
  \includegraphics[width=\columnwidth]{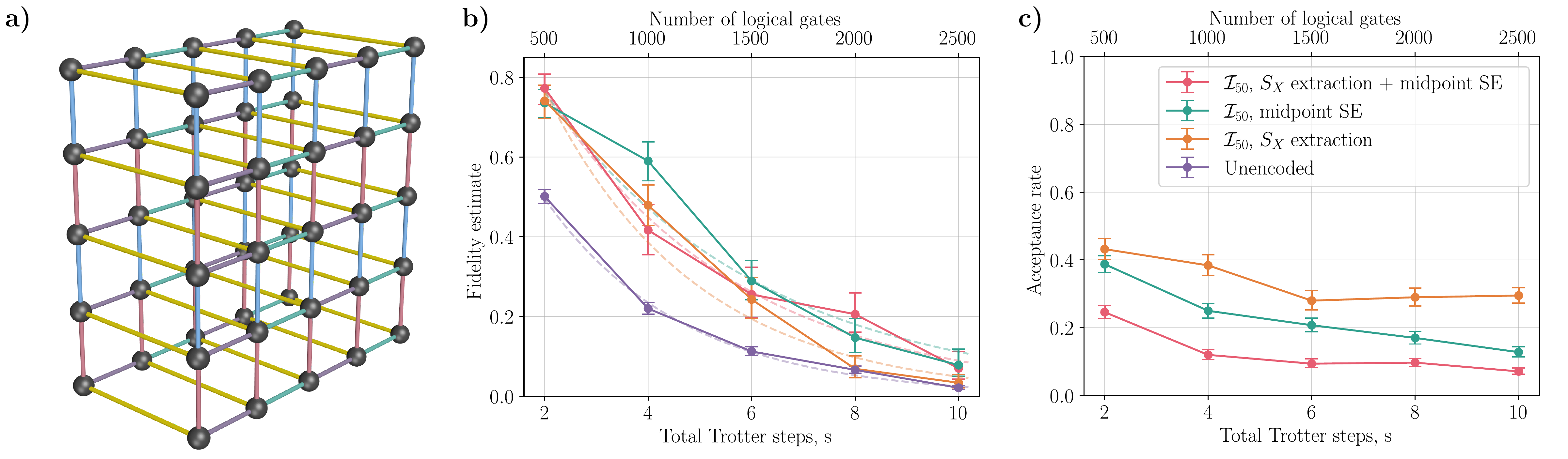}
  \caption{\textbf{Mirror benchmarking of encoded Hamiltonian simulation in the $\I{50}$ code.} \textbf{(a)} The natural edge-$5$-coloring of the $2\times 5 \times 5$ periodic lattice. Edges connecting the periodic boundaries in the latter two directions are not shown to reduce visual clutter, but the corresponding gates are included in the circuits executed, as in the main text. \textbf{(b)} Fidelity estimate as a function of total Trotter steps in the full mirrored circuit. \textbf{(c)} Acceptance rate as a function of total Trotter steps in the full mirrored circuit. }
  \label{fig:50lq_results}
\end{figure}

\begin{table}[!h]
\begin{tabular}{llll}\hline \hline\\[-.25cm]
\multicolumn{1}{c}{Experiment} & \multicolumn{1}{c}{\quad Total Trotter steps},  $s$ & \multicolumn{1}{c}{Fidelity estimate} & \multicolumn{1}{c}{Total shots ($\times 10^{2}$)}  \\[.3em]  \hline \\[-.6em]
$\I{50}$, $S_X$ extraction $+$ midpoint SE & \qquad [$2$, $4$, $6$, $8$, $10$]  & \qquad[$.77(4)$, $.42(6)$, $.26^{+7}_{-6}$, $.21^{+5}_{-4}$, $.07^{+4}_{-3}$]  & \qquad [$5$, $5$, $5$, $7$, $8$]  \\[.1em] 
$\I{50}$, midpoint SE & \qquad [$2$, $4$, $6$, $8$, $10$]  & \qquad[$.73^{+3}_{-4}$, $.59(5)$, $.29(5)$, $.15^{+5}_{-4}$, $.08^{+4}_{-3}$]  & \qquad [$4$, $4$, $4$, $4$, $5$]  \\[.1em] 
$\I{50}$, $S_X$ extraction & \qquad [$2$, $4$, $6$, $8$, $10$]  & \qquad[$.74(4)$, $.48(5)$, $.24(5)$, $.07^{+3}_{-2}$, $.03^{+2}_{-1}$]  & \qquad [$2.5$, $2.5$, $2.5$, $3$, $4$]  \\[.1em] 
Unencoded & \qquad [$2$, $4$, $6$, $8$, $10$]  & \qquad[$.50(2)$, $.22(1)$, $.11(1)$, $.066^{+9}_{-8}$, $.021^{+6}_{-5}$]  & \qquad [$8$, $8$, $8$, $8$, $8$]  \\[.1em] \hline\hline
\end{tabular}
\caption{Fidelity and shot counts across 50 logical qubit $XY$ model quantum simulation experiments at different numbers of total Trotter steps $s$. The five values reported for fidelity estimate and total shots correspond to the five different total Trotter steps respectively.}
\label{tab:50lq_results}
\end{table}

\begin{table}[!h]
\begin{tabular}{ll} \hline \hline\\[-.25cm]
\multicolumn{1}{c}{Experiment} & \multicolumn{1}{c}{\quad Overall accept rate}  \\[.3em]  \hline \\[-.6em]
$\I{50}$, $S_X$ extraction $+$ midpoint SE & \qquad [$.25(2)$, $.12^{+2}_{-1}$, $.09(1)$, $.10(1)$, $.07(1)$]    \\[.1em] 
$\I{50}$, midpoint SE & \qquad [$.39(2)$, $.25(2)$, $.21(2)$, $.17(2)$, $.13^{+2}_{-1}$]  \\[.1em] 
$\I{50}$, $S_X$ extraction & \qquad [$.43(3)$, $.38(3)$, $.28(3)$, $.29(3)$, $.30(2)$]   \\[.1em] 
Unencoded & \multicolumn{1}{c}{\qquad  --}   \\[.1em] \hline\hline
\end{tabular}
\caption{Overall acceptance rates across 50 logical qubit $XY$ model quantum simulation experiments. The five values reported for acceptance rate again correspond to the five different total Trotter steps.}
\label{tab:50lq_results2}
\end{table}

\begin{table}[!h]
\begin{tabular}{lll}\hline \hline\\[-.25cm]
\multicolumn{1}{c}{Experiment} &   \multicolumn{1}{c}{SPAM parameter $A$} & \multicolumn{1}{l}{\qquad Effective average 2Q infidelity $\tilde{\epsilon}_{2Q}$}  \\[.3em]  \hline \\[-.6em]
$\I{50}$, $S_X$ extraction $+$ midpoint SE & \qquad $1.32(9)$ & \qquad \qquad\qquad  $8.6(7)\times 10^{-4}$  \\[.1em] 
$\I{50}$, midpoint SE & \qquad $1.2(2)$ & \qquad \qquad\qquad   $8(1) \times 10^{-4}$  \\[.1em] 
$\I{50}$, $S_X$ extraction & \qquad  $1.5(2)$ &  \qquad \qquad\qquad  $11(1)\times 10^{-4}$  \\[.1em] 
Unencoded & \qquad  $1.05(7)$ & \qquad \qquad \qquad $11.9(6) \times 10^{-4}$  \\[.1em]  \hline\hline
\end{tabular}
\caption{Fit parameters $A$, $\tilde{\epsilon}_{2Q}$ obtained from fitting the model $F(N_{\text{gates}}) = A(1-\frac54 \tilde{\epsilon}_{2Q})^{N_{\text{gates}}}$ to the survival probability as a function of number of logical gates in 50 logical qubit $XY$ model quantum simulation experiments. The SPAM parameters $A$ all exceed $1$, which we take as an general indication of minor deviations from the exponential-decay model. (discussed further below).}
\label{tab:50lq_fit_params}
\end{table}

\section{Data from $\I{64}$ XY model experiments} \label{sec:64lq_data}

We report here the explicit values appearing in Fig.~\ref{fig:computing}d in Tables~\ref{tab:64lq_results},~\ref{tab:64lq_results2} along with additional information about shot counts and leakage acceptance rates. As explained in Methods~\ref{sec:Helios} and plotted in Fig.~\ref{fig:computing}d, we calculate effective average two-qubit gate fidelities for each experiment by fitting the simple model $F(N_{\text{gates}}) = A(1-\frac54 \tilde{\epsilon}_{2Q})^{N_{\text{gates}}}$ to the survival probability as a function of number of logical 2Q gates. In Table~\ref{tab:fit_params} we report both fit parameters obtained by this method. We note that in general we do not expect logical survival probability to decay exponentially on theoretical grounds since logical errors are in general not caused by i.i.d. gate errors but can result from multiple faults aggregated across the physical gates. However, empirically the survival probability appears to be well-modeled by exponential decay and we find the $\tilde{\epsilon}_{2Q}$ thus extracted to be an informative metric for the effective gate error. One explanation is simply that at shallow circuit depth, the leading contribution to logical errors are i.i.d. $ZZ$ errors on individual physical gates. For deep circuits, considering the coarse resolution of total number of Trotter steps, there is a rapid transition to the regime where multiple errors typically occur and syndromes appear effectively random, in which case acceptance is not strongly correlated to particular patterns of errors. This is consistent with the acceptance rate reported for the deepest circuits in the $\I{64}$ experiments which approach a plateau around $12.5\%$ acceptance rate, which would be expected for three random syndrome values. We emphasize that the effective two-qubit gate error obtained by this procedure is not directly comparable to the bare component infidelity of the two-qubit gate, as it also includes additional overheads such as memory errors (see also Methods~\ref{sec:Helios}). Rather, it should be compared to the effective 2Q gate infidelity calculated from full-circuit benchmarks such as mirror benchmarking as in \cite{helios}, which yielded $\tilde{\epsilon}_{2Q} = 20(2) \times 10^{-4}$. We note also that the full-circuit metric in Ref.~\cite{helios} used maximally entangling $U_{ZZ} (\pi/2)$ gates and arbitrary connectivity; in this work, we use half-maximally entangling gates and a circuit with more local three-dimensional connectivity, so the fairer comparison point is the $\tilde{\epsilon}_{2Q} = 10.7(5) \times 10^{-4}$ obtained in our unencoded experiment.

For all encoded experiments in this section, we terminate shots early if an error is detected during midcircuit syndrome extraction. It is possible that errors can be detected midcircuit due to leakage of data qubits leading to nontrivial syndromes. However, we did not obtain leakage information for the data qubits in shots that terminated early. Therefore, leakage acceptance rates are computed only with respect to shots that survived midcircuit postselection and reached terminal measurement.

\begin{table}[!h]
\begin{tabular}{lcll}\hline \hline\\[-.25cm]
\multicolumn{1}{c}{Experiment} & \multicolumn{1}{c}{Total Trotter steps},  $s$ & \multicolumn{1}{c}{Fidelity estimate} & \multicolumn{1}{c}{Total shots ($\times 10^{3}$)}  \\[.3em]  \hline \\[-.6em]
$\I{64}$, LH & [$2$, $4$, $6$, $8$, $10$]  & \qquad[$.69(2)$, $.40(3)$, $.22^{+4}_{-3}$, $.08(2)$, $.03^{+2}_{-1}$] & \qquad[$2$, $2$, $2.2$, $4.4$, $4.8$]  \\[.1em] 
$\I{64}$ & [$2$, $4$, $6$, $8$, $10$]  & \qquad[$.55(2)$, $.27^{+3}_{-2}$, $.08(2)$, $.034^{+8}_{-7}$, $.005^{+4}_{-2}$] & \qquad[$1.6$, $1.6$, $1.8$, $4$, $4.4$] \\[.1em] 
Unencoded, LH & [$2$, $4$, $6$, $8$, $10$]  & \qquad[$.43(2)$, $.14(1)$, $.057^{+9}_{-8}$, $.026^{+6}_{-5}$, $.009^{+5}_{-3}$] & \qquad[$1.2$, $1.2$, $1.5$, $1.8$, $2$]  \\[.1em] 
Unencoded & [$2$, $4$, $6$, $8$, $10$]  & \qquad[$.24^{+2}_{-1}$, $.10(1)$, $.029^{+7}_{-5}$, $.010^{+4}_{-3}$, $.005^{+3}_{-2}$] & \qquad[$0.8$, $0.8$, $0.8$, $0.8$, $0.8$]  \\[.1em] \hline \hline
\end{tabular}
\caption{Fidelity and shot counts across 64 logical qubit $XY$ model quantum simulation experiments at different numbers of total Trotter steps $s$. LH refers to the use of leakage-heralded measurement on Helios. The five values reported for fidelity estimate and total shots correspond to the five different total Trotter steps respectively.}
\label{tab:64lq_results}
\end{table}

\begin{table}[!h]
\begin{tabular}{lll}\hline \hline\\[-.25cm]
\multicolumn{1}{c}{Experiment} &   \multicolumn{1}{c}{\qquad Leakage accept rate} & \multicolumn{1}{l}{\qquad\qquad \qquad Overall accept rate}  \\[.3em]  \hline \\[-.6em]
$\I{64}$, LH & \qquad [$.72^{+1}_{-2}$, $.59(2)$, $.42(2)$, $.33(1)$, $.26(1)$]  & \qquad [$.232^{+10}_{-9}$, $.125^{+8}_{-7}$, $.061(5)$, $.049(3)$, $.032^{+3}_{-2}$]  \\[.1em] 
$\I{64}$ & \multicolumn{1}{c}{\qquad --} & \qquad [$.25(1)$, $.20(1)$, $.137(8)$, .$142^{+6}_{-5}$, $.125(5)$]  \\[.1em] 
Unencoded, LH   & \qquad [$.80(1)$, $.64(1)$, $.50(1)$, $.42(1)$, $.32(1)$]  & \qquad [$.80(1)$, $.64(1)$, $.50(1)$, $.42(1)$, $.32(1)$]  \\[.1em] 
Unencoded &  \multicolumn{1}{c}{\qquad --}  & \multicolumn{1}{c}{\qquad --}  \\[.1em] \hline \hline
\end{tabular}
\caption{Leakage and overall acceptance rates across 64 logical qubit $XY$ model quantum simulation experiments. LH refers to the use of leakage-heralded measurement on Helios and the five values reported for acceptance rates again correspond to the five different total Trotter steps.}
\label{tab:64lq_results2}
\end{table}

\begin{table}[!h]
\begin{tabular}{lll}\hline \hline\\[-.25cm]
\multicolumn{1}{c}{Experiment} &   \multicolumn{1}{c}{\qquad SPAM parameter $A$} & \multicolumn{1}{l}{\qquad Effective average 2Q infidelity $\tilde{\epsilon}_{2Q}$}  \\[.3em]  \hline \\[-.6em]
$\I{64}$, LH & \qquad\qquad$1.4(1)$ & \qquad\qquad\qquad $7.1(6) \times 10^{-4}$ \\[.1em] 
$\I{64}$ & \qquad\qquad$1.5(2)$ & \qquad\qquad \qquad $10.2(8) \times 10^{-4}$ \\[.1em] 
Unencoded, LH & \qquad\qquad $1.16(9)$ &  \qquad\qquad\qquad $10.5(5) \times 10^{-4}$ \\[.1em] 
Unencoded & \qquad\qquad $.69(5)$ & \qquad\qquad \qquad $10.7(4) \times 10^{-4}$ \\[.1em] \hline \hline
\end{tabular}
\caption{Fit parameters $A$, $\tilde{\epsilon}_{2Q}$ obtained from fitting the model $F(N_{\text{gates}}) = A(1-\frac54 \tilde{\epsilon}_{2Q})^{N_{\text{gates}}}$ to the survival probability as a function of number of logical gates in 64 logical qubit $XY$ model quantum simulation experiments.}
\label{tab:fit_params}
\end{table}

\section{Alternative fault-tolerant preparation of resource states}
\label{sec:alternative ft prep}
The FT preparation of diverse resource states is fundamental for QEC: they can be consumed to perform Steane- or Knill-style QEC, to enable long-range interactions between logical qubit blocks via teleportation, to perform non-Clifford gates, or to achieve addressability in high-rate codes.
We present three FT preparation techniques alternative to Goto's FT preparation (Goto) of concatenated hypercubes code~\cite{Goto:2024xd}: a projective method (Projective) analogous to the typical initialization of surface codes, flag at origin (F@O)~\cite{Forlivesi:2025ilq}, and the verification of a non-FT state with a flag circuit optimized with integer linear programming (ILP). 
We then compare the logical error rates and acceptance rates of the resulting circuits for the $\I{8} \circ \I{6}$ under a simple circuit-level noise model.

We analyze the preparation of three logical resource states: a logical $\ket{\Bar{0}}^{\otimes 48}$ state for initializing an algorithm or performing Steane-style QEC in Fig.~\ref{fig:all-zero-state-by-method}, a tensor product of 48 randomly chosen logical $\ket{\Bar{0}}$ and $\ket{\Bar{+}}$ states for injecting addressable logical Hadamard gates as in~\cite{Goto:2024xd} in Fig.~\ref{fig:random-css-state-by-method}, and random CSS state preparation on all logical qubits in Fig.~\ref{fig:exotic-css-state-by-method}.
The noise model inserts bit-flip (phase-flip) channels after the initialization of every qubit in $\ket{0}$ ($\ket{+}$) and before the measurement of every ancillary and flag qubit in the $Z$ ($X$) basis, and inserts a two-qubit depolarizing channel after every two-qubit gate. 
All noise channels have the same error rate $p$.

For the three states we report the logical error rate $P_L$ (left plots), the post-discard rate (middle plot), and the pre-discard rate (right plot). 
The \textit{pre-discard rate} is the probability of discarding the state during preparation, upon detection of one or more faults. This should happen with probability $O(p)$, because all state preparation protocols considered in this work are designed to abort on any fault.
The \textit{post-discard rate} is the probability of discarding a pre-accepted state when performing ideal noiseless decoding on the shots obtained when the state is measured destructively. Such a discard would occur if, for example, the prepared state is used to perform Steane error correction and an uncorrectable error is detected. A post-discard is expected to occur with probability $O(p^2)$, because at least two faults are required to produce a detectable but not correctable error when ideally decoding a distance-4 code.
Finally, the \textit{logical error rate} here is the probability of an error which would result in a large-scale computation continuing with an unintentionally altered state. Such errors are equivalent to a logical operator anticommuting with one that stabilizes the prepared state, or differ from such an operator on one qubit. For a distance-4 code, this is expected to require three faults to occur, and therefore happen with probability $O(p^3)$. 
The relevant logical operators are: $\bar{Z}_i$ for the $\ket{\Bar{0}}^{\otimes 48}$ state, $\bar{Z}_i$ for the logical qubits $i$ prepared in $\ket{\Bar{0}}$ and $\bar{X}_j$ for the logical qubits $j$ prepared in $\ket{\Bar{+}}$ in the random tensor product of these states, and $\bar{Z}_i\bar{Z}_{i+1}$ and $\bar{X}^{\otimes 48}$ for the logical GHZ state.

Looking at the three figures, we first confirm the expected asymptotic scalings for all preparation methods and states.  
For the preparation of $\ket{\Bar{0}}^{\otimes 48}$ in Fig.~\ref{fig:all-zero-state-by-method} we find that the ILP method provides the best logical error rate and post-discard rates, followed by Goto's method.
F@O performs comparably to the Projective method in logical error rate and pre-discard rate, but presents better post-discard rate. 
It is worth noting that Goto and ILP methods are specific to this state while F@O and the Projective methods are general constructions for every logical CSS state.
We can then use them to produce FT state preparation circuits for the random tensor product of $\ket{\Bar{0}}$ and $\ket{\Bar{+}}$ and the GHZ state.
Since these states are stabilized by logical operators in the $Z$ as well as the $X$ bases, we destructively measure in these two bases, obtaining one logical error rate and post-discard rate for each basis (note that the pre-discard rate is independent of this measurement basis).
For both states in Figs.~\ref{fig:random-css-state-by-method} and~\ref{fig:exotic-css-state-by-method} we observe that the natural bias in the Projective method reflects in better results in the $Z$ basis.
The F@O method performs worse in logical error rate and post-discard rate, but provides a significantly better pre-discard rate.

These results showcase the great diversity of QEC strategies that can be implemented on an all-to-all connectivity device, each one being beneficial in some situations over others. For example, the ILP method performs well but requires a large computational problem to be solved for each logical state to be prepared; Goto state preparation performs well for the specific state $\ket{\Bar{0}}^{\otimes k}$ and for higher-distance concatenated codes where ILP cannot yet be used; F@O and Projective methods can be implemented in polynomial time for arbitrary CSS logical states with F@O being more useful when good pre-discard rate is preferable over a good logical error rate.
This is the case when, for example, at some point in the computation the resource state is needed but the available qubit resources on the hardware are so scarce than the entire computation needs to wait until the state preparation attempt is accepted. 

\begin{figure}[htbp]
  \centering
  \includegraphics[width=\columnwidth]{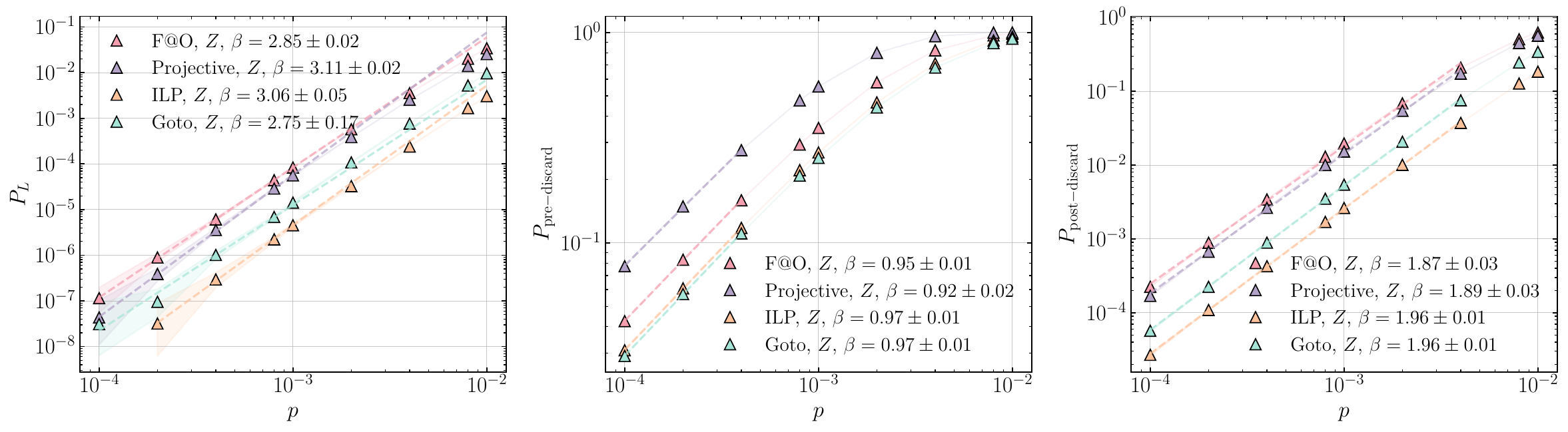}
  \caption{\textbf{Performance of logical state preparation for the encoded all-zero state $\ket{\Bar{0}}^{\otimes 48}$.} We compare the flag-at-origin (Sec.~\ref{sec:flag-at-origin-state-prep}), projective preparation (Sec.~\ref{sec:projection-based-state-prep}), ILP-based (Sec.~\ref{sec:flag-based-prep}) methods, and the circuit obtained using the method described in Ref.~\cite{Goto:2024xd}. The left panel shows the logical error rate conditioned on acceptance, while the middle and right panels respectively show the pre- and post-discard probabilities arising from flag-based and syndrome-based postselection. Data are plotted against physical error rate $p$. Dashed lines indicate power-law fits in the low error-rate regime, and the fitted exponent $\beta$ characterizes the effective scaling of the logical error rate, as pre-(post-) discard rates, and shaded areas indicate $95\%$ Wilson confidence intervals.}
  \label{fig:all-zero-state-by-method}
\end{figure}

\begin{figure}[htbp]
  \centering
  \includegraphics[width=\columnwidth]{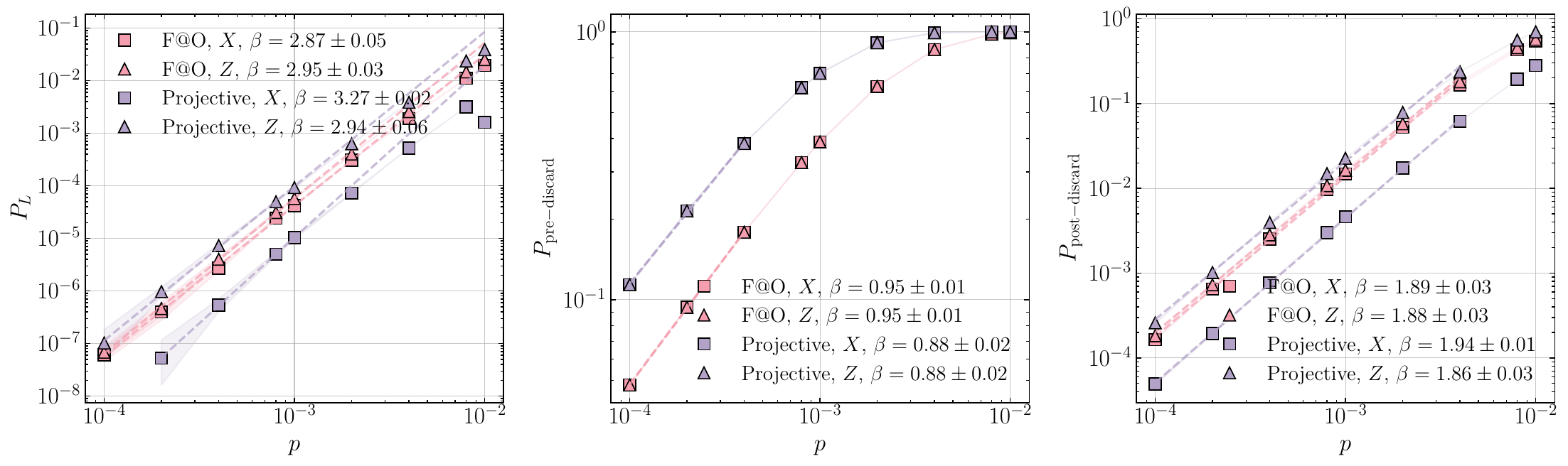}
  \caption{\textbf{Performance of logical state preparation of a random CSS product state benchmarked in $X$ and $Z$ basis.} We compare the flag-at-origin (Sec.~\ref{sec:flag-at-origin-state-prep}) and projective state preparation (Sec.~\ref{sec:projection-based-state-prep}) methods. The left panel shows the logical error rate conditioned on acceptance, while the middle and right panels respectively show the pre- and post-discard probabilities arising from flag-based and syndrome-based postselection. Data are plotted against physical error rate $p$.  Dashed lines indicate power-law fits in the low-error-rate regime, and the fitted exponent $\beta$ characterizes the effective scaling of the logical error rate, as pre-(post-) discard rates, and shaded areas indicate $95\%$ Wilson confidence intervals.}
  \label{fig:random-css-state-by-method}
\end{figure}

\begin{figure}[h!]
  \centering
  \includegraphics[width=\columnwidth]{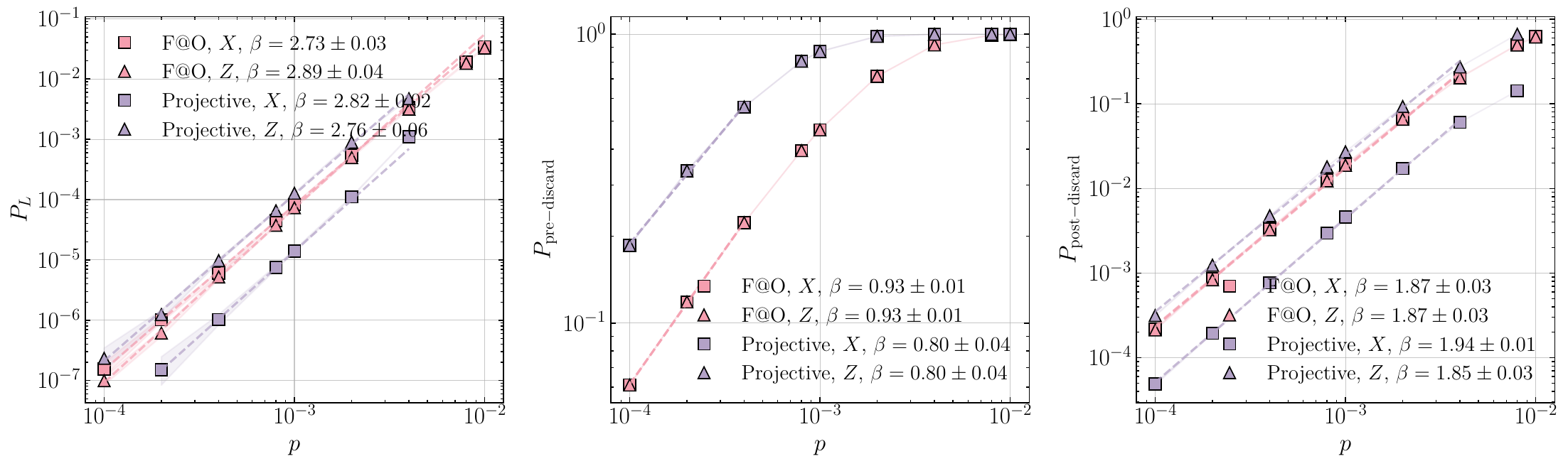}
  \caption{\textbf{Performance of logical state preparation of an exotic CSS state benchmarked in both $X$ and $Z$ bases.} We compare the flag-at-origin (Sec.~\ref{sec:flag-at-origin-state-prep}) and projective state preparation (Sec.~\ref{sec:projection-based-state-prep}) methods. The left panel shows the logical error rate conditioned on acceptance, while the middle and right panels respectively show the pre- and post-discard probabilities arising from flag-based and syndrome-based postselection. Data are plotted against physical error rate $p$.  Dashed lines indicate power-law fits in the low-error-rate regime, and the fitted exponent $\beta$ characterizes the effective scaling of the logical error rate, as pre-(post-) discard rates, and shaded areas indicate $95\%$ Wilson confidence intervals.}
  \label{fig:exotic-css-state-by-method}
\end{figure}

\subsection{Projective preparation}
\label{sec:projection-based-state-prep}

\begin{figure}
    \centering
    \includegraphics[width=0.4\linewidth]{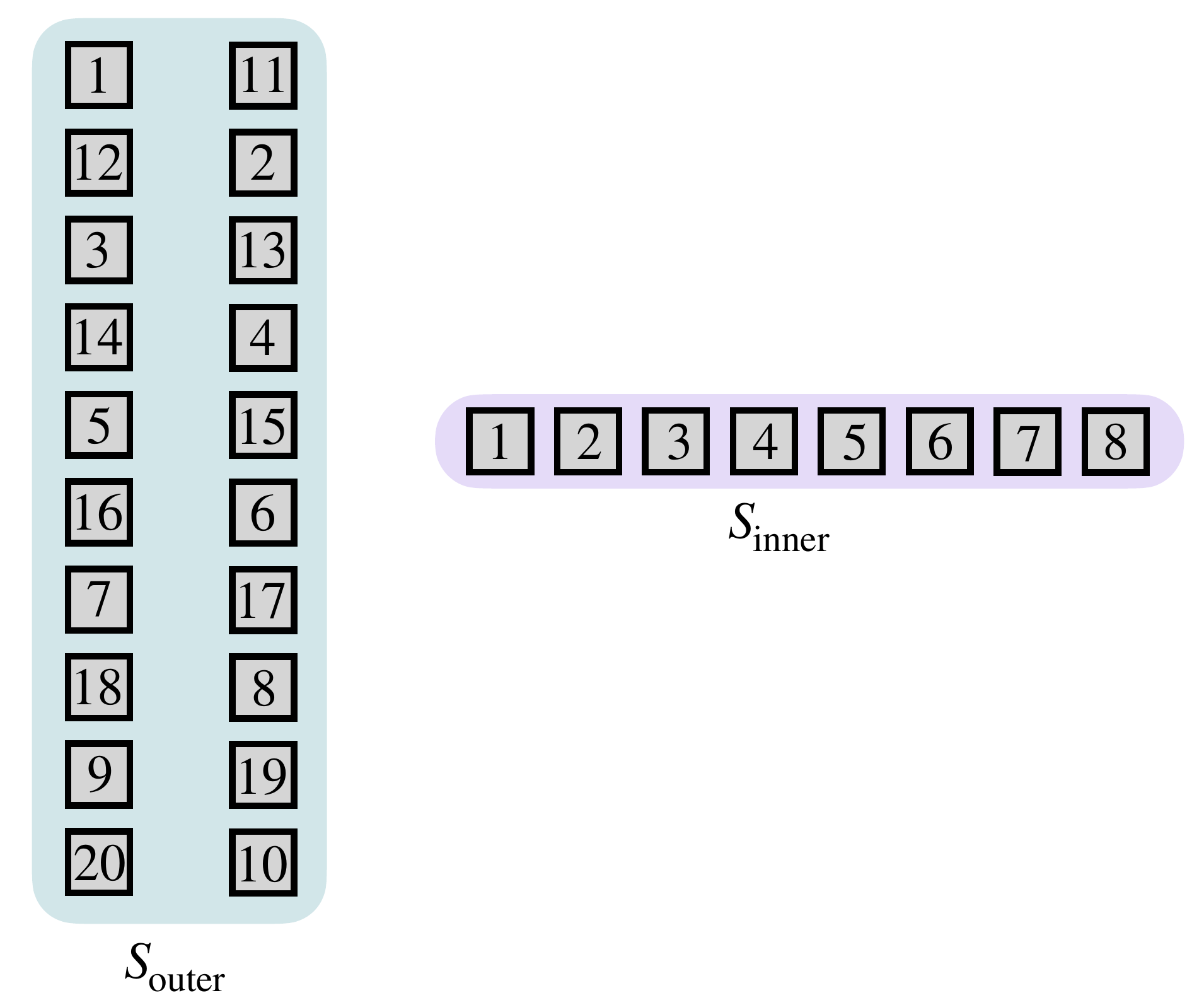}
    \caption{\textbf{Schematic stabilizer measurement schedule for the concatenated iceberg code in Sec.~\ref{sec:projection-based-state-prep}.} We show one representative outer stabilizer (left) and one inner stabilizer (right). Numbers label discrete CX time steps between stabilizer ancilla and data qubits, which are executed in parallel across all inner (respectively, outer) stabilizers. The zigzag measurement ordering of the outer stabilizer enables FT syndrome extraction using a single flag ancilla to detect hook errors.}
    \label{fig:projective_stabilizer_schedule}
\end{figure}

We fault-tolerantly prepare encoded CSS states using a postselected projective measurement scheme. This method consists of initializing the data qubits of the code in a product state of a chosen basis, and projecting the system onto the desired logical state by repeated measurement of an appropriate set of commuting parity checks. In our setting, we apply this method by measuring stabilizers of the concatenated iceberg code in the basis opposite to the initial product state, as well as additional commuting Pauli operators, referred to as state stabilizers, which fix the logical sector and uniquely define the target logical state. The concatenated iceberg stabilizer generators separate into two sets, $S_{\textrm{inner}}$ and $S_{\textrm{outer}}$, which, respectively, have support at lower and higher levels of concatenation. The former each have support on a single lower-level iceberg code, while the outer stabilizers have larger support on multiple lower-level blocks, as shown in Fig.~\ref{fig:encoding}a in the main text. 

Stabilizer measurements are implemented using minimal one-flag gadgets. Each check in $S_{\textrm{inner}}$ and $S_{\textrm{outer}}$ is measured using a dedicated syndrome ancilla and a single flag ancilla, where two additional CX gates couple the syndrome ancilla to the flag ancilla at the beginning and end of the syndrome extraction circuit, with the control (target) depending on the CSS stabilizer basis. The purpose of this construction is to ensure that the probability of an accepted logical failure is suppressed to second order in the physical error rate, i.e. scales as $O(p^2)$. Indeed, any single fault either triggers a flag, leading to immediate rejection of the shot, or results in an error that remains within the correctable regime of the preparation protocol, such that accepted logical errors have no $O(p)$ contribution. Moreover, it is not necessary to individually detect all two-fault events in this postselected scheme, since for distance-4 FT preparation it suffices that two-fault events do not cause accepted at $O(p^2)$. Equivalently, a shot is accepted only if its syndrome is compatible with the desired logical state, within a decoding radius of one physical error. Any undetected two-fault event producing a syndrome outside of this radius is considered non-correctable and the shot is rejected. This prevents the failure mechanism where a two-fault event produces a syndrome that an ideal decoder would interpret as a low-weight correctable error in the wrong logical coset.

The validity of the flagged circuit for the inner stabilizers follows from the fact that any of these is a global weight-$n_{\text{in}}$ check of an iceberg block, and different blocks are disjoint. Hence, all inner stabilizers can be measured in parallel. A logical operator of the concatenated iceberg code can have support on only two qubits per inner block, and a single syndrome ancilla error can cause a hook error within the block. The flag ancilla detects any hook that produces a weight $\geq$ 1 error, such that accepted outcomes are at worst weight-one within each block. The outer (higher-level) stabilizers have support that spans multiple inner (lower-level) blocks, and a naive schedule can lead to malignant hook errors. To avoid this, we choose an interleaved syndrome ancilla-data qubit CX ordering, in a zigzag pattern across the support of each stabilizer, as shown in Fig.~\ref{fig:projective_stabilizer_schedule}. This ordering spreads potential hook endpoints across blocks and ensures that any unflagged two-fault events cannot be creating an error that is indistinguishable from a single-qubit error up to a logical operator. 

To prepare a desired logical state after fixing the codespace, we additionally perform projective measurements of a chosen set of mutually commuting Pauli operators, referred to as the state stabilizers. Such operators commute with all stabilizers of the code, and their joint eigenvalues uniquely define the target logical sector. Operationally, each of these logical operators is measured using a FT flagged Pauli-product measurement circuit, implemented after stabilizer extraction, where each flagged measurement schedule is taken from the library of flag gadgets generated in~\cite{Forlivesi:2025ilq}, where the number of flag ancillae per operator is parametrized by the target code distance. In order to simplify benchmarking for comparison with other state preparation methods, we do not consider idling noise, and assume that these (gauged) logical operator measurements are realized sequentially. This choice isolates the intrinsic fault tolerance of the measurement gadgets themselves. We emphasize that more efficient FT schedules can be constructed by interleaving commuting logical operator measurements, both within this set and potentially with the stabilizer measurements, in order to reduce circuit depth, flag-ancillary overhead, and the impact of idling noise. In particular, the flagged syndrome extraction circuits considered here possess a constrained interaction structure, in which the roles of certain qubit types, i.e. which qubits act exclusively as controls or targets in multi-qubit gates, are fixed by construction. This structure could be exploited to enable greater parallelization in the schedule. We leave a systematic exploration of such schedule optimizations for future work. 

We simulate this logical state preparation protocol using \texttt{Stim} for two-qubit depolarizing noise of strength $p$. We start by initializing the data qubits in the computational basis product state $\ket{0}^{\otimes 80}$. FT state preparation is simulated using repeated flagged stabilizer measurements in the $Z$ basis and, when applicable, flagged logical-operator measurements, followed by postselection onto the desired syndrome outcomes. Three rounds of measurements are performed. To benchmark logical performance in both Pauli bases, we use projective Pauli-product measurements (PPM) at the end of the circuit. When the benchmarking basis coincides with the preparation basis, stabilizers and logical observables can be measured directly. Since we benchmark both $X$ and $Z$ bases in a single numerical experiment, the stabilizers and logical observables of the conjugate basis are measured via PPM projections, which naturally commute with the stabilizer group. This approach enables the extraction of logical error rates in both $X$ and $Z$ bases from a single preparation experiment. In Figures~\ref{fig:all-zero-state-by-method}, \ref{fig:random-css-state-by-method}, and \ref{fig:exotic-css-state-by-method}, we report the logical error rate for the preparation of (i) the encoded zero state, (ii) a random CSS state drawn from an arbitrary $\ket{\Bar{0}} / \ket{\Bar{+}}$ distribution, and (iii) an exotic CSS state comprising a mixture of single-qubit CSS states and Bell pairs. The pre-discard and post-discard rates, shown in panels (b) and (c), are defined as the discard probabilities after flag detection, and after both flag- and syndrome-based postselection.

\subsection{Flag-at-origin preparation}
\label{sec:flag-at-origin-state-prep}
We use the flag-at-origin (F@O) method introduced in \cite{Forlivesi:2025ilq}
as an alternative method for preparing the logical states of the concatenated iceberg code. In this subsection, we will give a brief overview of the F@O method, and describe the modifications made to the original protocol to optimize the resulting circuit further to use the number of qubits available in the Helios trapped-ion quantum processor.

For a CSS state with $n_X$ ($n_Z$) number of $X$-type ($Z$-type) independent stabilizer generators, one can find a non-fault-tolerant bipartite unitary preparation circuit $C$ which acts on $|0\rangle^{\otimes n_Z}|+\rangle^{\otimes n_X}$ and composed of only CX gates.
Let $Q_X$ and $Q_Z$ denote the set of qubits initialized in $|+\rangle$ and $|0\rangle$ states respectively.
The preparation circuit $C$ is bipartite in the sense that for all CX gates, the control qubits are in $Q_X$ and the target qubits are $Q_Z$.
This bipartite structure is closely related to the fact that any CSS states can be represented by a graph state $\ket{G}$ with $G$ bipartite~\cite{Forlivesi:2025ilq}.
The bipartite structure of $C$ ensures that errors will not propagate within $Q_X$ (or within $Q_Z$), and any $X$ error can only propagate once from $Q_X$ to a qubit in $Q_Z$ (and will not propagate further), and similarly any $Z$ error can only propagate once from a qubit in $Q_Z$ to a qubit in $Q_X$.
Importantly, this straightforward error propagation structure allows one to inspect the fault tolerance of the sub-circuit connected to each qubit \emph{separately} and therefore attach a flag gadget to that qubit (i.e. to use flag qubits to prevent harmful errors from propagating undetected), making the corresponding sub-circuit fault-tolerant (see Ref.~\cite{Forlivesi:2025ilq} for further details).
We use $C_\text{FT}$ to denote the FT circuit that results from attaching flag gadgets to each qubit in $Q_X \cup Q_Z$.

Since the concatenated iceberg code $\I{8}\circ\I{6}$ is CSS we may use F@O to prepare any CSS logical state fault-tolerantly. We use the following additional optimization techniques to ensure $C_\text{FT}$ can be executed on the Helios trapped-ion quantum processor:
\begin{itemize}
    \item 
    In the original F@O technique in Ref.~\cite{Forlivesi:2025ilq}, the same gadgets are used for $d=4$ and $d=5$ codes. This ensures any fault pairs in the preparation circuit that result in weight $w > 2$ errors on the data qubits are flagged and discarded. This is sub-optimal since an ideal decoder for the $\I{8}\circ \I{6}$ code discards any shot with $w \ge 2$ errors on data qubits anyway. Therefore we only flag a fault pair if it propagates to an error that the ideal decoder mistakenly corrects instead of discards, i.e. if it propagates to an error with the same syndrome as a single-qubit error.
    \item 
    Ideally, one would apply the gates in $C_\text{FT}$ in parallel. However, doing so requires more qubits than what is available on the Helios trapped-ion quantum processor. Therefore we have to run flag gadgets serially. Different choices of which flag gadgets to run first may then result in different total number of simultaneously allocated qubits. In \cite{Forlivesi:2025ilq}, it was suggested to start by running the gadgets attached to qubits in $Q_X$ and run the gadgets in $Q_Z$ accordingly. However, for states with $|Q_Z| \gg |Q_X|$, such as the logical all-zero state $|\bar{0}\rangle^{\otimes 48}$, it is better to start with gadgets attached to qubits in $Q_Z$ and build the gadgets for qubits in $Q_X$ accordingly. Moreover, within the qubits in $Q_Z$, we begin with the flag gadget of the qubit with the largest number of CX gates attached to it in $C$, followed by the one with the second-to-largest, and so on. This ensures that we are done with flag gadgets which potentially need more flag qubits first, when we have more qubits available, before we entangle more data qubits to the state we are building.
    \item 
    The circuits in Ref.~\cite{Forlivesi:2025ilq} were found by first finding a circuit $C$ which has the lowest number of CX gates (among some fixed number of randomly built $C$) and then, using the fact that all CX gates in $C$ commute with each other, finding the CX ordering which results in a $C_\text{FT}$ with the lowest number of required qubits. Since the required number of qubits in $C_\text{FT}$ is not simply a function of the number of CX gates in $C$, we combine these two steps together; we find a large number of $C$ (of the order of $10^3$) with relatively small number of CX gates, and build $C_\text{FT}$ for each of them and choose the one with the smallest number of required simultaneous qubits.  
    \end{itemize}

\subsection{Flag-based $\ket{\overline{0}}^{\otimes 48}$ state preparation}
\label{sec:flag-based-prep}
Using the insight that all operations required for $\ket{\overline{0}}^{48}$ state preparation in the iceberg code are themselves transversal in the iceberg code, a non-FT state preparation circuit for the two-layer concatenated iceberg code can easily be created, see Figure \ref{fig:non_ft_zero_prep}.
\begin{figure}[!htpb]
\centering
\scalebox{0.5}{\input{bare_zero_prep.tikz}}
\caption{\textbf{Depth-$9$ non-FT circuit preparing $\ket{\overline{0}}^{\otimes 48}$ in the $\I{8}\circ \I{6}$ code.} The circuit contains 142 \CX{} gates.}
\label{fig:non_ft_zero_prep}
\end{figure}

Given the limited error propagation in this circuit, flag-based fault tolerance \cite{chao2018quantum} is an effective method for deriving an equivalent FT circuit.
We produce a preparation circuit using an automated method for producing flag FT circuits for CSS codes \cite{CrigerInPreparation}, see Figure \ref{fig:ft_zero_prep}.
\clearpage
\begin{turnpage}
\begin{figure}[!htpb]
\centering
\scalebox{0.35}{\input{ft_zero_prep.tikz}}
\caption{\textbf{Flag FT circuit preparing $\ket{\overline{0}}^{\otimes 48}$ in the $\I{8}\circ \I{6}$ code.} The circuit contains 344 \CX{} gates and 67 measurements.}
\label{fig:ft_zero_prep}
\end{figure}
\end{turnpage}
\clearpage

\section{Decoding Goto's state preparation}
\label{sec:decoding}

Here we discuss more details on the decoding of the state preparation circuit for the $\I{8} \circ \I{6}$ concatenated code that was introduced in Ref.~\cite{Goto:2024xd}.
First note that products of four operators $X_{i,j} X_{k,j} X_{k,l} X_{i,l}$ (and $Z_{i,j}Z_{k,j}Z_{k,l}Z_{i,l}$) that are, using the conventions from Fig.~\ref{fig:encoding}a, arranged in a rectangle, commute with all stabilizers, thus, they must be products of logical operators and stabilizers.
Since all minimum-weight stabilizers have at least weight $\min (n_1,2 n_2)=8$, the products $X_{i,j} X_{k,j} X_{k,l} X_{i,l}$ (and $Z_{i,j}Z_{k,j}Z_{k,l}Z_{i,l}$) must contain at least one logical operator.
Hence, each two-qubit $X$ error is connected to at least one other two-qubit error by some combination of logical operators, i.e., for each syndrome corresponding to a two-qubit error there are at least two same-weight errors in different equivalence classes.
We thus do not attempt to correct any two-qubit errors.
Similar arguments hold for three-qubit errors: for each three-qubit error, we can find at least another three-qubit error with the same syndrome, but in a different equivalence class; additionally, some three-qubit errors share the same syndrome with one-qubit errors in a different equivalence class.
Hence, we do not attempt to correct any syndrome corresponding to weight-two-or-higher errors.

We simulate the quantum circuit using \texttt{Stim}~\cite{Gidney:2021wm} for two-qubit depolarizing noise of strength $p$ after each two-qubit gate. The error rates for qubit initialization and measurement (bit flip or phase flip errors) also equal $p$.
We measure each final state in its $Z$ basis to reconstruct its $S^Z$ syndrome $s$ and logical equivalence class $q$. Since the code has 16 $S^Z$ stabilizers and 48 logical qubits, we consider $2^{16}$ different syndromes $s$ and $2^{48}$ equivalence classes $q$.

\begin{figure}
  \includegraphics[scale=1]{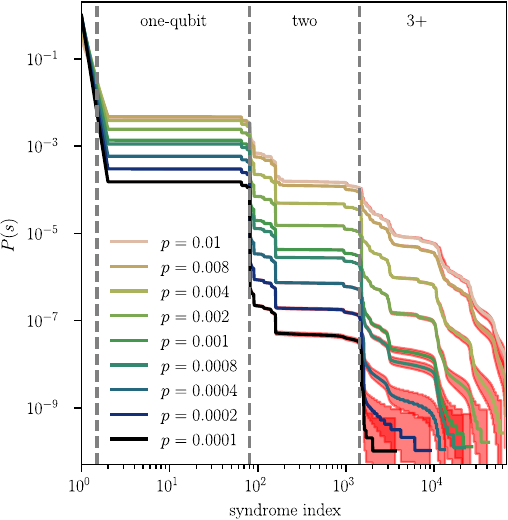}
  \caption{\textbf{Sampled syndrome probabilities for different error rates $p$.} The syndrome may be caused by errors of different weights; the vertical dashed lines separate regimes whose lowest-weight errors are of weight zero to four. The syndromes are ordered by $P(s)$ in descending order. For each $p$, we take $10^{10}$ samples.}
  \label{fig:syndrome_probabilities}
\end{figure}

We show the distribution of syndromes in Fig.~\ref{fig:syndrome_probabilities}. For each $p$, we sample $10^{10}$ different circuits and use the results to estimate the syndrome probability $P(s)$. We categorize all syndromes by their minimum-weight error; note that higher-weight errors (connected via logicals and stabilizers) may cause the same syndrome.
The gray dashed lines in Fig.~\ref{fig:syndrome_probabilities} separate syndromes with minimum-weight zero-, one-, two-, and higher-qubit errors. $P(s)$ is not constant in each regime since (i) errors are not caused by code capacity noise, but circuit-level noise, and (ii) some two-weight error syndromes can be caused by two inequivalent weight-two errors, while others can be caused by eight or ten inequivalent weight-two errors (and the latter thus occur with higher probability).

\begin{figure}
  \includegraphics[scale=1]{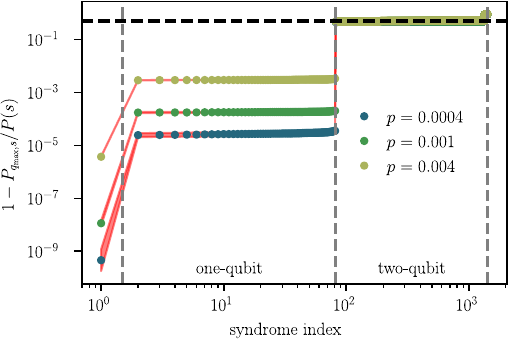}
  \caption{\textbf{Probability $P_{q,s}$ that an error is characterized by an equivalence class $q$ and syndrome $s$, divided by probability of the syndrome $P(s)=\sum_q P_{q,s}$.} Dots denote data points, and the red shaded area the 95\,\% Wilson confidence interval. For most two-qubit errors, $P_{q_\text{max},s}/P(s) \lesssim 1/2$ (black dashed line denotes $1/2$).}
  \label{fig:equivalence_class_probabilities}
\end{figure}

As argued before, we do not attempt to correct any syndrome caused by weight-two errors.
In Fig.~\ref{fig:equivalence_class_probabilities}, we provide numerical evidence that supports our arguments.
We show the ratio of the (numerically sampled) maximum-likelihood equivalence class $P_{q_\text{max},s}$ and the total syndrome probability $P(s) = \sum_q P_{q,s}$.
This ratio is a measure for the confidence in the chosen correction: when the ratio approaches one, we have a high confidence that the syndrome $s$ was caused by an error in class $q$; if another equivalence class is more likely, the ratio will drop below $1/2$.
We show $1-P_{q_\text{max},s}/P(s)$ in Fig.~\ref{fig:equivalence_class_probabilities}, which for the trivial syndrome grows as $p^4$, and for syndromes caused by weight-one errors increases as $p^2$.
For syndromes caused by two-qubit errors, as argued before, there are at least two same-weight errors in different equivalence classes, hence, assuming that for circuit-level noise both errors have similar probabilities, $P_{q_\text{max},s}/P(s) \lesssim 1/2$, as evident from Fig.~\ref{fig:equivalence_class_probabilities}, where the black dashed line denotes $P_{q_\text{max},s}/P(s)=1/2$.

\begin{figure}
  \includegraphics[scale=1]{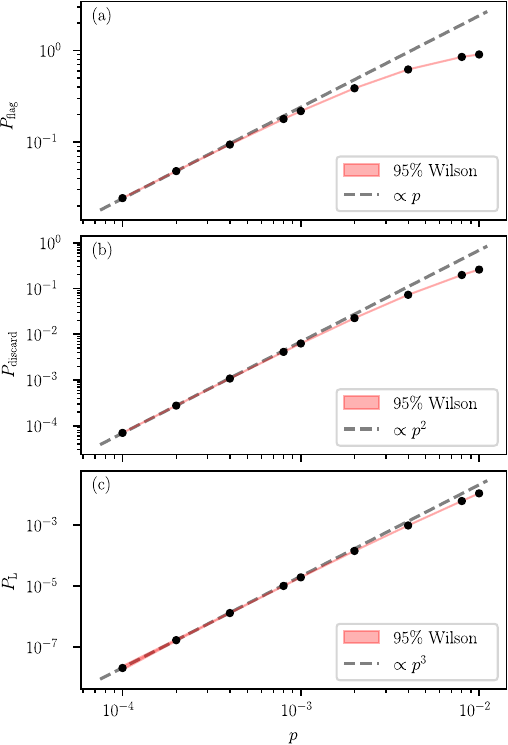}
  \caption{\textbf{Error rates for the state preparation protocol described in Sec.~\ref{sec:decoding}.} The error rates are obtained via Stim~\cite{Gidney:2021wm} simulations, where $p$ is the strength of the depolarizing noise. (a) Rate of flagged circuit runs, $P_\text{flag}$, (b) discard rate of syndromes that are not compatible with weight-zero or weight-one errors, $P_\text{post}$,  and (c) final logical error rate $P_\text{L}$. The black dots represent the data, the red shaded region the 95\,\% Wilson confidence interval, and the gray dashed line a power law as a guide for the eyes.}
  \label{fig:error_rates}
\end{figure}

In Fig.~\ref{fig:error_rates}, we show three different resulting error rates:
in Fig.~\ref{fig:error_rates}a, we show $P_\text{flag}$, the sampled rate of flagged circuit runs that are discarded. Only the postselected fraction $1-P_\text{flag}$ of runs will be used for state preparation. Since single-qubit errors can cause flags to fire, we expect that $P_\text{flag}$ grows linearly with $p$ for small error rates. We show a gray dashed line $\propto p$ as a guide for the eye.
We then discard all syndromes compatible with weight-two-and-higher errors. We show the discard rate $P_\text{post}$ in panel (b), and include $\propto p^2$ as a guide for the eye since two-qubit errors are the dominant contribution to $P_\text{post}$, which occur with probability $\propto p^2$.
Finally, we show the resulting logical error rate $P_\text{L}$ in panel (c): logical errors occur with probability $p^3$ as they correspond to three-qubit errors causing the same syndrome as one-qubit errors in a different equivalence class.

\clearpage

\bibliography{references}

\begin{thebibliography}{51}%
\makeatletter
\providecommand \@ifxundefined [1]{%
 \@ifx{#1\undefined}
}%
\providecommand \@ifnum [1]{%
 \ifnum #1\expandafter \@firstoftwo
 \else \expandafter \@secondoftwo
 \fi
}%
\providecommand \@ifx [1]{%
 \ifx #1\expandafter \@firstoftwo
 \else \expandafter \@secondoftwo
 \fi
}%
\providecommand \natexlab [1]{#1}%
\providecommand \enquote  [1]{``#1''}%
\providecommand \bibnamefont  [1]{#1}%
\providecommand \bibfnamefont [1]{#1}%
\providecommand \citenamefont [1]{#1}%
\providecommand \href@noop [0]{\@secondoftwo}%
\providecommand \href [0]{\begingroup \@sanitize@url \@href}%
\providecommand \@href[1]{\@@startlink{#1}\@@href}%
\providecommand \@@href[1]{\endgroup#1\@@endlink}%
\providecommand \@sanitize@url [0]{\catcode `\\12\catcode `\$12\catcode `\&12\catcode `\#12\catcode `\^12\catcode `\_12\catcode `\%12\relax}%
\providecommand \@@startlink[1]{}%
\providecommand \@@endlink[0]{}%
\providecommand \url  [0]{\begingroup\@sanitize@url \@url }%
\providecommand \@url [1]{\endgroup\@href {#1}{\urlprefix }}%
\providecommand \urlprefix  [0]{URL }%
\providecommand \Eprint [0]{\href }%
\providecommand \doibase [0]{https://doi.org/}%
\providecommand \selectlanguage [0]{\@gobble}%
\providecommand \bibinfo  [0]{\@secondoftwo}%
\providecommand \bibfield  [0]{\@secondoftwo}%
\providecommand \translation [1]{[#1]}%
\providecommand \BibitemOpen [0]{}%
\providecommand \bibitemStop [0]{}%
\providecommand \bibitemNoStop [0]{.\EOS\space}%
\providecommand \EOS [0]{\spacefactor3000\relax}%
\providecommand \BibitemShut  [1]{\csname bibitem#1\endcsname}%
\let\auto@bib@innerbib\@empty
\bibitem [{\citenamefont {Yoder}\ \emph {et~al.}(2025)\citenamefont {Yoder} \emph {et~al.}}]{yoder2025tour}%
  \BibitemOpen
  \bibfield  {author} {\bibinfo {author} {\bibfnamefont {T.~J.}\ \bibnamefont {Yoder}} \emph {et~al.},\ }\href {https://arxiv.org/abs/2506.03094} {\bibinfo {title} {Tour de gross: A modular quantum computer based on bivariate bicycle codes}} (\bibinfo {year} {2025}),\ \Eprint {https://arxiv.org/abs/2506.03094} {arXiv:2506.03094 [quant-ph]} \BibitemShut {NoStop}%
\bibitem [{\citenamefont {Yoshida}\ \emph {et~al.}(2025)\citenamefont {Yoshida}, \citenamefont {Tamiya},\ and\ \citenamefont {Yamasaki}}]{Yoshida2025concatenate}%
  \BibitemOpen
  \bibfield  {author} {\bibinfo {author} {\bibfnamefont {S.}~\bibnamefont {Yoshida}}, \bibinfo {author} {\bibfnamefont {S.}~\bibnamefont {Tamiya}},\ and\ \bibinfo {author} {\bibfnamefont {H.}~\bibnamefont {Yamasaki}},\ }\bibfield  {title} {\bibinfo {title} {Concatenate codes, save qubits},\ }\href {https://doi.org/10.1038/s41534-025-01035-8} {\bibfield  {journal} {\bibinfo  {journal} {npj Quantum Information}\ }\textbf {\bibinfo {volume} {11}},\ \bibinfo {pages} {88} (\bibinfo {year} {2025})}\BibitemShut {NoStop}%
\bibitem [{\citenamefont {Berthusen}\ and\ \citenamefont {Durso-Sabina}(2025)}]{berthusen2025simple}%
  \BibitemOpen
  \bibfield  {author} {\bibinfo {author} {\bibfnamefont {N.}~\bibnamefont {Berthusen}}\ and\ \bibinfo {author} {\bibfnamefont {E.}~\bibnamefont {Durso-Sabina}},\ }\href {https://arxiv.org/abs/2510.18753} {\bibinfo {title} {Simple logical quantum computation with concatenated symplectic double codes}} (\bibinfo {year} {2025}),\ \Eprint {https://arxiv.org/abs/2510.18753} {arXiv:2510.18753 [quant-ph]} \BibitemShut {NoStop}%
\bibitem [{\citenamefont {Gottesman}(1997)}]{gottesmanPHD}%
  \BibitemOpen
  \bibfield  {author} {\bibinfo {author} {\bibfnamefont {D.}~\bibnamefont {Gottesman}},\ }\emph {\bibinfo {title} {{Stabilizer Codes and Quantum Error Correction}}},\ \href {https://doi.org/10.7907/rzr7-dt72} {Ph.D. thesis},\ \bibinfo  {school} {California Institute of Technology} (\bibinfo {year} {1997}),\ \Eprint {https://arxiv.org/abs/quant-ph/9705052} {arXiv:quant-ph/9705052} \BibitemShut {NoStop}%
\bibitem [{\citenamefont {Chao}\ and\ \citenamefont {Reichardt}(2018)}]{chao2018quantum}%
  \BibitemOpen
  \bibfield  {author} {\bibinfo {author} {\bibfnamefont {R.}~\bibnamefont {Chao}}\ and\ \bibinfo {author} {\bibfnamefont {B.~W.}\ \bibnamefont {Reichardt}},\ }\bibfield  {title} {\bibinfo {title} {Quantum error correction with only two extra qubits},\ }\href@noop {} {\bibfield  {journal} {\bibinfo  {journal} {Physical review letters}\ }\textbf {\bibinfo {volume} {121}},\ \bibinfo {pages} {050502} (\bibinfo {year} {2018})}\BibitemShut {NoStop}%
\bibitem [{\citenamefont {Self}\ \emph {et~al.}(2024)\citenamefont {Self}, \citenamefont {Benedetti},\ and\ \citenamefont {Amaro}}]{Self:2024py}%
  \BibitemOpen
  \bibfield  {author} {\bibinfo {author} {\bibfnamefont {C.~N.}\ \bibnamefont {Self}}, \bibinfo {author} {\bibfnamefont {M.}~\bibnamefont {Benedetti}},\ and\ \bibinfo {author} {\bibfnamefont {D.}~\bibnamefont {Amaro}},\ }\bibfield  {title} {\bibinfo {title} {Protecting expressive circuits with a quantum error detection code},\ }\href {https://doi.org/10.1038/s41567-023-02282-2} {\bibfield  {journal} {\bibinfo  {journal} {Nature Physics}\ }\textbf {\bibinfo {volume} {20}},\ \bibinfo {pages} {219} (\bibinfo {year} {2024})}\BibitemShut {NoStop}%
\bibitem [{\citenamefont {Goto}(2024)}]{Goto:2024xd}%
  \BibitemOpen
  \bibfield  {author} {\bibinfo {author} {\bibfnamefont {H.}~\bibnamefont {Goto}},\ }\bibfield  {title} {\bibinfo {title} {High-performance fault-tolerant quantum computing with many-hypercube codes},\ }\bibfield  {journal} {\bibinfo  {journal} {Science Advances}\ }\textbf {\bibinfo {volume} {10}},\ \href {https://doi.org/10.1126/sciadv.adp6388} {10.1126/sciadv.adp6388} (\bibinfo {year} {2024})\BibitemShut {NoStop}%
\bibitem [{\citenamefont {Ransford}\ \emph {et~al.}(2025)\citenamefont {Ransford} \emph {et~al.}}]{helios}%
  \BibitemOpen
  \bibfield  {author} {\bibinfo {author} {\bibfnamefont {A.}~\bibnamefont {Ransford}} \emph {et~al.},\ }\href {https://arxiv.org/abs/2511.05465} {\bibinfo {title} {Helios: A 98-qubit trapped-ion quantum computer}} (\bibinfo {year} {2025}),\ \Eprint {https://arxiv.org/abs/2511.05465} {arXiv:2511.05465} \BibitemShut {NoStop}%
\bibitem [{\citenamefont {Akahoshi}\ \emph {et~al.}(2024)\citenamefont {Akahoshi}, \citenamefont {Maruyama}, \citenamefont {Oshima}, \citenamefont {Sato},\ and\ \citenamefont {Fujii}}]{Fujii2024star}%
  \BibitemOpen
  \bibfield  {author} {\bibinfo {author} {\bibfnamefont {Y.}~\bibnamefont {Akahoshi}}, \bibinfo {author} {\bibfnamefont {K.}~\bibnamefont {Maruyama}}, \bibinfo {author} {\bibfnamefont {H.}~\bibnamefont {Oshima}}, \bibinfo {author} {\bibfnamefont {S.}~\bibnamefont {Sato}},\ and\ \bibinfo {author} {\bibfnamefont {K.}~\bibnamefont {Fujii}},\ }\bibfield  {title} {\bibinfo {title} {Partially fault-tolerant quantum computing architecture with error-corrected clifford gates and space-time efficient analog rotations},\ }\href {https://doi.org/10.1103/PRXQuantum.5.010337} {\bibfield  {journal} {\bibinfo  {journal} {PRX Quantum}\ }\textbf {\bibinfo {volume} {5}},\ \bibinfo {pages} {010337} (\bibinfo {year} {2024})}\BibitemShut {NoStop}%
\bibitem [{\citenamefont {Ibe}\ \emph {et~al.}(2025)\citenamefont {Ibe}, \citenamefont {Hirano}, \citenamefont {Ozu}, \citenamefont {Kawakubo},\ and\ \citenamefont {Fujii}}]{Fujii2025}%
  \BibitemOpen
  \bibfield  {author} {\bibinfo {author} {\bibfnamefont {Y.}~\bibnamefont {Ibe}}, \bibinfo {author} {\bibfnamefont {Y.}~\bibnamefont {Hirano}}, \bibinfo {author} {\bibfnamefont {Y.}~\bibnamefont {Ozu}}, \bibinfo {author} {\bibfnamefont {T.}~\bibnamefont {Kawakubo}},\ and\ \bibinfo {author} {\bibfnamefont {K.}~\bibnamefont {Fujii}},\ }\href {https://arxiv.org/abs/2510.18652} {\bibinfo {title} {Measurement-based fault-tolerant quantum computation on high-connectivity devices: A resource-efficient approach toward early ftqc}} (\bibinfo {year} {2025}),\ \Eprint {https://arxiv.org/abs/2510.18652} {arXiv:2510.18652 [quant-ph]} \BibitemShut {NoStop}%
\bibitem [{\citenamefont {Jin}\ \emph {et~al.}()\citenamefont {Jin} \emph {et~al.}}]{iceberg-beyond-the-tip}%
  \BibitemOpen
  \bibfield  {author} {\bibinfo {author} {\bibfnamefont {Y.}~\bibnamefont {Jin}} \emph {et~al.},\ }\bibfield  {title} {\bibinfo {title} {{Iceberg Beyond the Tip: Co-Compilation of a Quantum Error Detection Code and a Quantum Algorithm}},\ }\Eprint {https://arxiv.org/abs/2504.21172} {arXiv:2504.21172} \BibitemShut {NoStop}%
\bibitem [{\citenamefont {Yamamoto}\ \emph {et~al.}(2025)\citenamefont {Yamamoto} \emph {et~al.}}]{Yamamoto2025molecular}%
  \BibitemOpen
  \bibfield  {author} {\bibinfo {author} {\bibfnamefont {K.}~\bibnamefont {Yamamoto}} \emph {et~al.},\ }\href {https://arxiv.org/abs/2505.09133} {\bibinfo {title} {Quantum error-corrected computation of molecular energies}} (\bibinfo {year} {2025}),\ \Eprint {https://arxiv.org/abs/2505.09133} {arXiv:2505.09133 [quant-ph]} \BibitemShut {NoStop}%
\bibitem [{\citenamefont {Goto}(2025)}]{goto2025optimized}%
  \BibitemOpen
  \bibfield  {author} {\bibinfo {author} {\bibfnamefont {H.}~\bibnamefont {Goto}},\ }\href {https://arxiv.org/abs/2512.00561} {\bibinfo {title} {Optimized many-hypercube codes toward lower logical error rates and earlier realization}} (\bibinfo {year} {2025}),\ \Eprint {https://arxiv.org/abs/2512.00561} {arXiv:2512.00561 [quant-ph]} \BibitemShut {NoStop}%
\bibitem [{\citenamefont {Nakai}\ and\ \citenamefont {Goto}(2026)}]{Nakai_2026}%
  \BibitemOpen
  \bibfield  {author} {\bibinfo {author} {\bibfnamefont {R.}~\bibnamefont {Nakai}}\ and\ \bibinfo {author} {\bibfnamefont {H.}~\bibnamefont {Goto}},\ }\bibfield  {title} {\bibinfo {title} {Subsystem many-hypercube codes: High-rate concatenated codes with low-weight syndrome measurements},\ }\bibfield  {journal} {\bibinfo  {journal} {Physical Review Applied}\ }\textbf {\bibinfo {volume} {25}},\ \href {https://doi.org/10.1103/xbzn-vn37} {10.1103/xbzn-vn37} (\bibinfo {year} {2026})\BibitemShut {NoStop}%
\bibitem [{\citenamefont {Combes}\ \emph {et~al.}(2017)\citenamefont {Combes}, \citenamefont {Granade}, \citenamefont {Ferrie},\ and\ \citenamefont {Flammia}}]{Combes2017}%
  \BibitemOpen
  \bibfield  {author} {\bibinfo {author} {\bibfnamefont {J.}~\bibnamefont {Combes}}, \bibinfo {author} {\bibfnamefont {C.}~\bibnamefont {Granade}}, \bibinfo {author} {\bibfnamefont {C.}~\bibnamefont {Ferrie}},\ and\ \bibinfo {author} {\bibfnamefont {S.~T.}\ \bibnamefont {Flammia}},\ }\href@noop {} {\bibinfo {title} {Logical randomized benchmarking}} (\bibinfo {year} {2017}),\ \Eprint {https://arxiv.org/abs/1702.03688} {arXiv:1702.03688 [quant-ph]} \BibitemShut {NoStop}%
\bibitem [{\citenamefont {Erhard}\ \emph {et~al.}(2019)\citenamefont {Erhard} \emph {et~al.}}]{Erhard2019}%
  \BibitemOpen
  \bibfield  {author} {\bibinfo {author} {\bibfnamefont {A.}~\bibnamefont {Erhard}} \emph {et~al.},\ }\bibfield  {title} {\bibinfo {title} {Characterizing large-scale quantum computers via cycle benchmarking},\ }\href {https://doi.org/10.1038/s41467-019-13068-7} {\bibfield  {journal} {\bibinfo  {journal} {Nature Communications}\ }\textbf {\bibinfo {volume} {10}},\ \bibinfo {pages} {5347} (\bibinfo {year} {2019})}\BibitemShut {NoStop}%
\bibitem [{Note1()}]{Note1}%
  \BibitemOpen
  \bibinfo {note} {We benchmark the gate at $\theta = \pi /2$; in general Helios can implement either gate with arbitrary $\theta $ and we expect that logical error decreases with $\theta $ similar to the physical gate, since the probability of a single fault causing a logical error is correspondingly scaled down.}\BibitemShut {Stop}%
\bibitem [{\citenamefont {Moses}\ \emph {et~al.}(2023)\citenamefont {Moses} \emph {et~al.}}]{PhysRevX.13.041052}%
  \BibitemOpen
  \bibfield  {author} {\bibinfo {author} {\bibfnamefont {S.~A.}\ \bibnamefont {Moses}} \emph {et~al.},\ }\bibfield  {title} {\bibinfo {title} {A race-track trapped-ion quantum processor},\ }\href {https://doi.org/10.1103/PhysRevX.13.041052} {\bibfield  {journal} {\bibinfo  {journal} {Phys. Rev. X}\ }\textbf {\bibinfo {volume} {13}},\ \bibinfo {pages} {041052} (\bibinfo {year} {2023})}\BibitemShut {NoStop}%
\bibitem [{\citenamefont {Haghshenas}\ \emph {et~al.}(2025)\citenamefont {Haghshenas}, \citenamefont {Chertkov} \emph {et~al.}}]{floquet}%
  \BibitemOpen
  \bibfield  {author} {\bibinfo {author} {\bibfnamefont {R.}~\bibnamefont {Haghshenas}}, \bibinfo {author} {\bibfnamefont {E.}~\bibnamefont {Chertkov}}, \emph {et~al.},\ }\href {https://arxiv.org/abs/2503.20870} {\bibinfo {title} {Digital quantum magnetism at the frontier of classical simulations}} (\bibinfo {year} {2025}),\ \Eprint {https://arxiv.org/abs/2503.20870} {arXiv:2503.20870} \BibitemShut {NoStop}%
\bibitem [{\citenamefont {Proctor}\ \emph {et~al.}(2022{\natexlab{a}})\citenamefont {Proctor} \emph {et~al.}}]{Proctor22}%
  \BibitemOpen
  \bibfield  {author} {\bibinfo {author} {\bibfnamefont {T.}~\bibnamefont {Proctor}} \emph {et~al.},\ }\href@noop {} {\bibinfo {title} {Establishing trust in quantum computations}} (\bibinfo {year} {2022}{\natexlab{a}}),\ \Eprint {https://arxiv.org/abs/2204.07568} {arXiv:2204.07568 [quant-ph]} \BibitemShut {NoStop}%
\bibitem [{\citenamefont {Proctor}\ \emph {et~al.}(2022{\natexlab{b}})\citenamefont {Proctor}, \citenamefont {Rudinger}, \citenamefont {Young}, \citenamefont {Nielsen},\ and\ \citenamefont {Blume-Kohout}}]{proctor2022measuring}%
  \BibitemOpen
  \bibfield  {author} {\bibinfo {author} {\bibfnamefont {T.}~\bibnamefont {Proctor}}, \bibinfo {author} {\bibfnamefont {K.}~\bibnamefont {Rudinger}}, \bibinfo {author} {\bibfnamefont {K.}~\bibnamefont {Young}}, \bibinfo {author} {\bibfnamefont {E.}~\bibnamefont {Nielsen}},\ and\ \bibinfo {author} {\bibfnamefont {R.}~\bibnamefont {Blume-Kohout}},\ }\bibfield  {title} {\bibinfo {title} {Measuring the capabilities of quantum computers},\ }\href {https://doi.org/10.1038/s41567-021-01409-7} {\bibfield  {journal} {\bibinfo  {journal} {Nature Physics}\ }\textbf {\bibinfo {volume} {18}},\ \bibinfo {pages} {75} (\bibinfo {year} {2022}{\natexlab{b}})}\BibitemShut {NoStop}%
\bibitem [{\citenamefont {Mayer}\ \emph {et~al.}(2023)\citenamefont {Mayer} \emph {et~al.}}]{mayer2021theory}%
  \BibitemOpen
  \bibfield  {author} {\bibinfo {author} {\bibfnamefont {K.}~\bibnamefont {Mayer}} \emph {et~al.},\ }\href@noop {} {\bibinfo {title} {Theory of mirror benchmarking and demonstration on a quantum computer}} (\bibinfo {year} {2023}),\ \Eprint {https://arxiv.org/abs/2108.10431} {arXiv:2108.10431 [quant-ph]} \BibitemShut {NoStop}%
\bibitem [{\citenamefont {Koch}\ \emph {et~al.}()\citenamefont {Koch}, \citenamefont {Lawrence}, \citenamefont {Singhal}, \citenamefont {Sivarajah},\ and\ \citenamefont {Duncan}}]{guppy}%
  \BibitemOpen
  \bibfield  {author} {\bibinfo {author} {\bibfnamefont {M.}~\bibnamefont {Koch}}, \bibinfo {author} {\bibfnamefont {A.}~\bibnamefont {Lawrence}}, \bibinfo {author} {\bibfnamefont {K.}~\bibnamefont {Singhal}}, \bibinfo {author} {\bibfnamefont {S.}~\bibnamefont {Sivarajah}},\ and\ \bibinfo {author} {\bibfnamefont {R.}~\bibnamefont {Duncan}},\ }\bibfield  {title} {\bibinfo {title} {Guppy: Pythonic quantum-classical programming},\ }\Eprint {https://arxiv.org/abs/2510.12582} {arXiv:2510.12582} \BibitemShut {NoStop}%
\bibitem [{\citenamefont {Acharya}\ \emph {et~al.}(2025)\citenamefont {Acharya}, \citenamefont {Abanin} \emph {et~al.}}]{Acharya2025googlethreshold}%
  \BibitemOpen
  \bibfield  {author} {\bibinfo {author} {\bibfnamefont {R.}~\bibnamefont {Acharya}}, \bibinfo {author} {\bibfnamefont {D.~A.}\ \bibnamefont {Abanin}}, \emph {et~al.},\ }\bibfield  {title} {\bibinfo {title} {Quantum error correction below the surface code threshold},\ }\href {https://doi.org/10.1038/s41586-024-08449-y} {\bibfield  {journal} {\bibinfo  {journal} {Nature}\ }\textbf {\bibinfo {volume} {638}},\ \bibinfo {pages} {920} (\bibinfo {year} {2025})}\BibitemShut {NoStop}%
\bibitem [{\citenamefont {Forlivesi}\ and\ \citenamefont {Amaro}(2025)}]{Forlivesi:2025ilq}%
  \BibitemOpen
  \bibfield  {author} {\bibinfo {author} {\bibfnamefont {D.}~\bibnamefont {Forlivesi}}\ and\ \bibinfo {author} {\bibfnamefont {D.}~\bibnamefont {Amaro}},\ }\href@noop {} {\bibinfo {title} {Flag at origin: a modular fault-tolerant preparation for css codes}} (\bibinfo {year} {2025}),\ \Eprint {https://arxiv.org/abs/2508.14200} {arXiv:2508.14200 [quant-ph]} \BibitemShut {NoStop}%
\bibitem [{\citenamefont {Kakizaki}(2026)}]{kakizaki2026concatenateddecoding}%
  \BibitemOpen
  \bibfield  {author} {\bibinfo {author} {\bibfnamefont {T.}~\bibnamefont {Kakizaki}},\ }\href {https://arxiv.org/abs/2601.18743} {\bibinfo {title} {Approximate level-by-level maximum-likelihood decoding based on the chase algorithm for high-rate concatenated stabilizer codes}} (\bibinfo {year} {2026}),\ \Eprint {https://arxiv.org/abs/2601.18743} {arXiv:2601.18743 [quant-ph]} \BibitemShut {NoStop}%
\bibitem [{\citenamefont {Liu}\ \emph {et~al.}(2025{\natexlab{a}})\citenamefont {Liu} \emph {et~al.}}]{liu2025coniq}%
  \BibitemOpen
  \bibfield  {author} {\bibinfo {author} {\bibfnamefont {P.}~\bibnamefont {Liu}} \emph {et~al.},\ }\href {https://arxiv.org/abs/2508.05779} {\bibinfo {title} {Coniq: Enabling concatenated quantum error correction on neutral atom arrays}} (\bibinfo {year} {2025}{\natexlab{a}}),\ \Eprint {https://arxiv.org/abs/2508.05779} {arXiv:2508.05779 [cs.AR]} \BibitemShut {NoStop}%
\bibitem [{\citenamefont {Dasu}\ \emph {et~al.}(2025)\citenamefont {Dasu} \emph {et~al.}}]{Dasu:2025ktd}%
  \BibitemOpen
  \bibfield  {author} {\bibinfo {author} {\bibfnamefont {S.}~\bibnamefont {Dasu}} \emph {et~al.},\ }\href@noop {} {\bibinfo {title} {Breaking even with magic: demonstration of a high-fidelity logical non-clifford gate}} (\bibinfo {year} {2025}),\ \Eprint {https://arxiv.org/abs/2506.14688} {arXiv:2506.14688 [quant-ph]} \BibitemShut {NoStop}%
\bibitem [{\citenamefont {Sales~Rodriguez}\ \emph {et~al.}(2025)\citenamefont {Sales~Rodriguez}, \citenamefont {Robinson} \emph {et~al.}}]{SalesRodriguez2025queramagic}%
  \BibitemOpen
  \bibfield  {author} {\bibinfo {author} {\bibfnamefont {P.}~\bibnamefont {Sales~Rodriguez}}, \bibinfo {author} {\bibfnamefont {J.~M.}\ \bibnamefont {Robinson}}, \emph {et~al.},\ }\bibfield  {title} {\bibinfo {title} {Experimental demonstration of logical magic state distillation},\ }\href {https://doi.org/10.1038/s41586-025-09367-3} {\bibfield  {journal} {\bibinfo  {journal} {Nature}\ }\textbf {\bibinfo {volume} {645}},\ \bibinfo {pages} {620} (\bibinfo {year} {2025})}\BibitemShut {NoStop}%
\bibitem [{\citenamefont {Bluvstein}\ \emph {et~al.}(2024)\citenamefont {Bluvstein}, \citenamefont {Evered} \emph {et~al.}}]{Bluvstein2024queradevice}%
  \BibitemOpen
  \bibfield  {author} {\bibinfo {author} {\bibfnamefont {D.}~\bibnamefont {Bluvstein}}, \bibinfo {author} {\bibfnamefont {S.~J.}\ \bibnamefont {Evered}}, \emph {et~al.},\ }\bibfield  {title} {\bibinfo {title} {Logical quantum processor based on reconfigurable atom arrays},\ }\href {https://doi.org/10.1038/s41586-023-06927-3} {\bibfield  {journal} {\bibinfo  {journal} {Nature}\ }\textbf {\bibinfo {volume} {626}},\ \bibinfo {pages} {58} (\bibinfo {year} {2024})}\BibitemShut {NoStop}%
\bibitem [{\citenamefont {Wineland}\ \emph {et~al.}(1998)\citenamefont {Wineland} \emph {et~al.}}]{Wineland1998-xz}%
  \BibitemOpen
  \bibfield  {author} {\bibinfo {author} {\bibfnamefont {D.~J.}\ \bibnamefont {Wineland}} \emph {et~al.},\ }\bibfield  {title} {\bibinfo {title} {Experimental issues in coherent quantum-state manipulation of trapped atomic ions},\ }\href@noop {} {\bibfield  {journal} {\bibinfo  {journal} {J Res Natl Inst Stand Technol}\ }\textbf {\bibinfo {volume} {103}},\ \bibinfo {pages} {259} (\bibinfo {year} {1998})}\BibitemShut {NoStop}%
\bibitem [{\citenamefont {Kielpinski}\ \emph {et~al.}(2002)\citenamefont {Kielpinski}, \citenamefont {Monroe},\ and\ \citenamefont {Wineland}}]{Kielpinski2002}%
  \BibitemOpen
  \bibfield  {author} {\bibinfo {author} {\bibfnamefont {D.}~\bibnamefont {Kielpinski}}, \bibinfo {author} {\bibfnamefont {C.}~\bibnamefont {Monroe}},\ and\ \bibinfo {author} {\bibfnamefont {D.~J.}\ \bibnamefont {Wineland}},\ }\bibfield  {title} {\bibinfo {title} {Architecture for a large-scale ion-trap quantum computer},\ }\href {https://doi.org/10.1038/nature00784} {\bibfield  {journal} {\bibinfo  {journal} {Nature}\ }\textbf {\bibinfo {volume} {417}},\ \bibinfo {pages} {709} (\bibinfo {year} {2002})}\BibitemShut {NoStop}%
\bibitem [{\citenamefont {Pino}\ \emph {et~al.}(2021)\citenamefont {Pino} \emph {et~al.}}]{Pino2021}%
  \BibitemOpen
  \bibfield  {author} {\bibinfo {author} {\bibfnamefont {J.~M.}\ \bibnamefont {Pino}} \emph {et~al.},\ }\bibfield  {title} {\bibinfo {title} {Demonstration of the trapped-ion quantum ccd computer architecture},\ }\href {https://doi.org/10.1038/s41586-021-03318-4} {\bibfield  {journal} {\bibinfo  {journal} {Nature}\ }\textbf {\bibinfo {volume} {592}},\ \bibinfo {pages} {209} (\bibinfo {year} {2021})}\BibitemShut {NoStop}%
\bibitem [{\citenamefont {Burton}\ \emph {et~al.}(2023)\citenamefont {Burton} \emph {et~al.}}]{PhysRevLett.130.173202}%
  \BibitemOpen
  \bibfield  {author} {\bibinfo {author} {\bibfnamefont {W.~C.}\ \bibnamefont {Burton}} \emph {et~al.},\ }\bibfield  {title} {\bibinfo {title} {Transport of multispecies ion crystals through a junction in a radio-frequency paul trap},\ }\href {https://doi.org/10.1103/PhysRevLett.130.173202} {\bibfield  {journal} {\bibinfo  {journal} {Phys. Rev. Lett.}\ }\textbf {\bibinfo {volume} {130}},\ \bibinfo {pages} {173202} (\bibinfo {year} {2023})}\BibitemShut {NoStop}%
\bibitem [{\citenamefont {Niroula}\ \emph {et~al.}(2025)\citenamefont {Niroula} \emph {et~al.}}]{Niroula2025}%
  \BibitemOpen
  \bibfield  {author} {\bibinfo {author} {\bibfnamefont {P.}~\bibnamefont {Niroula}} \emph {et~al.},\ }\href@noop {} {\bibinfo {title} {Realization of a quantum streaming algorithm on long-lived trapped-ion qubits}} (\bibinfo {year} {2025}),\ \Eprint {https://arxiv.org/abs/2511.03689} {arXiv:2511.03689 [quant-ph]} \BibitemShut {NoStop}%
\bibitem [{\citenamefont {Liu}\ \emph {et~al.}(2025{\natexlab{b}})\citenamefont {Liu} \emph {et~al.}}]{Liu2025}%
  \BibitemOpen
  \bibfield  {author} {\bibinfo {author} {\bibfnamefont {M.}~\bibnamefont {Liu}} \emph {et~al.},\ }\href@noop {} {\bibinfo {title} {Certified randomness amplification by dynamically probing remote random quantum states}} (\bibinfo {year} {2025}{\natexlab{b}}),\ \Eprint {https://arxiv.org/abs/2511.03686} {arXiv:2511.03686 [quant-ph]} \BibitemShut {NoStop}%
\bibitem [{\citenamefont {Lee}\ \emph {et~al.}(2005)\citenamefont {Lee} \emph {et~al.}}]{Lee_2005}%
  \BibitemOpen
  \bibfield  {author} {\bibinfo {author} {\bibfnamefont {P.~J.}\ \bibnamefont {Lee}} \emph {et~al.},\ }\bibfield  {title} {\bibinfo {title} {Phase control of trapped ion quantum gates},\ }\href {https://doi.org/10.1088/1464-4266/7/10/025} {\bibfield  {journal} {\bibinfo  {journal} {Journal of Optics B: Quantum and Semiclassical Optics}\ }\textbf {\bibinfo {volume} {7}},\ \bibinfo {pages} {S371} (\bibinfo {year} {2005})}\BibitemShut {NoStop}%
\bibitem [{\citenamefont {Beale}\ \emph {et~al.}(2018)\citenamefont {Beale}, \citenamefont {Wallman}, \citenamefont {Guti\'errez}, \citenamefont {Brown},\ and\ \citenamefont {Laflamme}}]{Beale2018}%
  \BibitemOpen
  \bibfield  {author} {\bibinfo {author} {\bibfnamefont {S.~J.}\ \bibnamefont {Beale}}, \bibinfo {author} {\bibfnamefont {J.~J.}\ \bibnamefont {Wallman}}, \bibinfo {author} {\bibfnamefont {M.}~\bibnamefont {Guti\'errez}}, \bibinfo {author} {\bibfnamefont {K.~R.}\ \bibnamefont {Brown}},\ and\ \bibinfo {author} {\bibfnamefont {R.}~\bibnamefont {Laflamme}},\ }\bibfield  {title} {\bibinfo {title} {Quantum error correction decoheres noise},\ }\href {https://doi.org/10.1103/PhysRevLett.121.190501} {\bibfield  {journal} {\bibinfo  {journal} {Phys. Rev. Lett.}\ }\textbf {\bibinfo {volume} {121}},\ \bibinfo {pages} {190501} (\bibinfo {year} {2018})}\BibitemShut {NoStop}%
\bibitem [{\citenamefont {Ryan-Anderson}\ \emph {et~al.}(2021)\citenamefont {Ryan-Anderson} \emph {et~al.}}]{RyanAnderson2021}%
  \BibitemOpen
  \bibfield  {author} {\bibinfo {author} {\bibfnamefont {C.}~\bibnamefont {Ryan-Anderson}} \emph {et~al.},\ }\href {https://arxiv.org/abs/2107.07505} {\bibinfo {title} {Realization of real-time fault-tolerant quantum error correction}} (\bibinfo {year} {2021}),\ \Eprint {https://arxiv.org/abs/2107.07505} {arXiv:2107.07505 [quant-ph]} \BibitemShut {NoStop}%
\bibitem [{\citenamefont {Brown}\ \emph {et~al.}(2020)\citenamefont {Brown}, \citenamefont {Cross},\ and\ \citenamefont {Brown}}]{brown2020}%
  \BibitemOpen
  \bibfield  {author} {\bibinfo {author} {\bibfnamefont {N.~C.}\ \bibnamefont {Brown}}, \bibinfo {author} {\bibfnamefont {A.~W.}\ \bibnamefont {Cross}},\ and\ \bibinfo {author} {\bibfnamefont {K.~R.}\ \bibnamefont {Brown}},\ }\href {https://arxiv.org/abs/2003.05843} {\bibinfo {title} {Critical faults of leakage errors on the surface code}} (\bibinfo {year} {2020}),\ \Eprint {https://arxiv.org/abs/2003.05843} {arXiv:2003.05843 [quant-ph]} \BibitemShut {NoStop}%
\bibitem [{\citenamefont {Hayes}\ \emph {et~al.}(2020)\citenamefont {Hayes} \emph {et~al.}}]{Hayes_2020}%
  \BibitemOpen
  \bibfield  {author} {\bibinfo {author} {\bibfnamefont {D.}~\bibnamefont {Hayes}} \emph {et~al.},\ }\bibfield  {title} {\bibinfo {title} {Eliminating leakage errors in hyperfine qubits},\ }\bibfield  {journal} {\bibinfo  {journal} {Physical Review Letters}\ }\textbf {\bibinfo {volume} {124}},\ \href {https://doi.org/10.1103/physrevlett.124.170501} {10.1103/physrevlett.124.170501} (\bibinfo {year} {2020})\BibitemShut {NoStop}%
\bibitem [{\citenamefont {Javadi-Abhari}\ \emph {et~al.}(2025)\citenamefont {Javadi-Abhari}, \citenamefont {Martiel}, \citenamefont {Seif}, \citenamefont {Takita},\ and\ \citenamefont {Wei}}]{Wei2025}%
  \BibitemOpen
  \bibfield  {author} {\bibinfo {author} {\bibfnamefont {A.}~\bibnamefont {Javadi-Abhari}}, \bibinfo {author} {\bibfnamefont {S.}~\bibnamefont {Martiel}}, \bibinfo {author} {\bibfnamefont {A.}~\bibnamefont {Seif}}, \bibinfo {author} {\bibfnamefont {M.}~\bibnamefont {Takita}},\ and\ \bibinfo {author} {\bibfnamefont {K.~X.}\ \bibnamefont {Wei}},\ }\href@noop {} {\bibinfo {title} {Big cats: entanglement in 120 qubits and beyond}} (\bibinfo {year} {2025}),\ \Eprint {https://arxiv.org/abs/2510.09520} {arXiv:2510.09520 [quant-ph]} \BibitemShut {NoStop}%
\bibitem [{\citenamefont {Lidar}\ \emph {et~al.}(1998)\citenamefont {Lidar}, \citenamefont {Chuang},\ and\ \citenamefont {Whaley}}]{PhysRevLett.81.2594}%
  \BibitemOpen
  \bibfield  {author} {\bibinfo {author} {\bibfnamefont {D.~A.}\ \bibnamefont {Lidar}}, \bibinfo {author} {\bibfnamefont {I.~L.}\ \bibnamefont {Chuang}},\ and\ \bibinfo {author} {\bibfnamefont {K.~B.}\ \bibnamefont {Whaley}},\ }\bibfield  {title} {\bibinfo {title} {Decoherence-free subspaces for quantum computation},\ }\href {https://doi.org/10.1103/PhysRevLett.81.2594} {\bibfield  {journal} {\bibinfo  {journal} {Phys. Rev. Lett.}\ }\textbf {\bibinfo {volume} {81}},\ \bibinfo {pages} {2594} (\bibinfo {year} {1998})}\BibitemShut {NoStop}%
\bibitem [{\citenamefont {Steane}(1997)}]{Steane_1997}%
  \BibitemOpen
  \bibfield  {author} {\bibinfo {author} {\bibfnamefont {A.~M.}\ \bibnamefont {Steane}},\ }\bibfield  {title} {\bibinfo {title} {Active stabilization, quantum computation, and quantum state synthesis},\ }\href {https://doi.org/10.1103/physrevlett.78.2252} {\bibfield  {journal} {\bibinfo  {journal} {Physical Review Letters}\ }\textbf {\bibinfo {volume} {78}},\ \bibinfo {pages} {2252–2255} (\bibinfo {year} {1997})}\BibitemShut {NoStop}%
\bibitem [{\citenamefont {Prabhu}\ and\ \citenamefont {Reichardt}(2024)}]{Prabhu_2024}%
  \BibitemOpen
  \bibfield  {author} {\bibinfo {author} {\bibfnamefont {P.}~\bibnamefont {Prabhu}}\ and\ \bibinfo {author} {\bibfnamefont {B.~W.}\ \bibnamefont {Reichardt}},\ }\bibfield  {title} {\bibinfo {title} {Distance-four quantum codes with combined postselection and error correction},\ }\bibfield  {journal} {\bibinfo  {journal} {Physical Review A}\ }\textbf {\bibinfo {volume} {110}},\ \href {https://doi.org/10.1103/physreva.110.012419} {10.1103/physreva.110.012419} (\bibinfo {year} {2024})\BibitemShut {NoStop}%
\bibitem [{\citenamefont {Gidney}(2021)}]{Gidney:2021wm}%
  \BibitemOpen
  \bibfield  {author} {\bibinfo {author} {\bibfnamefont {C.}~\bibnamefont {Gidney}},\ }\bibfield  {title} {\bibinfo {title} {Stim: a fast stabilizer circuit simulator},\ }\href {https://doi.org/10.22331/q-2021-07-06-497} {\bibfield  {journal} {\bibinfo  {journal} {{Quantum}}\ }\textbf {\bibinfo {volume} {5}},\ \bibinfo {pages} {497} (\bibinfo {year} {2021})}\BibitemShut {NoStop}%
\bibitem [{\citenamefont {Nielsen}(2002)}]{Nielsen2002}%
  \BibitemOpen
  \bibfield  {author} {\bibinfo {author} {\bibfnamefont {M.~A.}\ \bibnamefont {Nielsen}},\ }\bibfield  {title} {\bibinfo {title} {A simple formula for the average gate fidelity of a quantum dynamical operation},\ }\href {https://doi.org/https://doi.org/10.1016/S0375-9601(02)01272-0} {\bibfield  {journal} {\bibinfo  {journal} {Physics Letters A}\ }\textbf {\bibinfo {volume} {303}},\ \bibinfo {pages} {249} (\bibinfo {year} {2002})}\BibitemShut {NoStop}%
\bibitem [{\citenamefont {Li}\ \emph {et~al.}(2020)\citenamefont {Li}, \citenamefont {Han},\ and\ \citenamefont {Zhu}}]{Li2020}%
  \BibitemOpen
  \bibfield  {author} {\bibinfo {author} {\bibfnamefont {Z.}~\bibnamefont {Li}}, \bibinfo {author} {\bibfnamefont {Y.-G.}\ \bibnamefont {Han}},\ and\ \bibinfo {author} {\bibfnamefont {H.}~\bibnamefont {Zhu}},\ }\bibfield  {title} {\bibinfo {title} {Optimal verification of greenberger-horne-zeilinger states},\ }\href {https://doi.org/10.1103/PhysRevApplied.13.054002} {\bibfield  {journal} {\bibinfo  {journal} {Phys. Rev. Appl.}\ }\textbf {\bibinfo {volume} {13}},\ \bibinfo {pages} {054002} (\bibinfo {year} {2020})}\BibitemShut {NoStop}%
\bibitem [{\citenamefont {Flammia}\ and\ \citenamefont {Liu}(2011)}]{Flammia2011}%
  \BibitemOpen
  \bibfield  {author} {\bibinfo {author} {\bibfnamefont {S.~T.}\ \bibnamefont {Flammia}}\ and\ \bibinfo {author} {\bibfnamefont {Y.-K.}\ \bibnamefont {Liu}},\ }\bibfield  {title} {\bibinfo {title} {Direct fidelity estimation from few pauli measurements},\ }\href {https://doi.org/10.1103/PhysRevLett.106.230501} {\bibfield  {journal} {\bibinfo  {journal} {Phys. Rev. Lett.}\ }\textbf {\bibinfo {volume} {106}},\ \bibinfo {pages} {230501} (\bibinfo {year} {2011})}\BibitemShut {NoStop}%
\bibitem [{\citenamefont {Hong}\ \emph {et~al.}(2024)\citenamefont {Hong}, \citenamefont {Durso-Sabina}, \citenamefont {Hayes},\ and\ \citenamefont {Lucas}}]{Hong_2024}%
  \BibitemOpen
  \bibfield  {author} {\bibinfo {author} {\bibfnamefont {Y.}~\bibnamefont {Hong}}, \bibinfo {author} {\bibfnamefont {E.}~\bibnamefont {Durso-Sabina}}, \bibinfo {author} {\bibfnamefont {D.}~\bibnamefont {Hayes}},\ and\ \bibinfo {author} {\bibfnamefont {A.}~\bibnamefont {Lucas}},\ }\bibfield  {title} {\bibinfo {title} {Entangling four logical qubits beyond break-even in a nonlocal code},\ }\bibfield  {journal} {\bibinfo  {journal} {Physical Review Letters}\ }\textbf {\bibinfo {volume} {133}},\ \href {https://doi.org/10.1103/physrevlett.133.180601} {10.1103/physrevlett.133.180601} (\bibinfo {year} {2024})\BibitemShut {NoStop}%
\bibitem [{\citenamefont {Criger}\ \emph {et~al.}(2026)\citenamefont {Criger} \emph {et~al.}}]{CrigerInPreparation}%
  \BibitemOpen
  \bibfield  {author} {\bibinfo {author} {\bibfnamefont {B.}~\bibnamefont {Criger}} \emph {et~al.},\ }\bibfield  {title} {\bibinfo {title} {Automated flag-based fault-tolerant state preparation using integer linear programming}} (\bibinfo {year} {2026}),\ \bibinfo {note} {in preparation}\BibitemShut {NoStop}%
\end{thebibliography}%

\end{document}